\documentclass[11pt]{amsart}
\usepackage[top=1.0in, bottom=1.0in, left=1.0in, right=1.0in]{geometry}

\usepackage{amsfonts, amssymb, amsmath, enumerate, mathtools}
\usepackage{amsthm}
\usepackage{hyperref}
\usepackage[foot]{amsaddr}
\usepackage{bbm}
\usepackage{xcolor}
\usepackage[style=nature]{biblatex}
\usepackage{wasysym}
\hypersetup{
	colorlinks,
	linkcolor={red!40!gray},
	citecolor={blue!40!gray},
	urlcolor={blue!70!gray}
}

\numberwithin{equation}{section}
\usepackage{amsfonts, amssymb, amsmath, enumerate}

\usepackage{mathrsfs}

\usepackage{mathtools}
\usepackage[normalem]{ulem}
\usepackage{comment}

\usepackage{bbm}
\usepackage{algorithm}
\usepackage{algorithmic}
\usepackage{tikz}

\usetikzlibrary{shapes.geometric, arrows.meta, positioning,calc}

\tikzstyle{roundedbox} = [
    rectangle, rounded corners,
    minimum width=3.5cm, minimum height=1cm,
    text centered, draw=black, fill=blue!20,
    text width=10cm, align=center
]
\tikzstyle{arrow} = [thick, ->, >=Stealth]

\usepackage{pgfplots}
\pgfplotsset{compat=newest}

\numberwithin{equation}{section}

\DeclareUnicodeCharacter{2113}{l}

\definecolor{darkgreen}{cmyk}{0.6,0,0.8,0}

\DeclareMathOperator{\E}{\mathbb{E}}
\DeclareMathOperator{\R}{\mathbb{R}}
\DeclareMathOperator{\N}{\mathbb{N}}
\DeclareMathOperator{\Z}{\mathbb{Z}}

\DeclareMathOperator{\bS}{\mathbb{S}}
\DeclareMathOperator{\Pb}{\mathbb{P}}

\DeclareMathOperator*{\diag}{diag}

\DeclareMathOperator*{\tr}{tr}
\DeclareMathOperator*{\Tr}{Tr}

\DeclareMathOperator{\poly}{poly}

\DeclareMathOperator{\nnz}{nnz}
\DeclareMathOperator{\sym}{sym}
\DeclareMathOperator{\cov}{Cov}

\DeclareMathOperator{\spec}{spec}

\def \etc {,\ldots,}

\newcommand{\cZ}{\mathcal{Z}}

\DeclarePairedDelimiter{\norm}{\lVert}{\rVert}
\newcommand{\norml}[1]{\left\lVert#1\right\rVert}
\DeclarePairedDelimiter{\abs}{\lvert}{\rvert}
\DeclarePairedDelimiter{\ip}{\langle}{\rangle}
\DeclarePairedDelimiter{\paren}{(}{)}
\DeclarePairedDelimiter{\sqbr}{[}{]}

\newtheorem{theorem}{Theorem}[section]
\newtheorem{proposition}[theorem]{Proposition}
\newtheorem{corollary}[theorem]{Corollary}
\newtheorem{lemma}[theorem]{Lemma}
\newtheorem{conjecture}[theorem]{Conjecture}

\newtheorem{remark}[theorem]{Remark}

\newtheorem{definition}[theorem]{Definition}

\begin{document}

\title[Optimal Subspace Embeddings]{
Optimal Subspace Embeddings: Resolving Nelson-Nguyen Conjecture Up to Sub-Polylogarithmic Factors
}
\date{}
\author{Shabarish Chenakkod}
\author{Micha{\l} Derezi\'nski}
\author{Xiaoyu Dong}
\thanks{Partially supported by DMS 2054408 and CCF 2338655. The authors are very grateful for the generous help and support of Mark Rudelson throughout the duration of this work.}
\address{University of Michigan, Ann Arbor, MI, USA}
\email{shabari@umich.edu, derezin@umich.edu}
\address{National University of Singapore, Singapore}
\email{xdong@nus.edu.sg}

\begin{abstract}
We give a proof of the conjecture of Nelson and Nguyen [FOCS 2013] on the optimal dimension and sparsity of oblivious subspace embeddings, up to sub-polylogarithmic factors: For any $n\geq d$ and $\epsilon\geq d^{-O(1)}$, there is a random $\tilde O(d/\epsilon^2)\times n$ matrix $\Pi$ with $\tilde O(\log(d)/\epsilon)$ non-zeros per column such that for any $A\in\mathbb{R}^{n\times d}$, with high probability, $(1-\epsilon)\|Ax\|\leq\|\Pi Ax\|\leq(1+\epsilon)\|Ax\|$ for all $x\in\mathbb{R}^d$, where $\tilde O(\cdot)$ hides only sub-polylogarithmic factors in $d$. Our result in particular implies a new fastest sub-current matrix multiplication time reduction of size $\tilde O(d/\epsilon^2)$ for a broad class of $n\times d$ linear regression tasks.

A key novelty in our analysis is a matrix concentration technique we call \emph{iterative decoupling}, which we use to fine-tune the higher-order trace moment bounds attainable via existing random matrix universality tools [Brailovskaya and van Handel, GAFA 2024]. 

\end{abstract}

\maketitle

\thispagestyle{empty}
\newpage
\setcounter{page}{1}

\section{Introduction.}

Dimensionality reduction techniques are some of the most powerful tools for designing efficient algorithms in computational linear algebra, numerical optimization, and related areas. A fundamental primitive in this context is matrix approximation, where given a large matrix $A$, we seek to efficiently construct a smaller matrix $\tilde A$ that approximately preserves the linear algebraic properties of $A$ such as its rank, norm, singular values, etc. Most of these properties can be naturally encapsulated by a single one, the \emph{subspace embedding} property, which requires that $\|\tilde Ax\|\approx_\varepsilon\|Ax\|$ for all vectors $x$ of appropriate dimension and a given $\varepsilon\in(0,1)$. Here, we use $a\approx_\varepsilon b$ to denote $(1-\varepsilon)b\leq a\leq (1+\varepsilon)b$. Subspace embeddings have become a workhorse of computational linear algebra, leading to fast algorithms for linear regression, low-rank approximation, clustering, linear programming, and more (see \cite{woodruff2014sketching,martinsson2020randomized,randlapack_book,derezinski2024recent} for comprehensive overviews).

A key advantage of subspace embeddings is that they can be both simple and efficient: an approximation $\tilde A$ obtained by multiplying $A$ with a sparse random matrix will, under appropriate conditions on sparsity and dimension, with high probability attain the subspace embedding property in nearly linear time. Such sparse random matrices are often called \emph{oblivious}, since their distribution need not depend on the input $A$. The first such distribution was proposed by Clarkson and Woodruff \cite{clarkson2013low}. Concretely, given $n\geq d\geq 1$, a matrix $\Pi\in \R^{m\times n}$ with exactly one uniformly placed non-zero random sign entry in each column (CountSketch) is an oblivious subspace embedding (OSE), i.e., for every $A\in\R^{n\times d}$, $\tilde A=\Pi A$ satisfies the subspace embedding property with good probability, as long as $m\geq O(d^2/\varepsilon^2)$ (see \cite{nelson2013osnap} for a simple proof). The sparsity pattern of the CountSketch implies that $\Pi A$ can be computed in $O(\nnz(A))$ time, i.e., proportional to the number of non-zero entries in $A$.
While this running time is the best one could hope for, the embedding dimension $m$ incurs an undesirable quadratic dependence on $d$. Shortly after this first result, Nelson and Nguyen \cite{nelson2013osnap} conjectured that a linear dependence on $d$ is possible without sacrificing too much in sparsity. 
\begin{conjecture}[Nelson and Nguyen, FOCS 2013 \cite{nelson2013osnap}]\label{c:nn}
    For any $n\geq d$ and $\varepsilon\in(0,1)$, there is a random  $\Pi\in\R^{m\times n}$ with $O(\log(d)/\varepsilon)$ non-zero entries per column and dimension $m=O(d/\varepsilon^2)$ such that for any $A\in\R^{n\times d}$, with probability $1-d^{-O(1)}$, $\|\Pi Ax\|\approx_\varepsilon\|Ax\|$ for all $x\in\R^d$.
\end{conjecture}
Nelson and Nguyen matched the conjecture up to polylogarithmic factors, giving an embedding $\Pi$ with $O(\log^3(d)/\varepsilon)$ non-zeros per column and dimension $m=O(d\log^8(d)/\varepsilon^2)$. Further works improved on the polylogarithmic factors in sparsity and dimension \cite{bourgain2015toward,cohen2016nearly,chenakkod2024optimal,chenakkod2025optimal} (see Table \ref{tab:comparison}), and gave lower bounds supporting the conjecture \cite{nelson2014lower,li2022lower,li2022sparsity}. Notably, Cohen \cite{cohen2016nearly} matched the conjecture in terms of sparsity while attaining sub-optimal dimension, $m=O(d\log(d)/\varepsilon^2)$, by adapting a variant of the matrix Chernoff bound \cite{troppmatrixconc}. Later, Chenakkod, Derezi\'nski, Dong, and Rudelson \cite{chenakkod2024optimal} matched the conjecture in terms of the dimension while attaining sub-optimal sparsity, $O(\log^4(d)/\varepsilon^6)$ non-zeros per column, by leveraging recent progress in matrix concentration via Gaussian universality of Brailovskaya and van Handel \cite{brailovskaya2022universality}. This latter approach was further refined by Chenakkod, Derezi\'nski, and Dong \cite{chenakkod2025optimal} to attain $O(\log^2(d)/\varepsilon+\log^3(d))$ sparsity. However, the original conjecture has remained a tantalizing open question that tests the fundamental limits of current matrix concentration tools used widely across mathematics and computer science.
\begin{table}\centering
\caption{Progress towards resolving the Nelson-Nguyen conjecture (Conjecture \ref{c:nn}, \cite{nelson2013osnap}). For simplicity, we hide the big-O notation and sub-polylogarithmic factors. We note that the initial CountSketch guarantee given by \cite{clarkson2013low} was worse. We state the optimized guarantee given by \cite{nelson2013osnap,meng2013low}.}\vspace{2mm}
\begin{tabular}{l||c|c}
     & Dimension $m$ & Sparsity (nnz per column) \\
    \hline\hline
    \textbf{Conjecture \ref{c:nn}} \cite{nelson2013osnap} & $d/\varepsilon^2$ & $\log(d)/\varepsilon$
    \\
    \hline
    Clarkson \& Woodruff \cite{clarkson2013low} (and \cite{nelson2013osnap,meng2013low}) &$d^2/\varepsilon^2$ & $1$ \\
Nelson \& Nguyen \cite{nelson2013osnap} & $d\log^8(d)/\varepsilon^2$ & $\log^3(d)/\varepsilon$ \\
Bourgain et al. \cite{bourgain2015toward}&  $d\log^2(d)/\varepsilon^2$ & $\log^4(d)/\varepsilon^2$\\
Cohen \cite{cohen2016nearly}& $d\log(d)/\varepsilon^2$ &
$\log(d)/\varepsilon$\\
Chenakkod et al. \cite{chenakkod2024optimal}& $d/\varepsilon^2$
& $\log^4(d)/\varepsilon^6$
                \\
Chenakkod et al. \cite{chenakkod2025optimal}& $d/\varepsilon^2$ & $\log^2(d)/\varepsilon +\log^3(d)$\\
\hline
\textbf{Our result (Thm.\ \ref{t:main})} & $d/\varepsilon^2$ & $\log(d)/\varepsilon$
\end{tabular}
\label{tab:comparison}
\end{table}

\subsection{Our Results.}

The main result of this paper is the following theorem, which resolves the conjecture of Nelson and Nguyen up to sub-polylogarithmic factors in $d$ when $\varepsilon$ is inverse polynomially bounded in $d$.

\begin{theorem}[Nelson-Nguyen up to sub-polylogarithmic factors]\label{t:main}
    For any $n\geq d$ and $\varepsilon \geq d^{-O(1)}$, there is a random $\Pi\in\R^{m\times n}$ with $\tilde O(\log(d)/\varepsilon)$\footnote{Here and throughout, the notation $\tilde O(\cdot)$ hides only sub-polylogarithmic factors in $d$, i.e., $\log^{o(1)}(d)$.} non-zeros per column and dimension $m=\tilde O(d/\varepsilon^2)$ such that for any $A\in\R^{n\times d}$, with probability $1-d^{-O(1)}$, $\|\Pi Ax\|\approx_\varepsilon\|Ax\|$ for all $x\in\R^d$. 
\end{theorem}

A more detailed version of our result is provided in Theorem \ref{thm:osedecoup}, which gives a parameterized family of subspace embedding guarantees that offer different trade-offs between the dimension and the sparsity of $\Pi$, including one that exactly attains the optimal dimension but with a slightly worse sparsity bound. Carefully balancing those trade-offs, we recover Theorem \ref{t:main}.

A key novelty in our analysis is a matrix concentration technique we call \emph{iterative decoupling}, which provides fine-grained bounds on higher-order trace moments of random matrices. In this approach, we recursively reduce the order of the matrix moment we wish to bound by decoupling the matrix power into a product of independent copies. While decoupling is a ubiquitous technique in probability theory, it tends to result in an exponential blow up when applied to a product of sufficiently many components. However, we show that when deployed recursively through a careful schedule, it can lead to improved bounds for the higher-order trace moments that naturally arise in matrix concentration analysis. In order to deploy this technique, we further develop several higher-order moment bounds for norm quantities associated with subspace embeddings, which should be of independent interest.

\paragraph{\textbf{Implications for Linear Regression.}}
As discussed, subspace embeddings have numerous algorithmic implications in linear algebra and optimization (see, e.g., \cite{woodruff2014sketching,martinsson2020randomized,randlapack_book} for detailed overviews). The most immediate one arises in the context of linear regression, where one often deals with minimizing a generalized least squares objective of the form $\|Ax-b\|^2+g(x)$ over some domain $\mathcal C\subseteq\R^d$, where $A\in\R^{n\times d}$, $ b\in\R^d$, $g(x)$ is, e.g., the scaled $l_1$-norm (Lasso regression, \cite{tibshirani1996regression}), and $\mathcal C$ is, e.g., the $l_p$-norm ball \cite{raskutti2011minimax}, non-negativity constraints \cite{boutsidis2009random}, or $\R^d$ (see \cite{pilanci2015randomized} for further discussion).

When $n\gg d$, we can reduce this optimization problem to solving a much smaller instance by computing  $\tilde A=SA$ and  $\tilde b=Sb$ for $S\in\R^{m\times n}$. If $S$ is a subspace embedding for the extended data matrix $[A\mid b]\in\R^{n\times (d+1)}$, then this smaller instance retains an $\varepsilon$-approximation of the objective, since $\|\tilde Ax-\tilde b\|\approx_\varepsilon\|Ax-b\|$ for all $x$. It follows that an $\varepsilon$-approximate solution to the sketched problem yields an $O(\varepsilon)$-approximate solution to the original problem (e.g., see \cite{chenakkod2024optimal}). 

Applying Theorem \ref{t:main} together with the above argument, we obtain the fastest sub-current matrix multiplication time reduction of size $\tilde O(d/\varepsilon^2)$ for this class of linear regression objectives.

\begin{corollary}[Fast reduction for linear regression]
    Given $A\in\R^{n\times d}$, $b\in\R^n$, $\varepsilon\in(0,1)$, constraint set $\mathcal C\subseteq\R^d$, and regularizer $g:\mathcal{C}\rightarrow \R_{\geq 0}$, let $f(x) = \|Ax-b\|^2+g(x)$ for $x\in\mathcal C$. We can reduce the task of finding $\hat x\in\mathcal C$ such that $f(\hat x)\le (1+\varepsilon)\min_{x\in\mathcal C}f(x)$ to an instance $\tilde A\in\R^{m\times d},\tilde b\in\R^m, \tilde\varepsilon=\Theta(\varepsilon)$ of the same task where $m = \tilde O(d/\varepsilon^2)$ in time $\tilde O(\nnz(A)\log(d)/\varepsilon)$ with probability $1-d^{-O(1)}$.
\end{corollary}
In addition to the above listed prior works on oblivious subspace embeddings, we can compare this guarantee to other approaches, which use information about the matrix $A$ to optimize the running time, via techniques such as leverage score sampling \cite{chepurko2022near,cherapanamjeri2023optimal,chenakkod2024optimal,chenakkod2025optimal}. Unfortunately, all of these reductions incur additional $\poly(d/\varepsilon)$ costs, including the $d^\omega$ cost of matrix multiplication (here, $\omega\approx 2.372$ denotes the current matrix multiplication exponent \cite{alman2025more}). The best known $\tilde O(d/\varepsilon^2)$ reduction of this type attains $O(\nnz(A)\log(d) + d^\omega + d^2\log^3(d)/\varepsilon)$ running time \cite{chenakkod2025optimal}. 

\subsection{Further Related Work.}

A number of subspace embedding techniques have been considered in addition to sparse random matrices. In particular, \cite{sarlos2006improved} was the first to propose the concept of a subspace embedding, and showed that it can be attained using the fast Johnson-Lindenstrauss transform introduced by \cite{ailon2009fast}, which is closely related to the subsampled randomized Hadamard transform \cite{tropp2011improved}. For example, in the latter, we first apply an orthogonal transformation that randomly mixes the rows of the input $A\in\R^{n\times d}$, and then we select a uniformly random sample of $m$ rows from the transformed matrix. A yet different approach by \cite{drineas2006sampling} uses i.i.d.\ importance sampling to select a size $m$ row sample from the original input, without the transformation step. Here, the importance weights are chosen based on the approximate leverage scores of the rows \cite{drineas2012fast,li2013iterative}. Interestingly, unlike sparse random matrices, the i.i.d\ row sampling approaches cannot attain linear dependence of $m$ on $d$ due to the coupon collector problem (they require at least $m\geq d\log d$). However, \cite{less-embeddings} showed that leverage score sampling can be incorporated into sparse random matrices, which has led to fast algorithms for subspace embeddings in certain regimes \cite{chenakkod2025optimal}, as mentioned earlier.

Finally, there have been works analyzing sparse random matrices as embeddings with somewhat different and incomparable claims \cite{cartis2021hashing,tropp2025comparison,less-embeddings,gaussianization}. Most notably, \cite{cartis2021hashing} obtained a linear in $d$ sparse subspace embedding under additional assumptions on the aspect ratio and leverage score distribution of the input, whereas \cite{tropp2025comparison} showed that an $O(d/\varepsilon^2)$-dimensional embedding with sparsity $O(\log(d)/\varepsilon^2)$ satisfies the lower $\varepsilon$-approximation, i.e., $(1-\varepsilon)\|Ax\|\leq\|\Pi Ax\|$.

\section{Main Results and Techniques.}

In this section, we give detailed versions of our main results (Section~\ref{s:main-result}) and provide an overview of our proof techniques (Section~\ref{sec:overview}).

\subsection{Main Result.}
\label{s:main-result}

The matrices $\Pi$ that we work with are part of the well studied family of subspace embeddings termed OSNAP (\cite{nelson2013osnap, chenakkod2025optimal}). Given dimensions $m$ and $n$, and a sparsity parameter $p\in[0,1]$, these embeddings are defined as $m \times n$ random matrices formed by stacking $pm$ many $(1/p) \times n$ CountSketch matrices vertically followed by scaling with $1/\sqrt{pm}$, leading to a construction with $pm$ non-zero entries per column (where we assume that $pm$ divides $m$). In other words, each column of $\Pi$ is a stack of $pm$ many subcolumns each of length $1/p$ with exactly one $\pm 1/\sqrt{pm}$ entry in a random position in each subcolumn. We give a formal definition in Section \ref{subsec:osnapprops}. 

Following usual practice (\cite{nelson2013osnap, cohen2016nearly, chenakkod2024optimal, chenakkod2025optimal}), we frame our subspace embedding result in terms of bounding the extreme singular values of matrices of the form $\Pi U$ for $n \times d$ matrices $U$ with orthonormal columns. In particular, for given $\varepsilon, \delta > 0$, we show, 
\begin{align}
\Pb(1-\varepsilon\leq s_{\min}(\Pi U)\leq s_{\max}(\Pi U)\leq 1+\varepsilon)\geq 1-\delta,\label{eq:ose-equiv}
 \end{align}
 for any $n \times d$ matrices $U$ with orthonormal columns when dimension $m$ and sparsity $pm$ are above specified thresholds. Here, $s_{\min}$ and $s_{\max}$ denote the smallest and largest singular values.

 We now state our parameterized family of theorems that establish \eqref{eq:ose-equiv} for different values of $m$ and $pm$, explaining the trade-off between dimension and sparsity that arises from these guarantees.
 
 \begin{theorem}[High Probability Bounds for the Embedding Error of OSNAP]\label{thm:osedecoup}
Let $\Pi$ be an $m \times n$ matrix distributed according to the OSNAP distribution with parameter $p$ as in Definition \ref{def:osnap} and let $U$ be an arbitrary $n \times d$ deterministic matrix such that $U^TU=I$. Let $0 < \delta, \varepsilon < 1$ and $d>10$. Then, there exist constants $c_{\ref{thm:osedecoup}.1}, c_{\ref{thm:osedecoup}.2}, c_{\ref{thm:osedecoup}.3} >0$ such that when,
\begin{itemize}
    \item $k$ is an arbitrary positive integer,
    \item $\theta$ is a positive parameter such that $\theta \ge c_{\ref{thm:osedecoup}.1} \exp(\exp(c_{\ref{thm:osedecoup}.2}k))^2$,
\end{itemize}
we have,
\begin{equation}\label{osepro}
    \begin{aligned}
\Pb \left( 1 - \varepsilon  \leq s_{\min}(\Pi U)   \leq s_{\max}(\Pi U) \leq 1 + \varepsilon \right) \geq 1-\delta
\end{aligned}
\end{equation}
when $m=\theta\cdot(d+\log(d/\delta))/\varepsilon^2$ and the number of non-zeros per column of $\Pi$ satisfies:
\begin{align}
    \begin{aligned} \label{eq:sparsitylb}
    {pm } \ge {\exp(\exp(c_{\ref{thm:osedecoup}.3}k))}\left(\frac{1}{\varepsilon^{1+\frac{1}{\log(d/\delta)}}}\Big(\theta \log(d/\delta)+\frac{\log(d/\delta)^{5/2}}{\theta ^{k/2-1/4}}\Big)+\frac{ \log(d/\delta)^{4}}{\theta ^{k+1/2}}\right)
\end{aligned}
\end{align}

\end{theorem}

Observe that the lower bounds on $m$ and $pm$ are parameterized by two different quantities - a positive integer $k$, and a parameter $\theta$. Roughly speaking, $k$ describes the number of iterations used in our iterative decoupling technique, whereas $\theta$ represents the multiplicative slack that we allow in the embedding dimension $m$. By varying the choice of $k$ and $\theta$, we are able to obtain subspace embedding guarantees with different dimension and sparsity bounds.

In the simplest case, choosing $k=1$ and $\theta = O(1)$, which effectively corresponds to no iterative decoupling, we obtain the following result. (Here $k=1$ technically means we do one step of iterative decoupling but this is just for the convenience of the proof. Under the choice $\theta = O(1)$, choosing $k=1$ and $k=0$ give the same embedding dimension and sparsity requirements.)

\begin{corollary}\label{cor:osebasic}
    An $m \times n$ matrix $\Pi$ with the OSNAP distribution satisfies \eqref{osepro} when 
    \[ m=\frac{c_{\ref{cor:osebasic}.1}(d+\log(d/\delta))}{\varepsilon^2} \text{ and } 
    pm \ge c_{\ref{cor:osebasic}.2}\left(\frac{\log(d/\delta)^{5/2}}{\varepsilon^{1+\frac{1}{\log(d/\delta)}}}+\log(d/\delta)^4\right)
    \]
\end{corollary}

Note that when $\varepsilon \geq d^{-O(1)}$, then $(1/\varepsilon)^{1+\frac1{\log(d/\delta)}} = O(1/\varepsilon)$.
This recovers the optimal embedding dimension exactly while attaining correct sparsity up to polylogarithmic factors, thus obtaining a claim similar to that of \cite[Theorem 7]{chenakkod2025optimal}, but we give a much simpler proof of this result, as discussed in Section \ref{subsubsec:overview-osnap}.

Observe that the lower bound on $pm$ in \eqref{eq:sparsitylb} has three terms, with only one term having the correct dependence on $\log(d/\delta)$ (i.e., matching Conjecture \ref{c:nn}). This is where our parameterized bounds come in: We can choose $\theta$ and $k$ so that this becomes the dominant term. In particular, the natural choice of $\theta = C \log(d/\delta)^\frac{5}{k - 1/2}$, leads to both $\frac{\log(d/\delta)^{5/2}}{\theta ^{k/2-1/4}}$ and $\frac{ \log(d/\delta)^{4}}{\theta ^{k+1/2}}$ becoming $O(1)$. We are able to choose this value of $\theta$ when $C \log(d/\delta)^\frac{5}{k - 1/2} \ge  c_{\ref{thm:osedecoup}.1} \exp(\exp(c_{\ref{thm:osedecoup}.2}k))^2$ (See Corollary \ref{cor:momenttheta} for more details). This condition is satisifed when $k$ grows sufficiently slowly with $d/\delta$, in particular when $k= \lceil c \cdot h \circ h \circ h (d/\delta) \rceil $ with $h(x) = \max \{ \log(x) , 1 \}$ for a fixed small $c$ (note that this is $\Omega( \log \log \log(d/\delta))$)  and leads to the following corollary.

\begin{corollary}\label{cor:osesubpolylogopt}
    An $m \times n$ matrix $\Pi$ with the OSNAP distribution satisfies \eqref{osepro} when 
    \[ m=\frac{c_{\ref{cor:osesubpolylogopt}.1}\log(d/\delta)^\frac{5}{k - 1/2}(d+\log(d/\delta))}{\varepsilon^2} \text{ and } 
    pm \ge c_{\ref{cor:osesubpolylogopt}.2}\log(d/\delta)^{\frac{5}{k - 1/2} + \frac{1}{k}} \cdot \frac{\log(d/\delta)}{\varepsilon^{1+\frac{1}{\log(d/\delta)}}}
    \]
    where $ k= \lceil c_{\ref{cor:osesubpolylogopt}.3}  h \circ h \circ h (d/\delta)  \rceil $ for $h(x) = \max \{ \log(x) , 1 \}$.
\end{corollary}

 Theorem \ref{t:main} now follows from Corollary \ref{cor:osesubpolylogopt}.  

\begin{proof}[Proof of Theorem \ref{t:main}]
We choose $k$ as in Corollary \ref{cor:osesubpolylogopt}. Observe that when $\delta= d^{-O(1)}$, we have $k=\Omega(\log\log\log(d))$, so $\log(d/\delta)^\frac{5}{k - 0.5} = \log(d)^{o(1)}$, which means that $m= O((d\log(d)^{o(1)}/\varepsilon^2)$. Furthermore,  when $\varepsilon \ge d^{-O(1)}$, then  $(1/\varepsilon)^\frac{1}{\log(d/\delta)} = O(1)$, so $pm = O(\log(d)^{1+o(1)}/\varepsilon)$.
\end{proof}

\subsection{Overview of the Techniques.}\label{sec:overview}

In this section, we provide an overview of our key technical contributions that lead to the proof of Theorem \ref{thm:osedecoup}. First, we give a very high-level proof sketch (Section \ref{subsubsec:momest}), and then we explain our key technique, iterative decoupling, in greater depth (Sections \ref{subsubsec:itergen} and \ref{subsubsec:overview-osnap}), with further details provided in Section \ref{subsubsec:whtnotoptimaldimension}. We use the symbols $C, c, C_1, c_1, \ldots$ to denote absolute constants whose values can change from line to line.

\subsubsection{High-level proof sketch.} \label{subsubsec:momest}

We start by reducing the problem of showing \eqref{eq:ose-equiv} to a bound on the moments, following \cite{nelson2013osnap, chenakkod2025optimal} and then using Markov's inequality. Concretely, it suffices to show that
\begin{align*}
    \left( \E[\tr((\Pi U)^T(\Pi U)-I_{d })^{2q}] \right)^{\frac{1}{2q}} \le \varepsilon\quad\text{for}\quad q=\log(d/\delta).
\end{align*}
Here $\tr(\cdot)$ is the normalized trace $\tr(M)=\frac {1}{d}\Tr(M)$.  For convenience, we work with a scaled version of $\Pi$ so that all of its entries are in $\{0,1,-1\}$, i.e., $S:=\sqrt{pm} \Pi$. The equivalent requirement is
\begin{align*}
    \left( \E[\tr((S U)^T(S U)-pm \cdot I_{d })^{2q}] \right)^{\frac{1}{2q}} \le pm\varepsilon.
\end{align*}
 Standard calculations (see \cite{nelson2013osnap,chenakkod2024optimal, chenakkod2025optimal}) show that to get optimal requirement of the embedding dimension $m$ and the sparsity level $pm$ as conjectured in \cite{nelson2013osnap}, it suffices to get the bound 
\begin{align*}
    \left( \E[\tr((S U)^T(S U)-pm I_{d })^{2q}] \right)^{\frac{1}{2q}} =O(\sqrt{pmpd+q^2})
\end{align*}
since $\sqrt{pmpd+q^2} = O(pm\varepsilon)$ when $m \ge c_1d/\varepsilon^2$ and $pm \ge c_2q/\varepsilon$, and recall that $q=\log(d/\delta)$.

It was observed by \cite{cohen2016nearly} and \cite{chenakkod2025optimal} that by using decoupling, one has
\begin{align*}
    \left( \E[\tr((S U)^T(S U)-pmI_{d })^{2q}] \right)^{\frac{1}{2q}}  \le C\left( \E[\tr|(S_1 U)^T(S_2 U)|^{2q}] \right)^{\frac{1}{2q}}=C\norm{(S_1 U)^T(S_2 U)}_{2q}
\end{align*}
for some constant $C$ where $S_1$ and $S_2$ are i.i.d.\ copies of $S$ and we define the absolute value $|M|=\sqrt{M^TM}$ and the norm
\begin{align*}
	\norm{M}_q = \begin{cases}
	\big(\E[\tr |M|^q]\big)^{\frac{1}{q}} & \text{if }1\le q<\infty\\
	\| \norm{M}_{op} \|_\infty & \text{if }q=\infty
	\end{cases}
\end{align*}
Therefore, the OSE problem is reduced to proving
\begin{equation}\label{eq:optmom}
    \begin{aligned}
    \norm{(S_1 U)^T(S_2 U)}_{2q} = O(\sqrt{pmpd+q^2})
\end{aligned}
\end{equation}
For convenience, we define
\begin{align*}
    K=&K(m,d,p,q):=\left(pmpd+(pm)^{1/(q)}q(pd+q) \right)^{1/2}  \approx \sqrt{pmpd+q^2}
\end{align*}
to be the almost correct moment bound for $\norm{(S_1 U)^T(S_2 U)}_{2q}$ (for technical reasons, we introduce some extra errors in $K$, but these errors are not significant). To better understand $K$, we observe that if the embedding dimension and sparsity satisfy the assumptions from Conjecture \ref{c:nn}, then using $q=\log(d/\delta)$, we have $pmpd=(pm\varepsilon)^2=Cq^2$, so $K \approx \sqrt{pmpd} \approx q$.

Thus, our goal is to show
\begin{equation}\label{eq:optmomK}
    \begin{aligned}
    \norm{(S_1 U)^T(S_2 U)}_{2q} \le C K .
\end{aligned}
\end{equation}

\paragraph{\textbf{Moment estimates obtained from iterative decoupling.}}
Using our iterative decoupling technique (explained in the following sections), we obtain a sequence of moment estimates of the following form, 
\begin{align} \norm{(S_1 U)^T(S_2 U)}_{2q \cdot 2^k} \le C_k(K+K^{\alpha_k}(q^{2}R(I_d))^{1-\alpha_k}),\qquad\alpha_k =\max\Big\{0,\frac{2k-1}{2k+1}\Big\}\label{eq:high-level-moment-estimates}
\end{align}
where $k=0,1,2,...$ corresponds to the number of steps of iterative decoupling, and $R(I_d)$ is a quantity that arises from Gaussian universality analysis (which we will ignore for now). Note that when $k=0$, the second term in the bound is of the order at least $q^{2}$, whereas $K \approx q$ (under  the conjectured sparsity), so the second term dominates and the estimate is suboptimal.

However, as $k$ increases, the exponent in $K^{\alpha_k}$ increases and the exponent in $(q^2R(I_d))^{1-\alpha_k}$ decreases, so the estimates get better in the sense that they can be (approximately) bounded by a power of $q$ that approaches 1 under the condition that $K \approx q$ which holds for the conjectured sparsity (see Figure \ref{fig:q-exponent-vs-k}). Ideally, we would take $k$ to infinity and recover \eqref{eq:optmomK}, but unfortunately the constant $C_k$ blows up very quickly. Still, for any $k$, by requiring \eqref{eq:high-level-moment-estimates} to be $O(pm\varepsilon)$ (which was our original goal) we get the sequence of estimates in Theorem \ref{thm:osedecoup}. In particular, via simple algebra it suffices to have,
\begin{align*}
   m \ge C_k'q^{\frac{5}{k-1/2}} \frac{d+\log(1/\delta)}{\varepsilon^2}
\qquad\text{and}\qquad   pm \ge C_k'q^{\frac{5}{k-1/2}} \frac{q}{\varepsilon}.
\end{align*}

\begin{figure}[htbp]
  \centering

  \begin{tikzpicture}[x=0.9cm,y=4cm] % x- & y-scales; tweak to resize
  %---------------- axes ----------------
  \draw[->] (0,0) -- (11,0) node[below right] {$k$};
  \draw[->] (0,0) -- (0,1) node[left] {Exponent of $q$};

  %---------------- grid ----------------
  % vertical grid lines (k = 0 … 10)
  \foreach \k in {0,...,10}
    \draw[gray!35, thin] (\k,0) -- (\k,1);
  % horizontal grid lines (y = 0,0.2,…,1.8)
  \foreach \y in {0,0.2,...,1}
    \draw[gray!35, thin] (0,\y) -- (10.8,\y);

  %---------------- ticks & labels -------
  % x-ticks at integers 0…10
  \foreach \k in {0,...,10}
    \draw (\k,0) -- ++(0,-0.04)
           node[below, font=\scriptsize] {\k};
  % y-ticks every 0.2
  \foreach \y in {1,1.2, 1.4, 1.6}
    \draw (0,\y-0.8) -- ++(-0.12,0)
          node[left, font=\scriptsize] {\y};

  %---------------- asymptote ------------
  \draw[dashed, red, thick] (0,0.2) -- (10.8,0.2)
        node[right, font=\small, black] {$\text{exponent}=1$};

  %---------------- data points & lines --
  \foreach \k in {1,...,10}{
     \pgfmathsetmacro{\y}{(2*\k+3)/(2*\k+1)-0.8}   % E(k)
     % emphasised point
     \fill[blue] (\k,\y) circle (2.6pt);
     % connect to previous integer point
     \ifnum\k>1
        \pgfmathsetmacro{\yp}{(2*(\k-1)+3)/(2*(\k-1)+1)-0.8}
        \draw[blue, thick] (\k-1,\yp) -- (\k,\y);
     \fi
  }
  
\end{tikzpicture}

  \caption{Plotting the exponent in front of $q=\log(d/\delta)$ as a function of \(k\) in $
      \norm{(S_1 U)^T(S_2 U)}_{2q \cdot 2^k} \le C_k(K+K^{\alpha_k}(q^{2}R(I_d))^{1-\alpha_k}) \approx O(q^{1+\frac{2}{2k+1}})
  $ under the condition that $pmpd=Cq^2$ (which holds for the conjectured sparsity).}
  \label{fig:q-exponent-vs-k}
\end{figure}
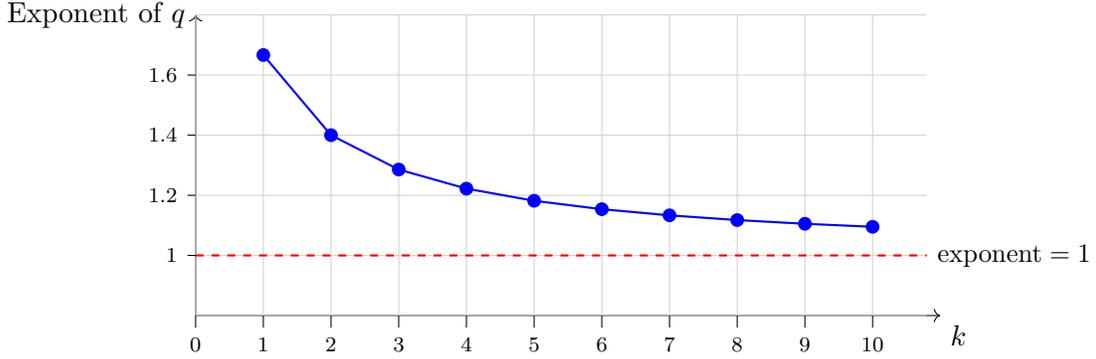

\subsubsection{Towards iterative decoupling: Intuition in a general random matrix model.}
\label{subsubsec:itergen}

Before formally explaining the idea of iterative decoupling that we use to obtain the sequence of moment estimates \eqref{eq:high-level-moment-estimates}, we first explain the motivation for iterative decoupling in a simplified setting. Suppose that we have a symmetric random matrix $M=\sum_lZ_l$ where $Z_1,...,Z_n$ are independent random matrices and we want to estimate the norm $\norm{M}_{2q}$ (here we consider a general situation where $M$ need not be OSNAP).

\paragraph{\textbf{First attempt: Na\"ive full decoupling.}} We first examine a simple and na\"ive strategy of decoupling higher-order moments, to see how it fails. We have the following decomposition:
\begin{align*}
    \norm{M}_{2q}^{2q}&=\tr(M^{2q})
    =\norm{M^{q}}_2^2
    =\norm{(\sum \limits_{l}Z_l)^{2q}}
    \\&=\Big\|\underbrace{\sum \limits_{(l_1,...,l_{q}):\exists k_1,k_2( k_1 \ne k_2,l_{k_1} = l_{k_1})}Z_{l_1} \cdots Z_{l_{q}}}_{T_{\text{diagonal}}}+\underbrace{\sum \limits_{(l_1,...,l_{q}):\forall k_1,k_2( k_1 \ne k_2 \Rightarrow l_{k_1}  \ne l_{k_1})}Z_{l_1} \cdots Z_{l_{q}}}_{T_{\text{off-diagonal}}}\Big\|_2^2
    \\ &\le  (\norm{T_{\text{diagonal}}}_2+\norm{T_{\text{off-diagonal}}}_2)^2
\end{align*}
By decoupling \cite{de2012decoupling} and Jensen's inequality, letting $Z_{1,l_i},Z_{2,l_i},...$ denote i.i.d.\ copies of $Z_{l_i}$ we have
\begin{align*}
    \norm{T_{\text{off-diagonal}}}_2 \le \Big\|\sum \limits_{(l_1,...,l_{q}):\forall k_1,k_2( k_1 \ne k_2 \Rightarrow l_{k_1}  \ne l_{k_1})}Z_{1,l_1} \cdots Z_{q,l_{q}}\Big\|_2
     \le  C({q})\norm{M_1 \cdots M_{q}}_2
\end{align*}
where $M_1,...,M_{q}$ are i.i.d.\ copies of $M$ and the constant $C({q})$ is a rapidly growing function of $q$,
\begin{align*}
    C(q) =2^{q}\bigl(q^{\,q}-1\bigr)\bigl((q-1)^{\,q-1}-1\bigr)\times\cdots\times 3
\end{align*}
Since $M_1,...,M_{q}$ are independent, the quantity
\begin{align*}
    \norm{M_1 \cdots M_{q}}_2^2=&\E\tr((M_1 \cdots M_{q})^TM_1 \cdots M_{q})
\end{align*}
is completely determined by the second moments of $M$ which can be calculated explicitly.
Therefore, by estimating the term
$\norm{T_{\text{diagonal}}}_2$
we can get a bound for $\norm{M}_{2q}$. However, there are two issues:
\begin{itemize}

    \item First, the term $\norm{T_{\text{diagonal}}}_2$
is very complicated, as it has many different types of components, e.g., terms with different numbers of collisions of indices, which makes it difficult to control.
\item Second, the constant $C(q)$ blows up very rapidly when $q \to \infty$, and it is far too large even when $q$ is logarithmic in the dimensions of $M$ (which is the case in the OSNAP analysis).

\end{itemize}

\paragraph{\textbf{Alternative approach: Moment bounds via Gaussian universality.}} 

Another approach to bounding the higher-order moments of $M$ is via Gaussian universality analysis, i.e., comparing these moments to the moments of a corresponding Gaussian random matrix.
We state here the universality result of Brailovskaya and van Handel, which forms the starting point of our analysis.

\begin{proposition}[Theorem 2.9, \cite{brailovskaya2022universality}]\label{prop:bvh}
    Let $M=\sum_{i=1}^nZ_{i}$ be a sum of independent symmetric matrices $Z_i$. Let $G$ be a gaussian random matrix with the same mean and covariance profile as $M$. Then,
    \begin{equation*}
        \abs*{ \norm{M}_{2q} - \norm{G}_{2q} } \le C\tilde{R}_{2q}(M)q^2
        \qquad\text{where}\qquad \tilde{R}_{2q}(M) = \paren*{\sum_{i=1}^n \norm{Z_i}_{2q}^{2q}}^\frac{1}{2q}.
    \end{equation*}
\end{proposition}

Without loss of generality, we can assume that $q=2^k$. Let $M_1, M_2$ be i.i.d.\ copies of $M=\sum_{i=1}^nZ_{i}$ and $G_1,G_2$ be i.i.d.\ copies of the corresponding gaussian models for $M$. We obtain
\begin{align}
    \norm{M}_{2q} \le \norm{G}_{2q}+C\tilde{R}_{2q}(M)q^2=\Theta+C \tilde{R}_{2q}(M)q^2 \label{eq:initialunivgeneral}
\end{align}
and we write $\Theta=\norm{G}_{2q}$ as the moment estimate from the gaussian model. In many cases, e.g., the OSNAP case and other models with sparse entries, the error term $\tilde{R}_{2q}(M)q^2$ dominates $\Theta$, so it would be desirable to improve this error term. More precisely, our goal is to show 
\begin{align*}
    \norm{M}_{2q} \le O(\Theta+\text{error})
\end{align*}
with the error term as small as possible. (In our OSNAP model, $\Theta=\sqrt{pmpd}$ and we want to show that the error term is as small as $q$ so that $\Theta+\text{error}=O(\sqrt{pmpd}+q)=O(\sqrt{pmpd+q^2})$). 

\paragraph{\textbf{Iterative decoupling to fine-tune universality error.}}
We now show how to improve the term $\tilde{R}_{2q}(M)q^2$ by iterative decoupling. In the first step, we write
\begin{align*}
    \norm{M}_{2q}=&\norm{M^2}_q^{1/2}
    =\Big\|\sum_{l}Z_l^2+\sum_{l_1 \ne l_2}Z_{l_1} \cdot Z_{l_2}\Big\|_q^{1/2}
    \le
    {\Big\lVert \!\!\!\underbrace{\sum_{l} Z_l^{2}}_{(\text{diagonal term})} \!\!\!\Big\rVert_{q}^{1/2}}
\;+\;
    {\Big\lVert\!\!\! \underbrace{\sum_{l_1 \neq l_2} Z_{l_1}\,Z_{l_2} }_{(\text{off-diagonal term})}\!\!\!\Big\rVert_{q}^{1/2}}
\end{align*}
Since the diagonal term is a sum of independent random matrices, we can use standard tools such as matrix Bernstein or Matrix Rosenthal to get an easy (but suboptimal) bound. Fortunately, the diagonal term is generally of smaller order than $M^2$ (because most terms $Z_{l_1}Z_{l_2}$ are in the off-diagonal term), so these suboptimal bounds are sufficient for our purpose. Thus, suppose that we have a bound $ \norm{\sum_{l}Z_l^2} \le \Theta^2$.

For the off-diagonal term, we use decoupling and obtain
\begin{align*}
    \norm{\sum_{l_1 \ne l_2}Z_{l_1} \cdot Z_{l_2}}_q \le C\norm{M_1 \cdot M_2}_q
\end{align*}

Now, we condition on $M_2$ and use Proposition \ref{prop:bvh} to obtain
\begin{align*}
    \norm{M_1 \cdot M_2}_q \le \norm{G_1 \cdot M_2}_q + C\norm{\tilde R_q(M_1M_2|M_2)}_{q}\,q^2
\end{align*}
where
\begin{align*}
    \tilde R_q(M_1M_2|M_2)=\left(\sum_{i=1}^n \E[ \tr(|Z_{1,i} \cdot M_2|^{q})|M_2]\right)^{1/q}
\end{align*}

For the term $\norm{G_1 \cdot M_2}_q$, we condition on $G_1$ and use Proposition \ref{prop:bvh} again to obtain
\begin{align*}
    \norm{G_1 \cdot M_2}_q \le \norm{G_1 \cdot G_2}_q + C\norm{ \tilde R_q(G_1M_2|G_1)}_{q}\,q^2
\end{align*}

In summary, we have
\begin{align*}
    \norm{M_1 \cdot M_2}_q \le &\norm{G_1 \cdot G_2}_q + C\norm{\tilde R_q(G_1M_2|G_1)}_{q}q^2+C\norm{\tilde R_q(M_1M_2|M_2)}_{q}\,q^2
    \\ \le & \Theta^2 + C\norm{\tilde R_q(G_1M_2|G_1)}_{q}q^2+C\norm{\tilde R_q(M_1M_2|M_2)}_{q}\,q^2
\end{align*}
where we use H\"older inequality to claim that $\norm{G_1 \cdot G_2}_q \le \Theta^2$.

For the sake of this illustration, assume that $\norm{\tilde R_{q}(M_1M_2|M_2)}$ and $\norm{\tilde R_q(G_1M_2|G_1)}$ are easier to estimate (we show this is true for the OSNAP model) and we have
\begin{align*}
    \norm{\tilde R_{q}(M_1M_2|M_2)}_{q} \le \Theta \cdot \tilde{R}_{2q}(M)
    \qquad\text{and}\qquad\norm{\tilde R_{q}(G_1M_2|G_1)}_{q} \le \Theta \cdot \tilde{R}_{2q}(M).
\end{align*}
Putting this all together, we obtain
\begin{align}
    \norm{M}_{2q} &\le 
        \Big\lVert\sum_{l} Z_l^{2}\Big\rVert_{q}^{1/2}
\;+\;
    {\Big\lVert\sum_{l_1 \neq l_2} Z_{l_1}\,Z_{l_2}\Big\rVert_{q}^{1/2}}
 \nonumber\\
 &\le \Theta+(\Theta^2+ C\norm{\tilde R_q(G_1M_2|G_1)}_{q}q^2+C\norm{\tilde R_q(M_1M_2|M_2)}_{q}q^2)^{1/2} \nonumber
    \\
    &\le \Theta+(\Theta^2+ C\Theta \tilde{R}_{2q}(M)q^2+C\Theta \tilde{R}_{2q}(M)q^2)^{1/2} \nonumber
    \\ &\le  2\Theta +2\sqrt{C}\sqrt{\Theta}\sqrt{\tilde{R}_{2q}(M)q^2} 
\label{eq:improveduniv}
\end{align}
Comparing the last line with the initial moment estimate \eqref{eq:initialunivgeneral} from universality, we see that the error term is changed from $\tilde{R}_{2q}(M)q^2$ to $\sqrt{\Theta}\sqrt{\tilde{R}_{2q}(M)q^2}$. In other words, square root of $\tilde{R}_{2q}(M)q^2$ is replaced by square root of $\Theta$ and this is an improvement when $\tilde{R}_{2q}(M)q^2$ is of higher order~than~$\Theta$.

To carry this step further, we can choose to not apply the universality result to $\norm{M_1 \cdot M_2}_q$, but do decoupling on it again. More precisely, we condition on $M_1$ and write
\begin{align*}
    \norm{M_1 \cdot M_2}_q&=\norm{M_2 \cdot M_1 \cdot M_1 \cdot M_2}_{q/2}^{1/2}
    \\&\le    
    {\Big\lVert \underbrace{\sum_{l} Z_l\cdot M_1 \cdot M_1 \cdot Z_l}_{(\text{diagonal term})} \Big\rVert_{q/2}^{1/2}}
\;+\;
    {\Big\lVert \underbrace{\sum_{l_1 \neq l_2} Z_{l_1} \cdot M_1 \cdot M_1 \cdot Z_{l_2}}_{(\text{off-diagonal term})} \Big\rVert_{q/2}^{1/2}}
\end{align*}

By conditioning on $M_1$, we can still estimate the diagonal term by matrix Bernstein or Matrix Rosenthal. And for the off-diagonal, we have, by decoupling,
\begin{align*}
    \left\lVert \sum_{l_1 \neq l_2} Z_{l_1} \cdot M_1 \cdot M_1 \cdot Z_{l_2} \right\rVert_{q/2}=&\norml{\norm{\sum_{l_1 \neq l_2} Z_{l_1} \cdot M_1 \cdot M_1 \cdot Z_{l_2}|M_1}_{q/2}}_{L_{q/2}(\Pb)}
    \\ \le & \norml{\norm{M_3 M_1 M_1 M_4|M_1}_{q/2}}_{L_{q/2}(\Pb)}
    \\ = & \norm{M_3 M_1 M_1 M_4}_{q/2}
\end{align*}
where we define $\norm{A|B}_{2q}=(\E[(\tr|A|^{2q})|B])^{\frac{1}{2q}}$ for any random matrices $A$ and $B$.

Then we can use Proposition \ref{prop:bvh} conditionally (first on $M_1$ and $M_3$, and then on $M_1$) to estimate $\norm{M_3 M_1 M_1 M_4}_{q/2}$. In the next step, we will further condition on $M_3M_1M_1$ and decouple $M_4$ to get a moment of the form  $\norm{M_5M_1M_1M_3M_3 M_1 M_1 M_6}_{q/4}$. The calculations become complicated here so we do not go into details in this overview. But finally we will further improve the moment estimate for $\norm{M}_{2q}$ after each step. Figure \ref{fig:itergen} shows a diagram that summarizes three steps of iterative decoupling.

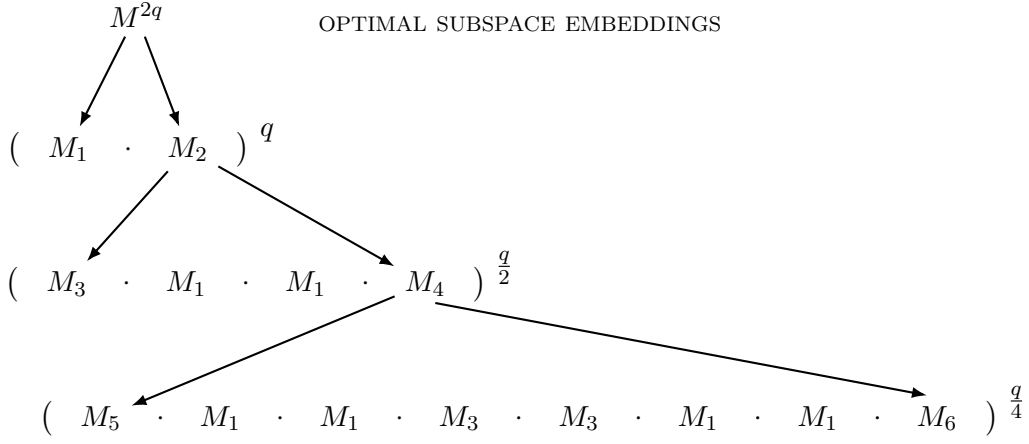
\begin{figure}[t]    
\centering
\vspace{-1cm}\begin{tikzpicture}[
    >=Latex,
    arrow/.style = {-{Latex[length=2mm]}, thick},
    every node/.style = {font=\normalsize}
]

%------------------ first row ------------------
\node (Top) {$M^{2q}$};

%------------------ second row -----------------
\node (M1)   [below left = 1.2cm and 0cm of Top] {$M_{1}$};
\node (dot12)[right = 0.20cm of M1]                {$\cdot$};
\node (M2)   [right = 0.20cm of dot12]             {$M_{2}$};

\node (lpar12)[left  = 0.10cm of M1]     {$\bigl($};
\node (rpar12)[right = 0.10cm of M2]     {$\bigr)$};
\node          at ($(rpar12.north east)+(0.10,-0.15)$) {$q$};

%------------------ third row ------------------
\node (M3)   [below left = 1.2cm and 0.8cm of M2]  {$M_{3}$};
\node (dot31)[right = 0.20cm of M3]                {$\cdot$};
\node (M1a)  [right = 0.20cm of dot31]             {$M_{1}$};
\node (dot32)[right = 0.20cm of M1a]               {$\cdot$};
\node (M1b)  [right = 0.20cm of dot32]             {$M_{1}$};
\node (dot33)[right = 0.20cm of M1b]               {$\cdot$};
\node (M4)   [right = 0.20cm of dot33]             {$M_{4}$};
\node (M4arrhead) [below = -0.2cm of M4] {};

\node (lpar3)[left  = 0.10cm of M3]      {$\bigl($};
\node (rpar3)[right = 0.10cm of M4]      {$\bigr)$};
\node          at ($(rpar3.north east)+(0.10,-0.15)$) {$\tfrac{q}{2}$};

%------------------ fourth row ------------------
\node (M5)   [below left = 1.2cm and 3.5cm of M4]  {$M_{5}$};
\node (dot41)[right = 0.20cm of M5]                {$\cdot$};
\node (M1c)  [right = 0.20cm of dot41]             {$M_{1}$};
\node (dot42)[right = 0.20cm of M1c]               {$\cdot$};
\node (M1d)  [right = 0.20cm of dot42]             {$M_{1}$};
\node (dot43)[right = 0.20cm of M1d]               {$\cdot$};
\node (M3a)  [right = 0.20cm of dot43]             {$M_{3}$};
\node (dot44)[right = 0.20cm of M3a]               {$\cdot$};
\node (M3b)  [right = 0.20cm of dot44]             {$M_{3}$};
\node (dot45)[right = 0.20cm of M3b]               {$\cdot$};
\node (M1e)  [right = 0.20cm of dot45]             {$M_{1}$};
\node (dot46)[right = 0.20cm of M1e]               {$\cdot$};
\node (M1f)  [right = 0.20cm of dot46]             {$M_{1}$};
\node (dot47)[right = 0.20cm of M1f]               {$\cdot$};
\node (M6)   [right = 0.20cm of dot47]             {$M_{6}$};
\node (M6arrhead)   [above = -0.15cm of M6]             {};

\node (lpar4)[left  = 0.10cm of M5]      {$\bigl($};
\node (rpar4)[right = 0.10cm of M6]      {$\bigr)$};
\node          at ($(rpar4.north east)+(0.10,-0.15)$) {$\tfrac{q}{4}$};

%------------------ arrows ---------------------
\draw[arrow] (Top) -- (M1);
\draw[arrow] (Top) -- (M2);

\draw[arrow] (M2) -- (M3);
\draw[arrow] (M2) -- (M4);

\draw[arrow] (M4) -- (M5);
\draw[arrow] (M4arrhead) -- (M6arrhead);

\end{tikzpicture}
\caption{First three steps of iterative decoupling for general random matrices.}
\label{fig:itergen}
\end{figure}

We can see that this new method avoids the issues of the na\"ive full decoupling method we described previously. First, in the na\"ive full decoupling method, the diagonal term \begin{align*}
\sum \limits_{(l_1,...,l_{q}):\exists k_1,k_2( k_1 \ne k_2,l_{k_1} = l_{k_1})}Z_{l_1} \cdots Z_{l_{q}}
\end{align*} consists of collision terms that arise in power $2q$ and these are very complicated and hard to deal with. However, in this new method the diagonal terms at each step, e.g.,
\begin{align*}
    \sum_{l} Z_{l}^2 \text{\ \ \ \  and \ \ \ \ } \sum_{l} Z_l\cdot M_1 \cdot M_1 \cdot Z_l,
\end{align*}  only include collision terms that arise in power $2$. 
By conditioning on appropriate random matrices, these diagonal terms are sums of independent random matrices and therefore are easy~to~analyze. 

Second, in the na\" ive full decoupling method, the constant factor $C(q)$ depends on the moment $q$ and since we need to choose $q=\log(d/\delta)$, this constant blows up rapidly. In our new method, after each time of decoupling, we still collect a constant factor, so our moment estimate for $\norm{M}_{2q}$ still includes a factor $C_k$. However, this constant $C_k$ depends only on the number of steps $k$ for which we do iterative decoupling (after which, we stop and apply the universality result), and we have seen that we can choose $k$ however we wish, e.g., a fixed large constant or a function growing extremely slowly with $q$, so that the factor $C_k$ does not blow up.

\subsubsection{Iterative Decoupling for OSNAP.}
\label{subsubsec:overview-osnap}

We now specialize this paradigm to our specific setting. For convenience, let $X_1=S_1U$, $X_2=S_2U$, and $Y=X_1^TX_2$. Recall that in step 2 of the general iterative decoupling argument, we condition on $M_1$ and decouple $M_2$ to obtain
\begin{align*}
    \norm{M_1 \cdot M_2}_{q} \le \text{(diagonal term)}+\norm{M_3M_1M_1M_4}_{q/2}^\frac{1}{2}
\end{align*}
And in the next step, we will further condition on $M_3M_1M_1$ and decouple $M_4$.
To simplify the argument (in the specialized situation for $Y$), we work with a matrix of the form $VY$ and its norm $\norm{VY}_{2q}$ where $V$ is a fixed matrix, and then we will plug in the random matrices that we condition on in each step  into $V$. In this way, we can do iterative decoupling by induction.

We start from some initial moment estimate results derived by applying Proposition \ref{prop:bvh} and improve this moment estimate iteratively (see Figure \ref{fig:overview} for an overview of the steps).  More precisely, assume that at step $k$ we have the moment estimate
\begin{align*}
    \norm{VY}_{2q} \le \Phi_{k}(V)
\end{align*}
Then in step $k+1$, we write 
\begin{align*}
    \norm{VY}_{4q}=&(\E\tr((Y^TV^TVY)^{2q}))^{\frac{1}{4q}}= \norm{Y^TV^TVY}_{2q}^{1/2}
\end{align*}
Now, $Y^TV^TVY = X_2^TX_1V^TVX_1^TX_2$ contains degree two terms in both $X_1$ and $X_2$, which are independent. Thus, we can decouple the two $X_1$ factors while conditioning on $X_2$ and then decouple the two $X_2$ factors similarly. This gives us the estimate
\begin{align*}
    \norm{Y^TV^TVY}_{2q} &\le \norm{\text{diagonal terms}}_{2q}+4\norm{X_2^TX_1V^TVX_3^TX_4}_{2q} \\
    &\le \norm{\text{diagonal terms}}_{2q}+4\norm{Y_1^TV^TVY_2}_{2q}
\end{align*}

Then, by conditioning on $Y_1$, we can estimate $\norm{Y_1^TV^TVY_2}_{2q}$ by applying the previous moment estimate to $\norm{WY_2}_{2q}$ where $W=Y_1^TV^TV$.

We get,
\begin{align*}
\norm{Y_1^TV^TVY_2}_{2q} \le \norm{\Phi_k(W)}_{L_{2q}}=\norm{\Phi_k(Y_1^TV^TV)}_{L_{2q}}
\end{align*}

Plugging this into the previous inequalities, we have
\begin{align*}
    \norm{VY}_{4q}=&(\E\tr((Y^TV^TVY)^{2q}))^{\frac{1}{4q}}
\\=&\norm{Y^TV^TVY}_{2q}^{1/2}
\\ \le & \norm{\text{diagonal terms}}^{1/2}_{2q}+2\norm{Y_1^TV^TVY_2}_{2q}^{1/2}
\\ \le & \norm{\text{diagonal terms}}^{1/2}_{2q}+2\norm{\Phi_k(Y_1^TV^TV)}_{L_{2q}}^{1/2}
\end{align*}

After estimating the diagonal terms and simplifying $\norm{\Phi_k(Y_1^TV^TV)}_{L_{2q}}$, the last line will become a function of $V$, and we denote this function by $\Phi_{k+1}(V)$.

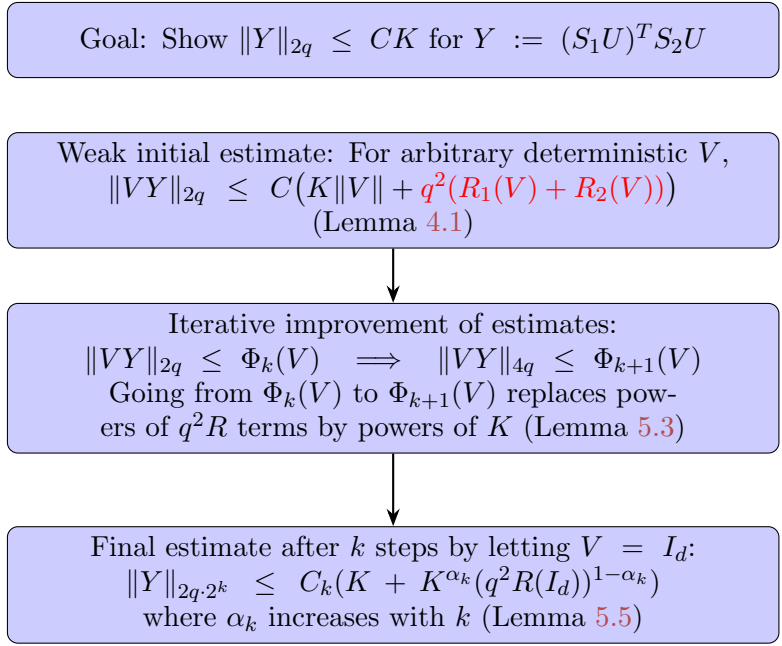
\begin{figure}[!ht]
\centering
\begin{tikzpicture}
\node (box1) [roundedbox] {Goal: Show $\norm{Y}_{2q} \le CK$ for $Y:= (S_1U)^TS_2U$};
\node (box2) [roundedbox, below of=box1, node distance=2cm] {
Weak initial estimate: For arbitrary deterministic $V$, \\
$\norm{VY}_{2q} \le C\paren*{ K\norm{V} + \textcolor{red}{q^2(R_1(V)+R_2(V))} }$ \\
(Lemma \ref{lem:initialuniv})};
\node (box3) [roundedbox, below of=box2, node distance=2.5cm] {Iterative improvement of estimates: \\
$\norm{VY}_{2q} \le \Phi_k(V) \implies \norm{VY}_{4q} \le \Phi_{k+1}(V)$ \\
Going from $\Phi_{k}(V)$ to $\Phi_{k+1}(V)$ replaces powers of $q^2 R$ terms by powers of $K$
(Lemma \ref{lem:iterdec})};
\node (box4) [roundedbox, below of=box3, node distance=2.75cm] {Final estimate after $k$ steps by letting $V=I_d$:\\
$\norm{Y}_{2q \cdot 2^k} \le C_k(K+K^{\alpha_k}(q^{2}R(I_d))^{1-\alpha_k})$ \\
where $\alpha_k$ increases with $k$
(Lemma \ref{lem:univseq})
};

\draw [arrow] (box2) -- (box3);
\draw [arrow] (box3) -- (box4);

\end{tikzpicture}
\caption{Overview of the iterative decoupling argument for OSNAP.}\label{fig:overview}
\end{figure}

By using this method, we can iteratively improve the moment estimate for $VY$. And in the final estimate, we will plug in $V=I_d$ to obtain an estimate for the moment of $Y=X_1^TX_2$, which is what we need for the analysis of OSE properties.

To illustrate our argument, we show details of the dynamics in this iterative process in the first two steps. In step 0, we apply Proposition \ref{prop:bvh} in two steps as described in Section \ref{subsubsec:itergen} and obtain an upper bound
\begin{align*}
    \norm{VY}_{2q}=&\norm{VX_1^TX_2}_{2q}
   \\ \le& C(\norm{VG_1^TG_2}_{2q} + q^{2}R(V))\\\le& C(\sqrt{pmpd}\norm{V} + q^{2}R(V))
   \\\le& C(K\norm{V} + q^{2}R(V))
\end{align*}
where $R(V)$ is a function of $V$ ($R$ is motivated by $\tilde R$ but has different meaning; see Lemma \ref{lem:initialuniv} for details the precise definition of $R$). Plugging in $V=I_d$, we have
\begin{align*}
    \norm{Y}_{2q}=&\norm{X_1^TX_2}_{2q}
   \\ \le& C(\norm{G_1^TG_2}_{2q} + q^{2}R(I_d))\\\le& C(\sqrt{pmpd} + q^{2}R(I_d))
   \\\le& C(K + q^{2}R(I_d))
\end{align*}
where we can bound $R(I_d) \approx (pm)^{1/q}\sqrt{pd+q}$ (see Lemma \ref{lem:R1R2} for details).

Right at this stage, we are able to obtain the first of our family of results, Corollary \ref{cor:osebasic}. Recall from Section \ref{subsubsec:momest} that we want $\norm{Y}_{2q}$ to be $ O(pm\varepsilon)$. It suffices to have $m \ge c_1d/\varepsilon^2$ and $pm \ge q/\varepsilon$ for $K$ to be $O(pm\varepsilon)$. For the second term, $q^2(pm)^{1/q}\sqrt{pd+q}$ to be $O(pm\varepsilon)$, we need $pm \ge (q^{5/2}/\varepsilon)+q^4$. This bound is similar to \cite[Theorem 7]{chenakkod2025optimal}, with slightly worse powers of $\log$ in the sparsity but much easier to obtain.

Next, in step 1, following the framework previously described, we have
\begin{align*}
\norm{VY}_{4q}
\le & (\norm{\text{diagonal terms 1}}_{2q}+4\norm{Y_1^TV^TVY_2}_{2q})^{1/2}&
\end{align*}
Then we set $W=Y_1^TV^TV$ and apply the previous moment estimate to $\norm{WY_2}_{2q}$ when conditioning on $W$ as we described before. We have
\begin{align*}
\norm{Y_1^TV^TVY_2}_{2q}=&\norm{\norm{WY_2|W}_{2q}}_{L_{2q}}
\\=&\norm{C(K\norm{W} + q^{2}R(W))}_{L_{2q}}
\\=&\norm{C(K\norm{Y_1^TV^TV} + q^{2}R(Y_1^TV^TV))}_{L_{2q}}
\\ \le &C \left( K\norm{Y_1}_{2q}\norm{V}^2+ q^{2}\norm{R(Y_1^TV^TV))}_{L_{2q}}\right)
\end{align*}
where we recall the definition $\norm{WY_2|W}_{2q}=(\E[(\tr|WY_2|^{2q})|W])^{\frac{1}{2q}}$ and use independence between $W$ and $Y_2$.

The key reason that each time the moment estimate can be improved is because we can ``replace" $Y_1$ in $\norm{R(Y_1^TV^TV))}_{L_{2q}}$ by the desired factor 
$K$. More precisely, (see Lemma \ref{lem:Rsimplify}) we have
\begin{align*}
    \norm{R(Y_1^TV^TV)}_{L_{2q}} \le K \norm{V} R(V)
\end{align*}
This is possible because of the special properties of the OSNAP model, e.g., it can be written as sum of independent rank 1 matrices, and we have strong bound for $\norm{\norm{Yu}_2}_{L_{2q}(\Pb)}=\norm{\norm{(S_1 U)^T(S_2 U)u}_2}_{L_{2q}(\Pb)}$ for an arbitrary fixed vector $u$ (see Section \ref{sec:vecmombd}). 

Therefore, we have
\begin{align*}
&\norm{VY}_{4q}\\\le & (\norm{\text{diagonal terms 1}}_{2q}+4\norm{Y_1^TV^TVY_2}_{2q})^{1/2}
\\\le & (\norm{\text{diagonal terms 1}}_{2q}+C \left( K\norm{Y_1}_{2q}\norm{V}^2+ q^{2}K \norm{V} R(V)\right)^{1/2}
\\\le & (\norm{\text{diagonal terms 1}}_{2q}+C \left( K^{1/2}\norm{Y_1}_{2q}^{1/2}\norm{V}+ K^{1/2}(q^{2} \norm{V} R(V))^{1/2}\right)
\end{align*}

In addition, we can bound
\begin{align*}
    \norm{\text{diagonal terms 1}}_{2q}^{1/2} \le (K+\sqrt{K}\sqrt{\norm{Y}_{2q}}) \norm{V}
\end{align*}
by using standard tools such as Matrix Rosenthal (see Lemma \ref{lem:decoupfirstlayer} and Lemma \ref{lem:decoupsecondlayer} for details).

Plugging in $V=I_d$, and simplifying,

\begin{align*}
\norm{Y}_{4q} \le 
C\left( K+K^{1/2}(q^{2}  R(V))^{1/2} \right)
\end{align*}
From the last line, we see that after first time decoupling, we can replace square root of the bad factor $q^{2}R(I_d)$ by square root of the good factor $K$.

The calculations of the later steps become more complicated and the formal inductive argument is shown in the proof of Lemma \ref{lem:iterdec}. In general, for large $k$, after $k$ iterative decoupling steps, our moment bound becomes
\begin{align*}
    \norm{Y}_{2q \cdot 2^k} \le C_k(K+K^{\frac{2k-1}{2k+1}}(q^{2}R(I_d))^{\frac{2}{2k+1}})
\end{align*}
See Lemma \ref{lem:univseq} for details. Note that this formula gives a weaker result than we have seen for $k=1$ because in the formal iterative decoupling argument in Lemma \ref{lem:iterdec}, we introduce some artificial extra terms in the moment bound $\Phi_k$ when $k$ is small to  to keep a convenient form of $\Phi_k$ for all $k$, but these minor modifications do not affect the long run behavior.

\subsubsection{Obtaining the Subspace Embedding Guarantee.}
\label{subsubsec:whtnotoptimaldimension}

Recall that to get the OSE property, it suffices to require
\begin{align*}
    C\norm{X_1^TX_2}_{2q} \le pm \varepsilon
\end{align*}
where $q=\log(d)$ and the constant $C$ comes from the initial decoupling step.

Since our iterative decoupling attains the bound
\begin{align*}
    \norm{X_1^TX_2}_{2q} \le \norm{Y}_{2q \cdot 2^k} \le C_k(K+K^{\frac{2k-1}{2k+1}}(q^{2}R(I_d))^{\frac{2}{2k+1}})
\end{align*}
it suffices to require
\begin{equation}\label{eq:momenttoose}
    C'_k(K+K^{\frac{2k-1}{2k+1}}(q^{2}R(I_d))^{\frac{2}{2k+1}}) \le pm \varepsilon
\end{equation}

Since $K \approx \sqrt{pmpd+q^2}$, roughly speaking we need
\begin{align*}    C'_k\Big(\sqrt{pmpd+q^2}+\big(\sqrt{pmpd+q^2}\big)^{\frac{2k-1}{2k+1}}\big(q^{2}R(I_d)\big)^{\frac{2}{2k+1}}\Big) \le pm \varepsilon.
\end{align*}

When $pmpd\ge q^2$, this requirement becomes
\begin{align*}
    C''_k\Big(\sqrt{pmpd}+\big(\sqrt{pmpd}\big)^{\frac{2k-1}{2k+1}}\big(q^{2}R(I_d)\big)^{\frac{2}{2k+1}}\Big) \le pm \varepsilon
\end{align*}

Under optimal embedding dimension, $m=\theta\frac{d}{\varepsilon^2}$ with $\theta=O(1)$, we have $pm \varepsilon=C'''\sqrt{pmpd}$, and so \eqref{eq:momenttoose} can only be satisfied when $C''_k=O(1)$ and
\begin{align*}
    pm\varepsilon = C'''\sqrt{pmpd} \ge q^{2}R(I_d)
\end{align*}
leading to the same sparsity guarantee that we would attain using the initial moment estimate. This is why we need to trade off the embedding dimension in order to improve the sparsity via iterative decoupling.

Formally, by Thoerem \ref{thm:osedecoup}, when $m=\theta\cdot(d+\log(d/\delta))/\varepsilon^2$, the requirement (\ref{eq:momenttoose}) reduces to 
\begin{align*}
    {pm } \ge {\exp(\exp(c_{\ref{thm:osedecoup}.3}k))}\left(\frac{1}{\varepsilon^{1+\frac{1}{\log(d/\delta)}}}\Big(\theta \log(d/\delta)+\frac{\log(d/\delta)^{5/2}}{\theta ^{k/2-1/4}}\Big)+\frac{ \log(d/\delta)^{4}}{\theta ^{k+1/2}}\right)
\end{align*}

In order to remove the extra $\log(d/\delta)$ factors, we need $\theta ^{k/2-1/4}$ to cancel out with $\log(d/\delta)^{5/2}$ and $\theta ^{k+1/2}$ to cancel out with $ \log(d/\delta)^{4}$, and this works only when $\theta$ grows with $\log(d/\delta)$ (but it may grow very slowly).

\subsection{Outline of the Paper.}
Section \ref{sec:prelim} details the basic facts from the existing literature that we will use in this paper in addition to proving some fundamental estimates for our main proof. Section \ref{sec:initialmom} establishes the initial weak moment estimate that is later improved by iterative decoupling. Section \ref{sec:iterdecoup} proves a sequence of moment estimates by iterative decoupling. Section \ref{sec:finalres} establishes the implications of the moment estimates for lower bounds on sparsity. Sections \ref{sec:diagterms} and \ref{sec:vecmombd} contain some of the longer proofs of the intermediate steps required in our main argument. In appendix \ref{sec:decoup}, we formally prove the decoupling results used throughout the paper.

\section{Preliminaries.} \label{sec:prelim}

\subsection{Notation.} \label{subsec:notation}
The following notation and terminology will be used in the paper. The notation $[n]$ is used for the set $\{1,2,...,n\}$ and the notation $\operatorname{P}([n])$ denotes the set of all partitions of $[n]$. Also, for two integers $a$ and $b$ with $a \le b$, we use the notation $[a:b]$ for the set $\{k \in \Z:a \le k \le b\}$. For $x \in \R$, we use the notation $\lfloor x \rfloor$ to denote the greatest integer less than or equal to $x$ and $\lceil x \rceil$ to denote the least integer greater than or equal to  $x$. In $\R^n$ (or $\R^m$ or $\R^d$), the $l$th coordinate vector is denoted by $e_l$. 

We use the notation $\mathbb{P}$ for the standard probability measure, and the notation $\mathbb{E}$ for the expectation with respect to this standard probability measure. To simplify the notation, we follow the convention from \cite{brailovskaya2022universality} and use the notation $\E [X]^{\alpha}$ for $(\E(X))^{\alpha}$, i.e., when a functional is followed by square brackets, it is applied before any other operations. The covariance of two random variables $X$ and $Y$ is denoted by $\cov(X,Y)$. The standard $L_q$ norm of a random variable $\xi$ is denoted by $\norm{\xi}_{L_q(\Pb)}$, for $1 \le q \le \infty$. When there is no possibility of confusion, we will use the short notation $\norm{\xi}_{q}$ for the $L_q$ norm of $\xi$. For convenience of iterated integration, we introduce the following notations. Let $X,Y$ be two random elements taking values in measurable spaces $\mathscr{S}_1$ and $\mathscr{S}_2$. Let $\phi:\mathscr{S}_1 \times \mathscr{S}_2 \to \R$ be a measurable function. Assume that $\phi(X,Y)$ is integrable. Define the function $\psi_2:\mathscr{S}_2 \to \R$ such that $\psi_2(y)=\E\phi(X,y)$. Then the notation $\E_X(\phi(X,Y))$ stands for $\psi_2(Y)$. Similarly, we define the function $\psi_1:\mathscr{S}_1 \to \R$ such that $\psi_1(x)=\E\phi(x,Y)$ and the notation $\E_Y(\phi(X,Y))$ stands for the random variable $\psi_1(X)$.

All matrices considered in this paper are real valued and the space of $m \times n$ matrices with real valued entries is denoted by $M_{m \times n}(\mathbb{R})$. Also, for a matrix $X \in M_{d \times d}$, the notation $\Tr (X)$ denotes the trace of the matrix $X$, and $\tr (X) = \frac{1}{d} \Tr (X)$ denotes the normalized trace. The $d \times d$ identity matrix is denoted by $I_d$. We write the operator norm of a matrix $X$ as $\norm{X}$, and it is also denoted by $\norm{X}_{op}$ in some places where other norms appear for clarity. The spectrum of a matrix $X$ is denoted by $\spec(X)$. $\norm{X}_{S_q}$ denotes the normalised Schatten norm, $\norm{X}_{S_{q}}=\big(\tr |X|^{q}\big)^{\frac{1}{q}}$, where $\abs{X} = (X^TX)^{1/2}$. For a $d \times d$ matrix $M$, following \cite{brailovskaya2022universality}, we define the absolute value $|M|=\sqrt{M^TM}$ and normalized trace $\tr(M)=\frac {1}{d}\Tr(M)$. Let $L_q(S_q^d)$ be the normed vector space of $d\times d$ random matrices $M$ with norm
\begin{align*}
	\norm{M}_q = \begin{cases}
	\big(\E[\tr |M|^q]\big)^{\frac{1}{q}} & \text{if }1\le q<\infty\\
	\| \norm{M}_{op} \|_\infty & \text{if }q=\infty
	\end{cases}
\end{align*}
More precisely, the space $L_q(S_q^d)$ consists of random matrices $M$ with $\norm{M}_q$ well defined (which means $\big(\E[\tr |M|^q]\big)<\infty$ for $1\le q<\infty$ and $\| \norm{M}_{op} \|_\infty$ for $q=\infty$). For any random matrix $A$ and any random element $B$, we define the conditional $L_q(S_q^d)$-norm $\norm{A|B}_{q}=(\E[(\tr|A|^{q})|B])^{\frac{1}{q}}$ conditioning on $B$ when the conditional expectation $\E[(\tr|A|^{q})|B]$ exists. Although this notation is valid for general random matrices $A$ and $B$, in this paper, we only need the special case where the random matrix $A=f(B,C)$ is a measurable function of two independent random elements $B$ and $C$. In this case, we have $\norm{f(B,C)|B}_{q}=(\E[(\tr|f(B,C)|^{q})|B])^{\frac{1}{q}}=(\E_{C}[(\tr |f(B,C)|^{q})])^{\frac{1}{q}}$.

Throughout the paper, the symbols $c_1, c_2, ...$, and $Const, Const', ...$ denote absolute constants. 

\bigskip 

\subsection{Measure Concentration.} \label{subsec:measconc}

We first list the standard measure concentration inequalities that we will need.

\begin{lemma}[Theorem 15.10 from \cite{boucheron2013concentration}]\label{lem:rosenthal1}
Let $\xi=\sum_{j=1}^n\zeta_i$ where $\zeta_1,...,\zeta_n$ are independent and nonnegative random variables. Then the exists a constant $c_{\ref{lem:rosenthal1}}$ such that, for all integers $q \ge 1$, we have
\begin{align*}
    (\E (\xi^q))^{1/q} \le 2 \E( \xi)+ c_{\ref{lem:rosenthal1}}q (\E ((\max \limits_{j=1,...,n} \zeta_j)^q))^{1/q}
\end{align*}
\end{lemma}

\begin{lemma}[Lemma 6.2.2. in \cite{vershynin2018high}]\label{lem:6.2.2}
Let $X, X' \sim \mathcal{N}(0, I_n)$ be independent and let $A = (a_{ij})$ be an $n \times n$ matrix. Then the exist constants $c_{\ref{lem:6.2.2}.1}$ and $c_{\ref{lem:6.2.2}.2}$ such that
\begin{align*}
\mathbb{E}\bigl[\exp\bigl(\lambda\,X^T A\,X'\bigr)\bigr]
&\;\le\; \exp\bigl(c_{\ref{lem:6.2.2}.1}^2\,\lambda^2\,\|A\|_F^2\bigr)
\end{align*}
for all $\lambda$ satisfying 
\begin{align*}
|\lambda| \;\le\; \frac{c_{\ref{lem:6.2.2}.2}}{\|A\|}.
\end{align*}
\end{lemma}

\begin{lemma}[Lemma 6.2.3. in \cite{vershynin2018high}]\label{lem:6.2.3}
Consider independent mean-zero sub-gaussian random vectors $X, X' \in \mathbb{R}^n$ 
with $\|X\|_{\psi_2} \le K$ and $\|X'\|_{\psi_2} \le K$. 
Consider also independent random vectors $g, g' \sim \mathcal{N}(0,I_n)$. 
Let $A$ be an $n \times n$ matrix. 
Then there exists a constant $c_{\ref{lem:6.2.3}}$ such that
\begin{align*}
\mathbb{E}\bigl[\exp\bigl(\lambda\,X^T A\,X'\bigr)\bigr]
\;\le\;
\mathbb{E}\bigl[\exp\bigl(c_{\ref{lem:6.2.3}}\,K^2\,\lambda\,g^T A\,g'\bigr)\bigr]
\end{align*}
for any $\lambda \in \mathbb{R}$.
\end{lemma}

\begin{lemma}
[Theorem 2.3 in \cite{boucheron2013concentration}]\label{thm:subgamma-moments}
Let $X$ be a centered random variable. Suppose there exist constants $v>0$ and $c\ge 0$ such that 
\begin{align*}
\mathbb{P}\bigl(\{X > \sqrt{2vt} +ct\} \cup \{-X > \sqrt{2vt} + ct\bigr\}) 
\le e^{-t}
\quad\text{for all}\; t>0.
\end{align*}
Then for every integer $q \ge 1$, 
\begin{align*}
\mathbb{E}\bigl[X^{2q}\bigr] 
\le
q!(8v)^{q} +(2q)!\bigl(4c^{2}\bigr)^{q}.
\end{align*}
\end{lemma} 

\begin{lemma}\label{lem:subgaussian-chaos}
Let $X, X' \in \mathbb{R}^n$ be independent, mean-zero sub-gaussian random vectors with subgaussian norm bounded by $1$, 
and let $A$ be an $n\times n$ matrix. 
Then there exist a constant $c_{\ref{lem:subgaussian-chaos}}$, such that for any integer $q \ge 1$, we have
\begin{align*}
\|\,X^T A X'\|_{L_q(P)} 
\le c_{\ref{lem:subgaussian-chaos}}(\sqrt{q}\|A\|_F +q\|A\|).
\end{align*}
\end{lemma}

\begin{proof}

By the Lemma \ref{lem:6.2.2} and \ref{lem:6.2.3}, there is a universal constant $c_1,c_2$ such that for any $|\lambda| \;\le\; \frac{c_1}{\|A\|}$, we have
\begin{align*}
\mathbb{E}\Bigl[\exp\bigl(\lambda\,X^T A\,X'\bigr)\Bigr]
\;\le\;
\exp\Bigl(c_2\,\lambda^2\,\|A\|_F^2\Bigr),
\end{align*}

By standard Chernoff method, for any $|\lambda| \;\le\; \frac{c_1}{\|A\|}$, we get
\begin{align*}
    \Pb(X^T AX'>t) \le \exp(-\lambda t/2+c_2\,\lambda^2\,\|A\|_F^2)
\end{align*}

Optimizing over $\lambda$, we have
\begin{align*}
    \Pb(X^T AX'>t) \le \exp(-c_3\min(\frac{t^2}{\|A\|_F^2},\frac{t}{\norm{A}}))
\end{align*}

Let $t=\sqrt{2 \|A\|_F^2 r} + \norm{A}r$.

Then we have
\begin{align*}
    \frac{t^2}{\|A\|_F^2} \ge \frac{2 \|A\|_F^2 r}{\|A\|_F^2}=2r
\end{align*}
and
\begin{align*}
    \frac{t}{\norm{A}} \ge \frac{\norm{A}r}{\|A\|}=r
\end{align*}

Therefore, we have
\begin{align*}
    \min(\frac{t^2}{\|A\|_F^2},\frac{t}{\norm{A}}) \ge \min(2r,r) \ge r
\end{align*}

Therefore, we have
\begin{align*}
\mathbb{P}\Bigl(|Z_+|>\sqrt{2 \|A\|_F^2 r} + \norm{A}r\Bigr)\le e^{-r}
\end{align*}
 where $Z_+=\max\{Z,0\}$.

Therefore, by Lemma \ref{thm:subgamma-moments}, we have
\begin{align*}
    \norm{Z_+}_{L_q(\Pb)} \le c_4(\sqrt{q}\|A\|_F +q\|A\|)
\end{align*}
for some constant $c_4$.

Similarly, we have
\begin{align*}
    \norm{Z_-}_{L_q(\Pb)} \le c_5(\sqrt{q}\|A\|_F +q\|A\|)
\end{align*}
for some constant $c_5$.

By triangle inequality, we have
\begin{align*}
    \norm{Z}_{L_q(\Pb)} \le \norm{Z_+}_{L_q(\Pb)}+ \norm{Z_-}_{L_q(\Pb)} \le c_6(\sqrt{q}\|A\|_F +q\|A\|)
\end{align*}
for some constant $c_6$.

\end{proof}

\bigskip

\subsection{Basic Facts of the $L_q(S_q^d)$ Space.}

We need the following H\"older inequality in $L_q(S_q^d)$ from \cite{brailovskaya2022universality}.

\begin{lemma}[Lemma 5.3. in \cite{brailovskaya2022universality}]
\label{lem:holdervanhandel}
Let $1\le 
\beta_1,\ldots,\beta_k\le\infty$ satisfy $\sum_{i=1}^k\frac{1}{\beta_i}=1$. Then
$$
	|\E[\tr Y_1\cdots Y_k]| \le
	\|Y_1\|_{\beta_1}\cdots \|Y_k\|_{\beta_k}
$$
for any $d\times d$ random matrices $Y_1,\ldots,Y_k$.
\end{lemma}

\subsection{Spectrum of Gaussian Matrices.}
Here we collect results about the spectrum of various Gaussian models that will be used later in conjunction with universality results.

   \begin{lemma}[(2.3), \cite{rudelson2010non}]\label{lem:Gaussianspectrum}
    For $m>d$, let $G$ be an $m \times d$ matrix whose entries are independent standard normal variables. Then,
    \begin{align*}
        \Pb(\sqrt{m}-\sqrt{d}-t \leq s_{\min}(G) \leq s_{\max}(G) \leq \sqrt{m}+\sqrt{d}+t) \geq 1 - 2e^{-t^2/2}
    \end{align*}
\end{lemma}

\begin{corollary}[Corollary 6.8  in \cite{chenakkod2025optimal} (Full version)] \label{cor:gaussianmom}
    Let $G$ be an $m \times d$ matrix whose entries are independent normal random variables with variance $\frac{1}{m}$. Let $\varepsilon<\frac{1}{6}$ and $q \in \mathbb{N} \le m \varepsilon^2 $. Then, there exists $c_{\ref*{cor:gaussianmom}} > 1$ such that for $m \geq \frac{c_{\ref*{cor:gaussianmom}}d}{\varepsilon^2}$,
    \[ \E[\tr(G^TG - I_d)^{2q}]^\frac{1}{2q} \leq  \varepsilon \]
    
\end{corollary}

\begin{lemma}[Trace Moment of Decoupled Gaussian Model, Lemma 6.9  in \cite{chenakkod2025optimal} (Full version) ]\label{cor:indgaussianmom}
    Let $G_1$ and $G_2$ be independent $m \times d$ random matrices with i.i.d. Gaussian entries. Then for any $q \le O(d)$, there exists $c_{\ref*{cor:indgaussianmom}}>0$ such that 
    \[ \norm{G_1^TG_2}_{2q} \le c_{\ref{cor:indgaussianmom}}\sqrt{\max\{m,q\}}\sqrt{\max\{d,q\}}\]
    
\end{lemma}

\subsection{Oblivious Subspace Embeddings.} \label{subsec:osnapprops}

Here, we define the OSNAP distribution that we work with and show some of its important properties that we shall use later. The definition is the same as \cite{chenakkod2025optimal} (Full version), which we reproduce here for convenience.

 \begin{definition}[OSNAP]\label{def:osnap} 
 An $m \times n$ random matrix $S$ is called an unscaled oblivious sparse norm-approximating projection with independent subcolumns distribution (unscaled OSNAP)
  with parameters $p, \varepsilon, \delta \in (0,1]$ such that $s=pm$ divides $m$, if 
 \[ S = \sum_{l=1}^n \sum_{\gamma=1}^s \xi_{(l,\gamma)} e_{\mu_{(l, \gamma)}} e_l ^\top \]
 where $\{ \xi_{(l,\gamma)} \}_{l \in [n], \gamma \in [s]}$ is a collection of independent Rademacher random variables, $\{ \mu_{(l,\gamma)} \}_{l \in [n], \gamma \in [s]}$ is a collection of independent random variables such that each $\mu_{(l,\gamma)}$ is uniformly distributed in $[(m/s)(\gamma-1)+1:(m/s)\gamma]$ and $e_{\mu_{(l, \gamma)}}$ and $e_l$ represent basis vectors in $\R^m$ and $\R^n$ respectively, with $\{ \xi_{(l,\gamma)} \}_{l \in [n], \gamma \in [s]}$ being independent with $\{ \mu_{(l,\gamma)} \}_{l \in [n], \gamma \in [s]}$.

 In this case, $\Pi = (1/\sqrt{pm})S$ is called an OSNAP with parameters $p, \varepsilon, \delta$. 
 \end{definition}

 Thus, each column of the OSNAP matrix $\Pi$ has $pm$ many non-zero entries, and the sparsity level can be varied by setting the parameter $p \in [0,1]$ appropriately.

 \begin{remark}
     The assumption that $\{ \xi_{(l,\gamma)} \}_{l \in [n], \gamma \in [s]}$ and $\{ \mu_{(l,\gamma)} \}_{l \in [n], \gamma \in [s]}$ are fully independent can be relaxed to log-wise independence. For further details, we refer the reader to the proof of \cite[Theorem 3.2]{chenakkod2025optimal} (Full Version).
 \end{remark}

From \cite{chenakkod2025optimal} (Full Version), we have the following result about the covariance structure of random OSNAP matrices.

\begin{lemma}[Lemma 6.1 in \cite{chenakkod2025optimal} (Full Version)] 
Let $p = p_{m,n} \in (0,1]$ and $S=\{s_{ij}\}_{i \in [m], j \in [n]}$ be a $m \times n$ random matrix as in the unscaled OSNAP distribution. Then, $\E(s_{ij})=0$ and $\operatorname{Var}(s_{ij})=p$ for all $i \in [m], j \in [n]$, and $\cov(s_{i_1 j_1},s_{i_2 j_2})=0$ for any $\{i_1,i_2\} \subset [m], \{j_1,j_2\} \subset [n]$ and $ (i_1,j_1)\neq (i_2,j_2) $.
\end{lemma}

We now derive bounds on the moments of various quantities associated to OSNAP matrices using results from Section \ref{subsec:measconc}.

\begin{lemma}\label{lem:xfrobnorm}
    Let $S$ be an $m \times n$ matrix distributed according to the unscaled OSNAP distribution with parameter $p$. Let $U$ be an arbitrary $n \times d$ deterministic matrix such that $U^TU=I_d$. Let $X=SU$. Let $v \in \R^d$ be a vector with $\norm{v^T}_{l_2([d])}=1$. Let $q \ge 1$ be an integer. Then there exists constants $c_{\ref{lem:xfrobnorm}.1}, c_{\ref{lem:xfrobnorm}.2} > 0$ such that,
    \[ \E \sqbr*{\norm{X}^{2q}_F}^\frac{1}{2q} \le c_{\ref{lem:xfrobnorm}.1}(\sqrt{pmd} + \sqrt{q}) \]
    and,
    \[ \E \sqbr*{\norm{Xv}^{2q}_2}^\frac{1}{2q} \le c_{\ref{lem:xfrobnorm}.2}(\sqrt{pm} + \sqrt{q}) \]
\end{lemma}

\begin{proof}
    Note that,
    \begin{align*}
        \norm{X}^2_F &= \Tr (X^T X) \\
        &= \Tr \paren*{\sum_{l_1, l_2 =1}^n \sum_{\gamma=1}^{pm} u_{l_1}e_{\mu_{(l_1, \gamma)}}^T \xi_{(l_1, \gamma)} \xi_{(l_2, \gamma)} e_{\mu_{(l_2, \gamma)}} u_{l_2}^T } \\
        &= \Tr \paren*{\sum_{l =1}^n \sum_{\gamma=1}^{pm} u_{l} u_{l}^T } + \Tr \paren*{\sum_{\substack{l_1, l_2 =1 \\ l_1 \neq l_2}}^n \sum_{\gamma=1}^{pm} u_{l_1}e_{\mu_{(l_1, \gamma)}}^T \xi_{(l_1, \gamma)} \xi_{(l_2, \gamma)} e_{\mu_{(l_2, \gamma)}} u_{l_2}^T } \\
        &= pmd + \sum_{\substack{l_1, l_2 =1 \\ l_1 \neq l_2}}^n \sum_{\gamma=1}^{pm} \xi_{(l_1, \gamma)} \xi_{(l_2, \gamma)} \ip{e_{\mu_{(l_1, \gamma)}}, e_{\mu_{(l_2, \gamma)}}} \ip{u_{l_1}, u_{l_2}}  \\
    \end{align*}

    Thus, 
\[ \E \sqbr*{\norm{X}^{2q}_F}^\frac{1}{2q} \le \sqrt{pmd} + \E \sqbr*{ \paren*{\sum_{\substack{l_1, l_2 =1 \\ l_1 \neq l_2}}^n \sum_{\gamma=1}^{pm} \xi_{(l_1, \gamma)} \xi_{(l_2, \gamma)} \ip{e_{\mu_{(l_1, \gamma)}}, e_{\mu_{(l_2, \gamma)}}} \ip{u_{l_1}, u_{l_2}}}^q }^\frac{1}{2q} \]

The latter term is a quadratic form of the random vector formed by embedding $\{ \xi_{(l, \gamma)} \}_{l \in [n], \gamma \in [pm]}$ in $\R^{pmn}$ with $\xi_{(l, \gamma)}$ going in the $((\gamma-1)pm + l)\textsuperscript{th}$ coordinate. Call this random vector $\xi$. Then the latter term is $\E \sqbr{ \paren{\xi^T A \xi}^q }^\frac{1}{2q}$, where $A$ is a $pmn \times pmn$ block diagonal matrix with $n \times n$ blocks $A_1 \etc A_{pm}$ on the diagonal, with $[A_\gamma]_{l_1,l_2} =  \mathbbm{1}_{l_1 \neq l_2} \ip{e_{\mu_{(l_1, \gamma)}}, e_{\mu_{(l_2, \gamma)}}} \ip{u_{l_1}, u_{l_2}} $ for $l_1,l_2 \in [n]$ and $\gamma \in [pm]$. After conditioning on $\{\mu_{(l,\gamma)}\}_{l \in [n], \mu \in [pm]}$ and decoupling as in the proof of the Hanson-Wright inequality (e.g., Theorem 6.2.1 in \cite{vershynin2018high}), we can apply Lemma \ref{lem:subgaussian-chaos} to obtain, 
\[ \E \sqbr*{ \paren{\xi^T A \xi}^q \big\vert \{\mu_{(l,\gamma)}\}_{l \in [n], \mu \in [pm]} } \le c_1^q\paren*{\sqrt{q}\norm{A}_F + q\norm{A}}^q \]
for some constant $c_1$
Observe that each matrix $A_\gamma$ is the matrix $UU^T$ with some entries being made zero. Since the spectral norm of a matrix with nonnegative entries can only decrease when some entries are made zero, $\norm{A_\gamma} \le \norm{UU^T} = 1$, which means $\norm{A} \le 1$. By the same argument, $\norm{A}_F^2 = \sum_{\gamma=1}^{pm} \norm{A_\gamma}_F^2 \le pmd$. So,
\begin{align*}
    \E \sqbr*{ \paren{\xi^T A \xi}^q \big\vert \{\mu_{(l,\gamma)}\}_{l \in [n], \mu \in [pm]} } &\le c_1^q\paren*{\sqrt{q}\sqrt{pmd} + q}^q \\
    &\le  c_1^q\paren*{pmd + 2q}^q \\
    \E \sqbr*{ \paren{\xi^T A \xi}^q }^\frac{1}{q} &\le c_1\paren*{pmd + 2q}
\end{align*}
Combining with the above inequality for $\E \sqbr*{\norm{X}^{2q}_F}^\frac{1}{2q}$, we get, 
\[ \E \sqbr*{\norm{X}^{2q}_F}^\frac{1}{2q} \le c_2(\sqrt{pmd} + \sqrt{q}) \]
for some constant $c_2$

Similarly,
\begin{align*}
        \norm{Xv}^2_F &= \ip{Xv, Xv} \\
        &= \sum_{l_1, l_2 =1}^n \sum_{\gamma=1}^{pm}  \xi_{(l_1, \gamma)} \xi_{(l_2, \gamma)} \ip{v, u_{l_1}} \ip{e_{\mu_{(l_1, \gamma)}}, e_{\mu_{(l_2, \gamma)}}} \ip{v, u_{l_2}}  \\
        &= \sum_{l =1}^n \sum_{\gamma=1}^{pm} \ip{v, u_{l_1}}^2  + \sum_{\substack{l_1, l_2 =1 \\ l_1 \neq l_2}}^n \sum_{\gamma=1}^{pm} \xi_{(l_1, \gamma)} \xi_{(l_2, \gamma)} \ip{v, u_{l_1}} \ip{e_{\mu_{(l_1, \gamma)}}, e_{\mu_{(l_2, \gamma)}}} \ip{v, u_{l_2}}  \\
    \end{align*}

In this case we have, $[A_\gamma]_{i,j} =  \mathbbm{1}_{i \neq j} \ip{e_{\mu_{(l_1, \gamma)}}, e_{\mu_{(l_2, \gamma)}}} \ip{u_{l_1}, v} \ip{u_{l_2}, v} $, and so the matrix $A_\gamma$ is the matrix $(Uv)(Uv)^T$ with some entries made zero. It still holds that $\norm{A} \le 1$ and $\norm{A}_F^2 \le pm$. Following the previous argument, we get,
\[ \E \sqbr*{\norm{Xv}^{2q}_2}^\frac{1}{2q} \le c_3(\sqrt{pm} + \sqrt{q}) \]
for some constant $c_3$.

\end{proof}

\begin{lemma}[Moments of quantities based on a single row of OSNAP]\label{lem:rowbounds}
    Let $S$ be an $m \times n$ matrix distributed according to the fully independent unscaled OSNAP distribution with parameter $p$. Let $U$ be an arbitrary $n \times d$ deterministic matrix such that $U^TU=I_d$. Let $X=SU$. Let $v \in \R^d$ be a vector with $\norm{v^T}_{l_2([d])}=1$. Let $q \ge 1$ be an integer. Then there exists constants $c_{\ref{lem:rowbounds}.1}, c_{\ref{lem:rowbounds}.2} > 0$ such that,
    \[ \E \sqbr*{\norm{e_{1}^TX}^{2q}}^\frac{1}{2q} \le c_{\ref{lem:rowbounds}.1}(\sqrt{pd} + \sqrt{q}) \]
    and,
    \[ \E \sqbr*{\norm{e_{1}^TXv}^{2q}_2}^\frac{1}{2q} \le c_{\ref{lem:rowbounds}.2}(\sqrt{p} + \sqrt{q}) \]

\end{lemma}
\begin{proof}
Note that $S_{1,1}$ (which is the $(1,1)$th entry of $S$) is non-zero when the one hot distribution on the column submatrix $S_{[1:m/s]\times{1}}$ has it's non zero entry on the first row. The probability that this happens is $1/(m/s) = s/m = p$. Moreover, the columns of $S$ are independent. Thus, we conclude that for looking at the distribution of $e_{1}^TX$ we may assume that $S$ has independent 
$\pm 1$ entries sparsified by independent Bernoulli $p$ random variables, i.e., $s_{i,j}=\delta_{i,j} \xi_{i,j}$ where $\delta_{i,j}$ are Bernoulli random variables taking value 1 with probability $p \in (0,1]$, $\xi_{i,j}$ are random variables with $\Pb(\xi_{i,j}=1)=\Pb(\xi_{i,j}=-1)=1/2$ and the collection $\{\delta_{i,j}, \xi_{i,j} \}_{i \in [m], j \in [n]}$ is independent.

Now, $\norm{e_{1}^TX}$ is a convex 1-Lipschitz function of the entries of  the first row of $S$, and
\begin{align*}
    \E_S[\norm{e_{1}^TX}] &\le  \E_S[\norm{e_{1}^TX}^2]^\frac{1}{2} \\
    &= \sqrt{ pd}
\end{align*}
By \cite[Theorem 6.10]{boucheron2013concentration}, we have, for $\lambda > 0$,
\begin{align*}
    \Pb_{S} ( \norm{e_{1}^TX} \ge \sqrt{ pd } + \lambda ) \le \exp(-c_1\lambda^2)
\end{align*}
for some $c_1>0$. Similarly, since $\abs{e_{1}^TXv}$ is a convex and 1-Lipschitz as a function of the entries of the first row of $S$, 
\begin{align*}
    \Pb_{S} ( \abs{e_{1}^TXv} \ge \sqrt{ p } + \lambda ) \le \exp(-c_1\lambda^2)
\end{align*}

Thus $\norm{e_{1}^TX}$ and $\abs{e_{1}^TXv}$ have a subgaussian tail. The claims now follow by a standard moment computation of subgaussian random variables.

\end{proof}

\begin{corollary}[Norm of a Random Row in OSNAP]\label{lem:rownormbound}
    Let $S$ be a fully independent unscaled OSNAP distribution with parameter $p$. Let $U$ be an $n \times d$ matrix such that $U^TU=I_d$. Let $\mu$ be a random variable uniformly distributed in $
    J \subset [m]$ and independent of $S$ and $G$. Then, there exists $c_{\ref{lem:rownormbound}}> 0$ such that for any positive integer $q>0$, we have
    \begin{align*}
        \E_{\mu, S} [ \norm{e_{\mu}^TSU}^{q} ]^\frac{1}{q} \le c_{\ref{lem:rownormbound}}\sqrt{\max\{pd,q \}} 
    \end{align*}

\end{corollary}
\begin{proof}
    By Hölder's inequality, it suffices to prove these bounds for moments of the order of the smallest even integer bigger than $q$, so without loss of generality, we may assume that $q$ is an even integer. 
    \begin{align*}
        \E_{\mu, S} [ \norm{e_{\mu}^TSU}^{q}] =& \frac{1}{\abs{J}}\sum_{j \in J} \E_{S(t)} [ \norm{e_{j}^TSU}^{q}] \\
        & = \E_{S} [ \norm{e_{1}^TS(t)U}^{q}]
    \end{align*}
    since rows of $S(t)U$ are identically distributed. The claim now follows from Lemma \ref{lem:rowbounds}.

\end{proof}

\section{Initial Moment Estimate.} \label{sec:initialmom}

Recall that, to get the optimal embedding dimension and sparsity requirement, it suffices to bound
\begin{equation}\label{eq:optmomcopy}
    \begin{aligned}
    \norm{(S_1 U)^T(S_2 U)}_{2q} = O(\sqrt{p\max\{m,q\}p\max\{d,q\}+q^2})
\end{aligned}
\end{equation}

To this end, we design the following iterative decoupling method to get almost optimal estimate for $\norm{(S_1 U)^T(S_2 U)}_{2q}$. For convenience, let $Y=X_1^TX_2$. Briefly speaking, we work with the matrix with the form $VY$ and its norm $\norm{VY}_{2q}$ where $V$ is a fixed matrix. We start from some initial moment estimate results derived by applying Proposition \ref{prop:bvh} and improve this moment estimate iteratively. Each time, we write 
\begin{align*}
    \norm{VY}_{4q}=&(\E\tr((Y^TV^TVY)^{2q}))^{\frac{1}{4q}}
\\=&\norm{Y^TV^TVY}_{2q}^{1/2}
\end{align*}
and then use decoupling to estimate
\begin{align*}
    \norm{Y^TV^TVY}_{2q} \le \norm{\text{diagonal terms}}_{2q}+\norm{Y_1^TV^TVY_2}_{2q}
\end{align*}
Then, by conditioning on $Y_1$, we can estimate $\norm{Y_1^TV^TVY_2}_{2q}$ by applying the previous universality result to $\norm{WY_2}_{2q}$ where $W=Y_1^TV^TV$. By using some basic tools, each time the universality result is improved.

By using this method, we can iteratively improve the moment estimate for $VY$. And in the final estimate, we will plug in $V=I_d$ to obtain an estimate for the moment of $Y=X_1^TX_2$, which is what we need for the analysis of OSE properties.

To make the above argument precise, we start with the following initial moment estimate, which follows directly from Proposition \ref{prop:bvh}.

\begin{lemma}[Initial Moment Estimate]\label{lem:initialuniv}
Let $S$ be an $m \times n$ matrix distributed according to the fully independent unscaled OSNAP distribution with parameter $p$ and standard independent summands $\{Z_{(l,\gamma)}\}_{(l,\gamma) \in \Xi}$. Let $\{(S_{\lambda},\{Z_{\lambda,(l,\gamma)}\}_{(l,\gamma) \in \Xi})\}_{\lambda=0,1,2,...}$ be independent copies of $(S,\{Z_{(l,\gamma)}\}_{(l,\gamma) \in \Xi})$. Let $U$ be an arbitrary $n \times d$ deterministic matrix such that $U^TU=I_d$.
Let $X=SU$, and $X_1,X_2,...$ be i.i.d. copies of $X$. Let $Y=X_{1}^TX_{2}$. For any $d \times d$ deterministic matrix $V$ and positive integer $r$, define
\begin{align*}
    R_{1,2r}(V)=(\sum\limits_{(l,\gamma) \in \Xi}(\norm{VX_1^TZ_{2,(l,\gamma)}U}_{2r}^{2r}))^{1/(2r)}
\end{align*}
and
\begin{align*}
    R_{2,2r}(V)=(\sqrt{pd+2r}) (\sum\limits_{(l,\gamma) \in \Xi}\frac{1}{d}\norm{Vu_l}_{l_2([d])}^{2r})^{1/(2r)}
\end{align*}
Then there exists a constant $c_{\ref{lem:initialuniv}}>1$ such that for any $n \times d$ deterministic matrix $V$, and $\log(d) \le q \le r$, we have
\begin{equation}\label{eq:initmom}
    \begin{aligned}
    &\norm{VY}_{2q} \\\le&  c_{\ref{lem:initialuniv}}(\sqrt{p\max\{m,q\}p\max\{d,q\}}\norm{V} + q^{2}R_{2,2r}(V) + q^{2}R_{1,2r}(V))
\end{aligned}
\end{equation}

\end{lemma}

\begin{proof}

We have $\norm{VY}_{2q}=\norm{VX_1^TX_2}_{2q}$. By conditioning on $X_1$, we first estimate $\norm{VW^TX_2}_{2q}$ for some fixed matrix $W$, and later we plug in $W=X_1$ and uncondition on $X_1$ using the fact that $\norm{VX_1X_2}_{2q}^{2q}$ can be computed using iterated integration.

Applying Proposition \ref{prop:bvh} to $\sym(VW^TX_2)$, we have
\begin{align*}
    &\norm{VW^TX_2}_{2q}
    \\ \le & \norm{VW^TG_2}_{2q} + c_2(\sum\limits_{(l,\gamma) \in \Xi}(\norm{VW^TZ_{2,(l,\gamma)}U}_{2r}^{2r}))^{1/(2r)}q^{2}
\end{align*}
where $\Xi := [n]\times[pm]$ denotes the set of index pairs $(l,\gamma)$.

Plugging in $W=X_1$ and conditioning on $X_1$, we have
\begin{align*}
    \norm{VY}_{2q}=&\norm{VX_1^TX_2}_{2q}
    \\ \le & \norm{VX_1^TG_2}_{2q} + c_2q^{2}\left(\E_{X_1}(\sum\limits_{(l,\gamma) \in \Xi}(\E_{\{Z_{2,(l,\gamma)}\}_{(l,\gamma) \in \Xi}} \tr(VX_1^TZ_{2,(l,\gamma)}U)^{2r}))^{\frac{1}{2r} \cdot (2q)}\right)^{1/2q}
    \\ \le & \norm{VX_1^TG_2}_{2q} + c_2q^{2}\left(\E_{X_1}(\sum\limits_{(l,\gamma) \in \Xi}(\E_{\{Z_{2,(l,\gamma)}\}_{(l,\gamma) \in \Xi}} \tr(VX_1^TZ_{2,(l,\gamma)}U)^{2r}))^{\frac{1}{2r} \cdot (2r)}\right)^{1/2r}
    \\ = & \norm{VX_1^TG_2}_{2q} + c_2R_{1,2r}(V)q^{2}
\end{align*}
where we use the fact that $\norm{\xi}_{L_{2q}(\Pb)} \le \norm{\xi}_{L_{2r}(\Pb)}$ for any random variable $\xi$ because $q \le r$.

Formally, we have the following argument.

Let $f_2(W)=\norm{VW^TX_2}_{2q}$. We have
\begin{align*}
    \norm{VY}_{2q}=\norm{VX_1^TX_2}_{2q} = \norm{f_2(X_1)}_{L_{2q}(\Pb)} 
\end{align*}

We have
\begin{align*}
    f_2(W)=&\norm{VW^TX_2}_{2q}
    \\ \le & \norm{VW^TG_2}_{2q} + c_2(\sum\limits_{(l,\gamma) \in \Xi}(\norm{VW^TZ_{2,(l,\gamma)}U}_{2r}^{2r}))^{1/(2r)}q^{2}
\end{align*}

Let $f_3(W)=\norm{VW^TG_2}_{2q} $ and $f_4(W)=(\sum\limits_{(l,\gamma) \in \Xi}(\norm{VW^TZ_{1,(l,\gamma)}U}_{2r}^{2r}))^{1/(2r)}$.

Therefore, we have
\begin{align*}
    \norm{VX_1^TX_2}_{2q}  = &\norm{f_2(X_1)}_{L_{2q}(\Pb)} 
\\  \le & \norm{f_3(X_1)+c_2q^{2}f_4(X_1)}_{L_{2q}(\Pb)} 
\\ \le & \norm{f_3(X_1)}_{L_{2q}(\Pb)} + c_2q^{2}\norm{f_4(X_1)}_{L_{2q}(\Pb)} 
\end{align*}

We have
\begin{align*}
    \norm{f_3(X_1)}_{L_{2q}(\Pb)} = & (\E_{X_1}((\E_{G_2} \tr (|VX_1^TG_2|^{2q}))^{1/(2q)})^{2q})^{1/(2q)} \\ = &\norm{VX_1^TG_2}_{2q}
\end{align*}
and
\begin{align*}
    \norm{f_4(X_1)}_{L_{2q}(\Pb)} \le &\norm{f_4(X_1)}_{L_{2r}(\Pb)} \\= & (\E_{X_1}((\sum\limits_{(l,\gamma) \in \Xi}(\E_{X_2}\tr((VW^TZ_{1,(l,\gamma)}U)^{2r})))^{1/(2r)})^{2r})^{1/(2r)} \\ = & (\E_{X_1}\sum\limits_{(l,\gamma) \in \Xi}\E_{X_2}\tr((VW^TZ_{1,(l,\gamma)}U)^{2r}))^{1/(2r)} 
    \\ = & (\sum\limits_{(l,\gamma) \in \Xi}\E_{X_1}\E_{X_2}\tr((VW^TZ_{1,(l,\gamma)}U)^{2r}))^{1/(2r)}
    \\=& R_{1,2r}(V)
\end{align*}

To estimate $\norm{VX_1^TG_2}_{2q}$, we condition on $G_2$ and estimate $\norm{VX_1^TW}_{2q}$ for some fixed matrix $W$. Applying Proposition \ref{prop:bvh} to $\sym(VX_1^TW)$, we have
\begin{align*}
    &\norm{VX_1^TW}_{2q} 
    \\ \le & \norm{VG_1^TW}_{2q} + c_2 (\sum\limits_{(l,\gamma) \in \Xi}(\norm{VU^TZ_{1,(l,\gamma)}W}_{2q}^{2q}))^{1/(2q)}q^{2}
\end{align*}

Plugging in $W=G_2$ and unconditioning on $G_2$ as before, we have
\begin{align*}
    &\norm{VX_1^TG_2}_{2q} 
    \\ \le & \norm{VG_1^TG_2}_{2q} + c_2 (\sum\limits_{(l,\gamma) \in \Xi}(\norm{VU^TZ_{1,(l,\gamma)}G_2}_{2r}^{2r}))^{1/(2r)}q^{2}
\end{align*}
where we use again the fact that $\norm{\xi}_{L_{2q}(\Pb)} \le \norm{\xi}_{L_{2r}(\Pb)}$ for any random variable $\xi$ because $q \le r$.

For the first term in the last line, by Lemma \ref{cor:indgaussianmom}, we have
\begin{align*}
    \norm{VG_1^TG_2}_{2q}
     \le &\norm{V}\norm{G_1^TG_2}_{2q}
    \\ \le &\sqrt{p\max\{m,q\}p\max\{d,q\}}\norm{V}
\end{align*}
and for the second term, we have
\begin{align*}
    &(\sum\limits_{(l,\gamma) \in \Xi}(\norm{VU^TZ_{1,(l,\gamma)}G_1}_{2r}^{2r}))^{1/(2r)} 
    \\ = & (\sum\limits_{(l,\gamma) \in \Xi}\E_{G_2}\E_{X_1}\tr((VU^TZ_{1,(l,\gamma)}G_2)^{2r}))^{1/(2r)}
\end{align*}

We want to estimate
\begin{align*}
    \E_{G_2}\E_{X_1}\tr((VU^TZ_{1,(l,\gamma)}G_2)^{2r})=\E_{X_1}\E_{G_2}\tr((VU^TZ_{1,(l,\gamma)}G_2)^{2r})
\end{align*}

By definition of OSNAP, we have
\begin{align*}
    U^TZ_{1,(l,\gamma)}=u_le_{\mu_{l,\gamma}}^T
\end{align*}

For any vector in $\bS^{m-1}$, we know that $v^TG_2$ is a random vector with i.i.d. gaussian coordinates of variance $p$. Therefore, we have
\begin{align*}
    \norm{\norm{v^TG_2}_{l_2([d])}}_{L_{2r}(\Pb)} \le c_3 \sqrt{pd+2r}
\end{align*}

Let $f_8(v)=\norm{\norm{v^TG_2}_{l_2([d])}}_{L_{2r}(\Pb)}$.

We have
\begin{align*}
&\E_{X_1}\E_{G_2}\tr((VU^TZ_{1,(l,\gamma)}G_2)^{2r}) \\ = & \E_{X_1}\E_{G_2}\tr((Vu_le_{\mu_{l,\gamma}}^TG_2)^{2r})
\\ \le &\E_{X_1}\E_{G_2}\frac{1}{d}(\norm{V u_l}_{l_2([d])}\norm{e_{\mu_{l,\gamma}}^TG_2}_{l_2([d])})^{2r}
\\ = & \frac{1}{d}\norm{V u_l}_{l_2([d])}^{2r}\E_{X_1} (f_8(e_{\mu_{l,\gamma}}^T))^{2r}
\\ \le & \frac{1}{d}\norm{V u_l}_{l_2([d])}^{2r} (c_3\sqrt{pd+2r})^{2r}
\end{align*}

Therefore, we have
\begin{align*}
    & (\sum\limits_{(l,\gamma) \in \Xi}\E_{G_2}\E_{X_1}\tr((VU^TZ_{1,(l,\gamma)}G_2)^{2r}))^{1/(2r)}
    \\ \le & (\sum\limits_{(l,\gamma) \in \Xi}\frac{1}{d}\norm{V u_l}_{l_2([d])}^{2r} (c_3\sqrt{pd+2r})^{2r})^{1/(2r)}
    \\ \le & c_3(\sqrt{pd+2r}) (\sum\limits_{(l,\gamma) \in \Xi}\frac{1}{d}\norm{Vu_l}_{l_2([d])}^{2r})^{1/(2r)}
    \\=& c_3R_{2,2r}(V)
\end{align*}

\end{proof}

Before moving on to improve this moment estimate, we first check what OSE results the initial moment estimate can give. To this end, we plug in $V=I_d$ into the estimate (\ref{eq:iterassum}) and get
\begin{align*}
    &\norm{Y}_{2q} \\\le&  c_{\ref{lem:initialuniv}}(\sqrt{p\max\{m,q\}p\max\{d,q\}} + q^{2}R_{2,2r}(I_d) + q^{2}R_{1,2r}(I_d))
\end{align*}
And the next lemma gives the estimates for $R_{1,2r}(I_d)$ and $R_{2,2r}(I_d)$.

\begin{lemma}[Estimates of $R_{1,2r}$ and $R_{2,2r}$]\label{lem:R1R2}
We have
\begin{align*}
    R_{1,2r}(I_d) \le (pm)^{1/(2r)}(c_{\ref{lem:R1R2}.1}\sqrt{pd+2r})
\end{align*}
and
\begin{align*}
    R_{2,2r}(I_d) \le (pm)^{1/(2r)}(c_{\ref{lem:R1R2}.2}\sqrt{pd+2r})
\end{align*}
where $R_{1,2r}(V)$ and $R_{2,2r}(V)$ are defined as in Lemma \ref{lem:initialuniv}.
\end{lemma}

\begin{proof}
To estimate $R_{1,2r}(I_d)$, we observe that
\begin{align*}
    R_{1,2r}(I_d)=&(\sum\limits_{(l,\gamma) \in \Xi}(\norm{X_1^TZ_{2,(l,\gamma)}U}_{2r}^{2r}))^{1/(2r)}
    \\=&(\sum\limits_{(l,\gamma) \in \Xi}(\norm{U^TS_1^T\xi_{2,(l,\gamma)}e_{\mu_{2,(l,\gamma)}}e_l^TU}_{2r}^{2r}))^{1/(2r)}
    \\=&(\sum\limits_{(l,\gamma) \in \Xi}(\frac{1}{d}\E(\norm{U^TS_1^Te_{\mu_{2,(l,\gamma)}}}_2^{2r})\norm{u_l}_2^{2r}))^{1/(2r)}
\end{align*}

By Lemma \ref{lem:rownormbound}, we know that
\begin{align*}
    \E(\norm{U^TS_1^Te_{\mu_{2,(l,\gamma)}}}_2^{2r}) \le (c_1\sqrt{pd+2r})^{2r}
\end{align*}
for some constant $c_1$.

Therefore, we have
\begin{align*}
    R_{1,2r}(I_d)
    =&(\sum\limits_{(l,\gamma) \in \Xi}(\frac{1}{d}\E(\norm{U^TS_1^Te_{\mu_{2,(l,\gamma)}}}_2^{2r})\norm{u_l}_2^{2r}))^{1/(2r)}
    \\ \le & (\sum\limits_{(l,\gamma) \in \Xi}(\frac{1}{d}(c_1\sqrt{pd+2r})^{2r}\norm{u_l}_2^{2r}))^{1/(2r)}
    \\ \le & (\sum\limits_{(l,\gamma) \in \Xi}(\frac{1}{d}(c_1\sqrt{pd+2r})^{2r}\norm{u_l}_2^{2}))^{1/(2r)}
    \\=& (pmd\frac{1}{d}(c_1\sqrt{pd+2r})^{2r})^{1/(2r)}
    \\=& (pm)^{1/(2r)}(c_1\sqrt{pd+2r})
\end{align*}

To estimate $R_{2,2r}(I_d)$, we calculate
\begin{align*}
    R_{2,2r}(I_d)=&(\sqrt{pd+2r}) (\sum\limits_{(l,\gamma) \in \Xi}\frac{1}{d}\norm{u_l}_{l_2([d])}^{2r})^{1/(2r)}
    \\ \le & (\sqrt{pd+2r}) (\sum\limits_{(l,\gamma) \in \Xi}\frac{1}{d}\norm{u_l}_{l_2([d])}^{2})^{1/(2r)}
    \\ = & (\sqrt{pd+2r}) (pmd\frac{1}{d})^{1/(2r)}
    \\ = & (\sqrt{pd+2r}) (pm)^{1/(2r)}
\end{align*}

\end{proof}

In summary, from initial moment estimates, we have
\begin{align}
    \norm{Y}_{2q} =  O(\sqrt{p\max\{m,q\}p\max\{d,q\}} + q^{2}(pm)^{1/(2r)}\sqrt{pd+2r}) \label{eq:bvhmombound}
\end{align}

Comparing this estimate with our final objective
\begin{align*}
    &\norm{Y}_{2q} \\=&  O(\sqrt{p\max\{m,q\}p\max\{d,q\}+q^2})
    \\=&  O(\sqrt{p\max\{m,q\}p\max\{d,q\}}+q)
\end{align*}
we see that the main issue is the term $q^{2}(pm)^{1/(2r)}\sqrt{pd+2r}$ which will introduce extra $\log (d/\delta)$ factors in the sparsity requirements. It is this term that we improve in subsequent estimates by iterative decoupling.

\section{Iteratively Improved Moment Estimate.}\label{sec:iterdecoup}

Now we explain in detail how one can remove the term $q^{2}(pm)^{1/(2r)}\sqrt{pd+2r}$, by the iterative decoupling argument. The iterative decoupling argument starts with using the following bound
\begin{align*}
    \norm{Y^TV^TVY}_{2q} \le \norm{\text{diagonal terms}}_{2q}+\norm{Y_1^TV^TVY_2}_{2q}
\end{align*}
to transform $\norm{Y^TV^TVY}_{2q}$ into diagonal terms and a decoupled term. Furthermore, we need to obtain a good bound for the diagonal terms. These objectives are achieved by the following first layer and second layer decoupling lemmas (Lemma \ref{lem:decoupfirstlayer} and \ref{lem:decoupsecondlayer}) whose proofs can be found in Section \ref{sec:diagterms}.

\begin{lemma}[First Layer Decoupling]\label{lem:decoupfirstlayer}
Let $S$, $p$, $\{Z_{(l,\gamma)}\}_{(l,\gamma) \in \Xi}$, $\{(S_{\lambda},\{Z_{\lambda,(l,\gamma)}\}_{(l,\gamma) \in \Xi})\}_{\lambda=0,1,2,...}$, $U$, $X$, $X_1,X_2,...$, and $Y_1,Y_2,...$ be as in Lemma \ref{lem:initialuniv}.
Then for any $d \times d$ deterministic matrix $V$, we have
\begin{align*}
    &\norm{Y_1^TV^TVY_1}_{2q} \\ \le&
    \norm{V}^2 c_{\ref{lem:decoupfirstlayer}}(pmpd+q^2+(pm)^{1/(q)}q(pd+q))\\&+4\norm{X_3^TX_1V^TVX_1^TX_4}_{2q}
\end{align*}

\end{lemma}

\begin{lemma}[Second Layer Decoupling]\label{lem:decoupsecondlayer}
Let $S$, $p$, $\{Z_{(l,\gamma)}\}_{(l,\gamma) \in \Xi}$, $\{(S_{\lambda},\{Z_{\lambda,(l,\gamma)}\}_{(l,\gamma) \in \Xi})\}_{\lambda=0,1,2,...}$, $U$, $X$, $X_1,X_2,...$, and $Y_1,Y_2,...$ be as in Lemma \ref{lem:initialuniv}.
Then for any $d \times d$ deterministic matrix $V$, we have
\begin{align*}
    &\norm{X_3^TX_1V^TVX_1^TX_4}_{2q} \\ \le&
    \norm{V}^2 c_{\ref{lem:decoupsecondlayer}}(pmpd+pd \norm{Y_1}_{2q}+(pm)^{1/(q)}q(pd+q))\\&+4\norm{X_3^TX_5V^TVX_6^TX_4}_{2q}
\end{align*}

\end{lemma}

Combining the estimates from Lemma \ref{lem:decoupfirstlayer} and Lemma \ref{lem:decoupsecondlayer}, we get,
\begin{align*}
    &\norm{Y_1^TV^TVY_1}_{2q} \\ \le&
    \norm{V}^2 (c_{\ref{lem:decoupfirstlayer}}+ 4c_{\ref{lem:decoupfirstlayer}})(pmpd+q^2+(pm)^{1/(q)}q(pd+q))\\&+4c_{\ref{lem:decoupfirstlayer}}pd\norm{V}^2\norm{Y_1}_{2q} +16\norm{X_3^TX_5V^TVX_6^TX_4}_{2q}
\end{align*}

Using this estimate, we can iteratively improve the moment estimate for $VY$.

\begin{lemma}[Iteratively Improved Moment Estimate]\label{lem:iterdec}
Let $S$ be an $m \times n$ matrix distributed according to the fully independent unscaled OSNAP distribution with parameter $p$ and standard independent summands $\{Z_{(l,\gamma)}\}_{(l,\gamma) \in \Xi}$. Let $\{(S_{\lambda},\{Z_{\lambda,(l,\gamma)}\}_{(l,\gamma) \in \Xi})\}_{\lambda=0,1,2,...}$ be independent copies of $(S,\{Z_{(l,\gamma)}\}_{(l,\gamma) \in \Xi})$. Let $U$ be an arbitrary $n \times d$ deterministic matrix such that $U^TU=I_d$.
Let $X=SU$. Let $Y_i=X_{2i-1}^TX_{2i}$.
For any $d \times d$ deterministic matrix $V$ and positive integer $r$, let $R_{1,2r}(V)$ and $R_{2,2r}(V)$ be defined as in Lemma \ref{lem:initialuniv}.
Let
\begin{align*}
    K_{q,r}=&K(m,d,p,q,r)\\=&\left(p\max\{m,q\}p\max\{d,q\}+(pm)^{1/(q)}r(pd+r) \right)^{1/2}
\end{align*}
Then there exist constants $c_{\ref{lem:iterdec}.1}>1$ and $c_{\ref{lem:iterdec}.2}>1$ such that for any $q \ge \log(d)$, $\alpha_1>0,...,\alpha_7>0$, and $\kappa > c_{\ref{lem:iterdec}.1}$, the following statement is true.

Assume that
    $\alpha_3 \le 1$ and $\alpha_1+\alpha_3=1$
and for any $n \times d$ deterministic matrix $V$ and $r \ge q$, we have
\begin{equation}\label{eq:iterassum}
    \begin{aligned}
    &\norm{VY}_{2q} \\\le& \kappa \Bigg((K_{q,r}+\sqrt{pd\norm{Y}_{2q}})\norm{V} \\&+ \left(R_{1,2r}(V)+R_{2,2r}(V)\right)^{\alpha_1}K_{q,r}^{\alpha_2}\norm{V}^{\alpha_3} \norm{Y}_{2q}^{\alpha_4}q^{\alpha_5}\Bigg)
\end{aligned}
\end{equation}
then for any $n \times d$ deterministic matrix $V$ and $r \ge 2q$, we have
\begin{equation}\label{eq:iterconclu}
    \begin{aligned}
    &\norm{VY}_{4q} \\\le& c_{\ref{lem:iterdec}.2}\kappa \Bigg((K_{q,r}+\sqrt{pd\norm{Y}_{4q}})\norm{V} \\&+ \left(R_{1,2r}(V)+R_{2,2r}(V)\right)^{\alpha_1/2}K_{q,r}^{(\alpha_2+\min\{\frac{1}{2},\alpha_1\})/2}\norm{V}^{\alpha_3+\alpha_1/2} \norm{Y}_{4q}^{(\alpha_4+\max\{\frac{1}{2},\alpha_3\})/2}q^{\alpha_5/2}\Bigg)
\end{aligned}
\end{equation}

\end{lemma}

\begin{proof}[Proof of Lemma \ref{lem:iterdec}]

Assume that (\ref{eq:iterassum}) holds for some $\alpha_1>0,...,\alpha_5>0$.

\textbf{Step 1: decoupling.}

Consider an arbitrary $n \times d$ deterministic matrix $V$ and $r \ge q$. If $\norm{Y}_{4q} \le K_{q,r}$, then we have
\begin{align*}
    &\norm{VY}_{4q} \\\le& 
    \norm{V}\norm{Y}_{4q}
    \\ \le & K_{q,2r}\norm{V}
    \\ \le &\kappa \Bigg((K_{q,r}+\sqrt{pd\norm{Y}_{4q}})\norm{V} \\&+ \left(R_{1,2r}(V)+R_{2,2r}(V)\right)^{\alpha_1/2}K_{q,r}^{(\alpha_2+\min\{\frac{1}{2},\alpha_1\})/2}\norm{V}^{\alpha_3+\alpha_1/2} \norm{Y}_{4q}^{(\alpha_4+\max\{\frac{1}{2},\alpha_3\})/2}q^{\alpha_5/2}\Bigg)
\end{align*}
where we use $\kappa>1$, and therefore the desired statement (\ref{eq:iterconclu}) holds automatically.

Therefore, it suffices to only consider the case when $\norm{Y}_{4q} > K_{q,r}$.

We have
\begin{align*}
    \norm{VY}_{4q}=&(\E\tr((Y^TV^TVY)^{2q}))^{\frac{1}{4q}}
\\=&\norm{Y^TV^TVY}_{2q}^{1/2}
\end{align*}

By Lemma \ref{lem:decoupfirstlayer} and Lemma \ref{lem:decoupsecondlayer}, we know that
\begin{align*}
    \norm{Y_1^TV^TVY_1}_{2q} \le& c_1 K^2 \norm{V}^2 + 4pd\norm{V}^2 \norm{Y_1}_{2q} + 16 \norm{X_3^TX_5V^TVX_6^TX_4}_{2q}
    \\=& c_1 K^4 \norm{V}^2 + 4pd\norm{V}^2 \norm{Y_1}_{2q}+ 16 \norm{Y_3^TV^TVY_4}_{2q}
    \\ \le & c_1 K^4 \norm{V}^2 + 4pd\norm{V}^2 \norm{Y_1}_{4q}+ 16 \norm{Y_3^TV^TVY_4}_{2q}
\end{align*}

\textbf{Step 2: conditioning on $W=Y_3^TV^TV$.}

Now we will condition on $W=Y_3^TV^TV$ and use our assumption to bound $\norm{WY_4}_{2q}$ for fixed $W$.

Since $r \ge 2q \ge q$, by our assumption (\ref{eq:iterassum}), we have
    \begin{align*}
    \norm{WY}_{2q} \le & \kappa \Bigg((K_{q,r}+\sqrt{pd\norm{Y}_{2q}})\norm{W} \\&+ \left(R_{1,2r}(W)+R_{2,2r}(W)\right)^{\alpha_1}K_{q,r}^{\alpha_2}\norm{W}^{\alpha_3} \norm{Y}_{2q} ^{\alpha_4}q^{\alpha_5}\Bigg)
    \\\le & \kappa \Bigg((K_{q,r}+\sqrt{pd\norm{Y}_{4q}})\norm{W} \\&+ \left(R_{1,2r}(W)+R_{2,2r}(W)\right)^{\alpha_1}K_{q,r}^{\alpha_2}\norm{W}^{\alpha_3} \norm{Y}_{4q} ^{\alpha_4}q^{\alpha_5}\Bigg)
\end{align*}
where we use the fact $\norm{Y}_{2q} \le \norm{Y}_{4q}$ by H\"older inequality in $L_q(S_q^d)$.

Since $q \ge \log(d)$, we have $\norm{\norm{W}}_{L_{2q}(\Pb)}  \le d^{\frac{1}{\log d}} \norm{W}_{2q}= e \norm{W}_{2q}$. Therefore, by triangle inequality for $\norm{\cdot}_{L_{4q}(\Pb)}$, we have
    \begin{align*}
    &\norm{Y_3^TV^TVY}_{2q} \\\le& e\kappa \Bigg((K_{q,r}+\sqrt{pd\norm{Y}_{4q}})\norm{Y_3^TV^TV}_{2q} \\&+ \norml{\left(R_{1,2r}(Y_3^TV^TV)+R_{2,2r}(Y_3^TV^TV)\right)^{\alpha_1}K_{q,r}^{\alpha_2}\norm{Y_3^TV^TV}^{\alpha_3} \norm{Y}_{4q} ^{\alpha_4}q^{\alpha_5}}_{L_{2q}(\Pb)}\Bigg)
\end{align*}

We first estimate the first term
\begin{align*}
    (K_{q,r}+\sqrt{pd\norm{Y}_{4q}})\norm{Y_3^TV^TV}_{2q}
\end{align*}

By using assumption (\ref{eq:iterassum}) again, we have
\begin{align*}
    \norm{Y_3^TV^TV}_{2q}=&\norm{V^TVY_3}_{2q} \\\le& \kappa \Bigg((K_{q,r}+\sqrt{pd\norm{Y}_{2q}})\norm{V^TV} \\&+ \left(R_{1,2r}(V^TV)+R_{2,2r}(V^TV)\right)^{\alpha_1}K_{q,r}^{\alpha_2}\norm{V^TV}^{\alpha_3} \norm{Y}_{2q}^{\alpha_4}q^{\alpha_5}\Bigg)
    \\\le& \kappa \Bigg((K_{q,r}+\sqrt{pd\norm{Y}_{4q}})\norm{V}^2 \\&+ \left(R_{1,2r}(V)+R_{2,2r}(V)\right)^{\alpha_1}K_{q,r}^{\alpha_2}\norm{V}^{2\alpha_3+\alpha_1} \norm{Y}_{4q}^{\alpha_4}q^{\alpha_5}\Bigg)
\end{align*}
where we use the fact that
\begin{align*}
    R_{1,2r}(V^TV)=&(\sum\limits_{(l,\gamma) \in \Xi}(\norm{V^TVX_1^TZ_{2,(l,\gamma)}U}_{2r}^{2r}))^{1/(2r)}
    \\ \le &(\sum\limits_{(l,\gamma) \in \Xi}((\norm{V}\norm{VX_1^TZ_{2,(l,\gamma)}U}_{2r})^{2r}))^{1/(2r)}
    \\ = &\norm{V}(\sum\limits_{(l,\gamma) \in \Xi}(\norm{VX_1^TZ_{2,(l,\gamma)}U}_{2r}^{2r}))^{1/(2r)}
    \\=& \norm{V}R_{1,2r}(V)
\end{align*}
and
\begin{align*}
    R_{2,2r}(V^TV) \le & \norm{V}R_{2,2r}(V)
\end{align*}
by the same calculation.

Next, we bound the second term
\begin{align*}
    \norml{\left(R_{1,2r}(Y_3^TV^TV)+R_{2,2r}(Y_3^TV^TV)\right)^{\alpha_1}K_{q,r}^{\alpha_2}\norm{Y_3^TV^TV}^{\alpha_3} \norm{Y}_{4q} ^{\alpha_4}q^{\alpha_5}}_{L_{2q}(\Pb)}
\end{align*}

We have
\begin{align*}
    &\norml{\left(R_{1,2r}(Y_3^TV^TV)+R_{2,2r}(Y_3^TV^TV)\right)^{\alpha_1}K_{q,r}^{\alpha_2}\norm{Y_3^TV^TV}^{\alpha_3} \norm{Y}_{4q} ^{\alpha_4}q^{\alpha_5}}_{L_{2q}(\Pb)}
    \\=&K_{q,r}^{\alpha_2}\norm{Y}_{4q}^{\alpha_4}q^{\alpha_5}\norml{\left(R_{1,2r}(Y_3^TV^TV)+R_{2,2r}(Y_3^TV^TV)\right)^{\alpha_1}\norm{Y_3^TV^TV}^{\alpha_3}  }_{L_{2q}(\Pb)}
    \\ \le &K_{q,r}^{\alpha_2}\norm{Y}_{4q}^{\alpha_4}q^{\alpha_5}\norml{\left(R_{1,2r}(Y_3^TV^TV)+R_{2,2r}(Y_3^TV^TV)\right)^{\alpha_1}\norm{Y_3}^{\alpha_3} \norm{V}^{2 \alpha_3}  }_{L_{2q}(\Pb)}
    \\ \le &K_{q,r}^{\alpha_2}\norm{Y}_{4q}^{\alpha_4}\norm{V}^{2 \alpha_3}q^{\alpha_5}\norml{\left(R_{1,2r}(Y_3^TV^TV)+R_{2,2r}(Y_3^TV^TV)\right)^{\alpha_1}\norm{Y_3}^{\alpha_3}   }_{L_{2q}(\Pb)}
    \\ \le &K_{q,r}^{\alpha_2}\norm{Y}_{4q}^{\alpha_4}\norm{V}^{2 \alpha_3}q^{\alpha_5}\norml{\left(R_{1,2r}(Y_3^TV^TV)+R_{2,2r}(Y_3^TV^TV)\right)^{\alpha_1}}_{L_{4q}(\Pb)}\norml{\norm{Y_3}^{\alpha_3}   }_{L_{4q}(\Pb)}
\end{align*}
where in the last step, we use Holder's inequality.

We also have two observations. Since $r \ge 2q$ and $\alpha_1 \le 1$, we have
\begin{align*}
    &\norml{\left(R_{1,2r}(Y_3^TV^TV)+R_{2,2r}(Y_3^TV^TV)\right)^{\alpha_1}}_{L_{4q}(\Pb)}
    \\ \le & \norml{\left(R_{1,2r}(Y_3^TV^TV)+R_{2,2r}(Y_3^TV^TV)\right)^{\alpha_1}}_{L_{2r}(\Pb)}
    \\ \le & \norml{\left(R_{1,2r}(Y_3^TV^TV)+R_{2,2r}(Y_3^TV^TV)\right)}_{L_{2r}(\Pb)}^{\alpha_1}
\end{align*}
And since $\alpha_3 \le 1$, we have
\begin{align*}
    \norml{\norm{Y_3}^{\alpha_3}   }_{L_{4q}(\Pb)} \le c_2\norm{Y_3}_{4q}^{\alpha_3}
\end{align*}
for some constant $c_2$ by the fact that $q \ge \log(d)$ and Jensen's inequality.

To estimate $\norml{\left(R_{1,2r}(Y_3^TV^TV)+R_{2,2r}(Y_3^TV^TV)\right)}_{L_{2r}(\Pb)}^{\alpha_1}$, we need the following lemma. This step is main reason why this iterative method can improve the estimate on $\norm{Y}_{4q}$. Intuitively, in this step, we can take out a copy of $Y$ from $R_{1,2r}(Y_3^TV^TV)+R_{2,2r}(Y_3^TV^TV)$ and replace it by the correct factor $K_{q,r}$.

\begin{lemma}\label{lem:Rsimplify}
Let $R_{1,2r}(V), R_{2,2r}(V)$ and $K_{q,r}$ be as in Lemma \ref{lem:iterdec} with $r \ge2q$, then we have
\begin{align*}
    \norm{R_{1,2r}(Y_3^TV^TV)}_{L_{2r}(\Pb)}  
    \le &c_{\ref{lem:Rsimplify}}(R_{1,2r}(V))\norm{V}K_{q,r}
\end{align*}
and
\begin{align*}
    \norm{R_{2,2r}(Y_3^TV^TV)}_{L_{2r}(\Pb)}  
    \le &c_{\ref{lem:Rsimplify}}(R_{2,2r}(V))\norm{V}K_{q,r}
\end{align*}
\end{lemma}

\begin{proof}
By definition, we have
\begin{align*}
    R_{1,2r}(Y_3^TV^TV)=(\sum\limits_{(l,\gamma) \in \Xi}(\E_{X_1,X_2} \tr(|Y_3^TV^TVX_1^TZ_{2,(l,\gamma)}U|^{2r})))^{1/(2r)}
\end{align*}

Therefore, we have
\begin{align*}
    \norm{R_{1,2r}(Y_3^TV^TV)}_{L_{2r}(\Pb)}^{2r}  
    = & \E_{Y_3} \sum\limits_{(l,\gamma) \in \Xi}(\E_{X_1,X_2} \tr(|Y_3^TV^TVX_1^TZ_{2,(l,\gamma)}U|^{2r}))
    \\= &  \sum\limits_{(l,\gamma) \in \Xi}(\E_{X_1,X_2} \E_{Y_3}\tr(|Y_3^TV^TVX_1^TZ_{2,(l,\gamma)}U|^{2r}))
\end{align*}

Recall that $Z_{k,(l,\gamma)}=\xi_{k,(l,\gamma)}e_{\mu_{k,(l,\gamma)}}e_l^T$. Also, define the vector $u_l$ such that $u_l^T$ is the $l$th row of the matrix $U$. Therefore, we observe that
\begin{align*}
    &\tr(|Y_3^TV^TVX_1^TZ_{2,\eta_2}U|^{2r})
\\=&\tr(|Y_3^TV^TVX_1^T\xi_{2,\eta_2}e_{\mu_{2,\eta_2}}e_{\eta_2(1)}^TU|^{2r})
\\=&\tr(|Y_3^TV^TVX_1^T\xi_{2,\eta_2}e_{\mu_{2,\eta_2}}u_{\eta_2(1)}^T|^{2r})
\\=&\frac 1 d\norm{Y_3^TV^TVX_1^Te_{\mu_{2,\eta_2}}}_{l_2([d])}^{2r}\norm{u_{\eta_2(1)}}_{l_2([d])}^{2r}
\end{align*}

Therefore, we have
\begin{align*}
    &E_{Y_3}\tr(|Y_3^TV^TVX_1^TZ_{2,\eta_2}U|^{2r})
\\=&\frac 1 dE_{Y_3}\norm{Y_3^TV^TVX_1^Te_{\mu_{2,\eta_2}}}_{l_2([d])}^{2r}\norm{u_{\eta_2(1)}}_{l_2([d])}^{2r}
\\=&\frac 1 d\norm{u_{\eta_2(1)}}_{l_2([d])}^{2r}E_{Y_3}\norm{Y_3^TV^TVX_1^Te_{\mu_{2,\eta_2}}}_{l_2([d])}^{2r}\\=&\frac 1 d\norm{u_{\eta_2(1)}}_{l_2([d])}^{2r}E_{X_6}\E_{X_5}\norm{X_5^TX_6V^TVX_1^Te_{\mu_{2,\eta_2}}}_{l_2([d])}^{2r}
\end{align*}

By Lemma \ref{lem:x1x2unorm}, we have
\begin{align*}
&\E_{X_5,X_6}\norm{X_5^TX_6V^TVX_1^Te_{\mu_{2,\eta_2}}}_{l_2([d])}^{2r} \\\le& \left(\norm{V^TVX_1^Te_{\mu_{2,\eta_2}}}_{l_2([m])}c_{\ref{lem:x1x2unorm}}\paren*{\sqrt{pmpd}+ (pmd)^\frac{1}{q} (\sqrt{pd\cdot 2r} + 2r)}\right)^{2r}
\\\le&\norm{V^TVX_1^Te_{\mu_{2,\eta_2}}}_{l_2([m])}^{4q}c_1K_{q,r}^{2r}
\end{align*}
where we use
\begin{align*}
    \sqrt{pmpd}+ (pmd)^\frac{1}{q} (\sqrt{pdq} + q) \le c_1K_{q,r}
\end{align*}
for some constant $c_1$ just by definition of $K_{q,r}$.

Therefore, we have
\begin{align*}
    &E_{Y_3}\tr(|Y_3^TV^TVX_1^TZ_{2,\eta_2}U|^{2r})
\\=&\frac 1 d\norm{u_{\eta_1(1)}}_{l_2([d])}^{2r}\norm{V^TVX_1^Te_{\mu_{2,\eta_2}}}_{l_2([m])}^{2r}(c_{1}K_{q,r})^{2r}
\\ \le & \frac 1 d\norm{V}^{2r} \norm{u_{\eta_2(1)}}_{l_2([d])}^{2r}\norm{VX_1^Te_{\mu_{2,\eta_2}}}_{l_2([m])}^{2r}(c_{1}K_{q,r})^{2r}
\\=&\norm{V}^{2r} (c_{1}K_{q,r})^{2r} \tr(|VX_1^Te_{\mu_{2,\eta_2}}u_{\eta_2(1)}|^{2r})
\end{align*}

Therefore, we have
\begin{align*}
    \norm{R_{1,2r}(Y_3^TV^TV)}_{L_{2r}(\Pb)}^{2r}  
    = &  \sum\limits_{(l,\gamma) \in \Xi}(\E_{X_1,X_2} \E_{Y_3}\tr(|Y_3^TV^TVX_1^TZ_{2,(l,\gamma)}U|^{2r}))
    \\ \le &\sum\limits_{(l,\gamma) \in \Xi}( \norm{V}^{2r} (c_{1}K_{q,r})^{2r} \E_{X_1,X_2}\tr(|VX_1^Te_{\mu_{2,\eta_2}}u_{\eta_2(1)}|^{2r}))
\end{align*}

Recall that, by definition, we have
\begin{align*}
    (R_{1,2q}(V))^{2r}=&\sum\limits_{(l,\gamma) \in \Xi}(\norm{VX_1^TZ_{2,(l,\gamma)}U}_{2r}^{2r})
    \\=&\sum\limits_{(l,\gamma) \in \Xi}\E_{X_1,X_2}\tr(|VX_1^Te_{\mu_{2,\eta_2}}u_{\eta_2(1)}|^{2r})
\end{align*}

Taking $2r$-th root of both sides, we have
\begin{align*}
    &\norm{R_{1,2r}(Y_3^TV^TV)}_{L_{2r}(\Pb)}  
    \\ \le &\norm{V} (c_{1}K_{q,r})R_{1,2r}(V)
\end{align*}
which is what we want.

The proof for the inequality for $\norm{R_{2,2r}(Y_3^TV^TV)}_{L_{2r}(\Pb)}  $ follows exactly the same argument.

\end{proof}

Therefore, by Lemma \ref{lem:Rsimplify} and triangle inequality, we have
\begin{align*}
    &\norml{\left(R_{1,2r}(Y_3^TV^TV)+R_{2,2r}(Y_3^TV^TV)\right)^{\alpha_1}K_{q,r}^{\alpha_2}\norm{Y_3^TV^TV}^{\alpha_3} \norm{Y}_{4q} ^{\alpha_4}q^{\alpha_5}}_{L_{2q}(\Pb)}
    \\ \le &c_2K_{q,r}^{\alpha_2}\norm{Y}_{4q}^{\alpha_4+\alpha_3}\norm{V}^{2 \alpha_3}q^{\alpha_5}\norml{\left(R_{1,2r}(Y_3^TV^TV)+R_{2,2r}(Y_3^TV^TV)\right)}_{L_{2r}(\Pb)}^{\alpha_1}
    \\ \le & c_2 K_{q,r}^{\alpha_2+\alpha_1}\norm{Y}_{4q}^{\alpha_4+\alpha_3}\norm{V}^{2 \alpha_3+\alpha_1}q^{\alpha_5}(\left(R_{1,2r}(V)+R_{2,2r}(V)\right)^{\alpha_1})
\end{align*}

Therefore, we have
    \begin{align*}
    &\norm{Y_3^TV^TVY}_{2q} \\\le& e\kappa \Bigg((K_{q,r}+\sqrt{pd\norm{Y}_{4q}})\norm{Y_3^TV^TV}_{2q} \\&+ \norml{\left(R_{1,2r}(Y_3^TV^TV)+R_{2,2r}(Y_3^TV^TV)\right)^{\alpha_1}K_{q,r}^{\alpha_2}\norm{Y_3^TV^TV}^{\alpha_3} \norm{Y}_{4q} ^{\alpha_4}q^{\alpha_5}}_{L_{2q}(\Pb)}\Bigg)
    \\ \le & e\kappa \Bigg((K_{q,r}+\sqrt{pd\norm{Y}_{4q}}) \\& \cdot \kappa \Bigg((K_{q,r}+\sqrt{pd\norm{Y}_{4q}})\norm{V}^2 + \left(R_{1,2r}(V)+R_{2,2r}(V)\right)^{\alpha_1}K_{q,r}^{\alpha_2}\norm{V}^{2\alpha_3+\alpha_1} \norm{Y}_{4q}^{\alpha_4}q^{\alpha_5}\Bigg) \\&+ c_2 K_{q,r}^{\alpha_2+\alpha_1}\norm{Y}_{4q}^{\alpha_4+\alpha_3}\norm{V}^{2 \alpha_3+\alpha_1}q^{\alpha_5}(\left(R_{1,2r}(V)+R_{2,2r}(V)\right)^{\alpha_1})\Bigg)
    \\ = & e\kappa^2 (K_{q,r}+\sqrt{pd\norm{Y}_{4q}})^2 \norm{V}^2 \\&+ \kappa^2 (K_{q,r}+\sqrt{pd\norm{Y}_{4q}})\left(R_{1,2r}(V)+R_{2,2r}(V)\right)^{\alpha_1}K_{q,r}^{\alpha_2}\norm{V}^{2\alpha_3+\alpha_1} \norm{Y}_{4q}^{\alpha_4}q^{\alpha_5} \\& + c_3 \kappa K_{q,r}^{\alpha_2+\alpha_1}\norm{Y}_{4q}^{\alpha_4+\alpha_3}\norm{V}^{2 \alpha_3+\alpha_1}q^{\alpha_5}(\left(R_{1,2r}(V)+R_{2,2r}(V)\right)^{\alpha_1})
\\ \le  & e\kappa^2 (K_{q,r}+\sqrt{pd\norm{Y}_{4q}})^2 \norm{V}^2 \\&+ \kappa^2 \left(R_{1,2r}(V)+R_{2,2r}(V)\right)^{\alpha_1}K_{q,r}^{\alpha_2+\frac{1}{2}}\norm{V}^{2\alpha_3+\alpha_1} \norm{Y}_{4q}^{\alpha_4+\frac{1}{2}}q^{\alpha_5} \\& + c_3 \kappa K_{q,r}^{\alpha_2+\alpha_1}\norm{Y}_{4q}^{\alpha_4+\alpha_3}\norm{V}^{2 \alpha_3+\alpha_1}q^{\alpha_5}(\left(R_{1,2r}(V)+R_{2,2r}(V)\right)^{\alpha_1})
\\ \le  & c_3\kappa^2 \big((K_{q,r}+\sqrt{pd\norm{Y}_{4q}})^2 \norm{V}^2 \\&+ \left(R_{1,2r}(V)+R_{2,2r}(V)\right)^{\alpha_1}K_{q,r}^{\alpha_2+\min\{\frac{1}{2},\alpha_1\}}\norm{V}^{2\alpha_3+\alpha_1} \norm{Y}_{4q}^{\alpha_4+\max\{\frac{1}{2},\alpha_3\}}q^{\alpha_5} \big)
\end{align*}
for some constant $c_3>0$,
where we use the assumption that $\norm{Y}_{4q}>K_{q,r}$ and the fact that $pd \le K_{q,r}$ to get
\begin{align*}
    K_{q,r}+\sqrt{pd\norm{Y}_{4q}} \le (\sqrt{K_{q,r}} + \sqrt{pd})\sqrt{\norm{Y}_{4q}} \le 2 \sqrt{K_{q,r}}\sqrt{\norm{Y}_{4q}}
\end{align*}
and we combine this with the fact that $\alpha_1+\alpha_3=1$ to conclude that the factor $\max \{ K_{q,r}^{\frac{1}{2}} \norm{Y}_{4q}^{\frac{1}{2}}, K_{q,r}^{\alpha_1} \norm{Y}_{4q}^{\alpha_3} \}$
is the first term when $\alpha_3<1/2 \text{ and } \alpha_1 >1/2$ and the second term otherwise.

\textbf{Step 3: final estimates.}

Therefore, we have
\begin{align*}
    &\norm{Y_1^TV^TVY_1}_{2q}  \\\le & c_1 K_{q,r}^4 \norm{V}^2 + 4pd\norm{V}^2 \norm{Y_1}_{4q}+ 16 \norm{Y_3^TV^TVY_4}_{2q}
    \\ \le & c_4 \kappa^2 \big((K_{q,r}+\sqrt{pd\norm{Y}_{4q}})^2 \norm{V}^2 \\&+ \left(R_{1,2r}(V)+R_{2,2r}(V)\right)^{\alpha_1}K_{q,r}^{\alpha_2+\min\{\frac{1}{2},\alpha_1\}}\norm{V}^{2\alpha_3+\alpha_1} \norm{Y}_{4q}^{\alpha_4+\max\{\frac{1}{2},\alpha_3\}}q^{\alpha_5} \big)
\end{align*}
for some constant $c_4$.

Therefore, we have
\begin{align*}
    &\norm{VY}_{4q}
\\=&\norm{Y^TV^TVY}_{2q}^{1/2}
 \\ \le & c_5 \kappa \big((K_{q,r}+\sqrt{pd\norm{Y}_{4q}}) \norm{V} \\&+ \left(R_{1,2r}(V)+R_{2,2r}(V)\right)^{\alpha_1/2}K_{q,r}^{(\alpha_2+\min\{\frac{1}{2},\alpha_1\})/2}\norm{V}^{\alpha_3+\alpha_1/2} \norm{Y}_{4q}^{(\alpha_4+\max\{\frac{1}{2},\alpha_3\})/2}q^{\alpha_5/2} \big)
\end{align*}
for some constant $c_5$.

\end{proof}

By applying Lemma \ref{lem:iterdec} iteratively, we can obtain the following sequence of improved moment estimates.

\begin{lemma}[Sequence of Moment Estimates]\label{lem:univseq}
Let $S$, $p$, $\{Z_{(l,\gamma)}\}_{(l,\gamma) \in \Xi}$, $\{(S_{\lambda},\{Z_{\lambda,(l,\gamma)}\}_{(l,\gamma) \in \Xi})\}_{\lambda=0,1,2,...}$, $U$, $X$, $X_1,X_2,...$, $Y_1,Y_2,...$, $R_{1,2r}(V)$, $R_{1,2r}(V)$, and $K_{q,r}$ be as in Lemma \ref{lem:iterdec}
Let $q_0 > \log(d)$.
There exists a constant $c_{\ref{lem:univseq}}$.
such that for any integers $k \ge 1$, an $n \times d$ deterministic matrix $V$, and $r>2q_0\cdot 2^k$, we have
\begin{equation}\label{eq:univseq}
    \begin{aligned}
    &\norm{VY}_{2q_0\cdot 2^k} \\\le& c_{\ref{lem:univseq}}^k \Bigg((K_{q_0 \cdot 2^k,r}+\sqrt{pd\norm{Y}_{2q_0\cdot 2^k}})\norm{V} \\&+ \left(R_{1,2r}(V)+R_{2,2r}(V)\right)^{\alpha_1(k)}K_{q_0 \cdot 2^k,r}^{\alpha_2(k)}\norm{V}^{\alpha_3(k)} \norm{Y}_{2q_0 \cdot 2^k}^{\alpha_4(k)}q^{\alpha_5(k)}\Bigg)
\end{aligned}
\end{equation}
where $\alpha_1(k)=\frac{1}{2^k}$, $\alpha_2(k)=\frac{2k-1}{2^{k+1}}$, $\alpha_3(k)=1-\frac{1}{2^k}$, $\alpha_4(k)=1-\frac{2k+1}{2^{k+1}}$, $\alpha_5(2k)=\frac{2}{2^k}$.

\end{lemma}

\begin{proof}

Define the sequence $\alpha_1(k),...,\alpha_5(k)$ with the initial condition $\alpha_1(0)=1,\alpha_5(0)=2$, $\alpha_2(0)=\alpha_3(0)=\alpha_4(0)=0$ and the recurrence relations
\begin{align*}
\alpha_1(k+1)&=\frac{\alpha_1(k)}{2}, \\[0.5em]
\alpha_2(k+1)&=\frac{\alpha_2(k)+\min\{\tfrac12,\;\alpha_1(k)\}}{2}, \\[0.5em]
\alpha_3(k+1)&=\alpha_3(k)+\frac{\alpha_1(k)}{2}, \\[0.5em]
\alpha_4(k+1)&=\frac{\alpha_4(k)+\max\{\tfrac12,\;\alpha_3(k)\}}{2}, \\[0.5em]
\alpha_5(k+1)&=\frac{\alpha_5(k)}{2}.
\end{align*}
for $k \ge 0$.

Standard calculation shows that the closed form solution for $\alpha_1(k),...,\alpha_5(k)$ is $\alpha_1(k)=\frac{1}{2^k}$, $\alpha_2(k)=\frac{2k-1}{2^{k+1}}$, $\alpha_3(k)=1-\frac{1}{2^k}$, $\alpha_4(k)=1-\frac{2k+1}{2^{k+1}}$, $\alpha_5(2k)=\frac{2}{2^k}$ for $k \ge 1$.

    The proof of the inequality \ref{eq:univseq} follows from using Lemma \ref{lem:iterdec} with induction.
    
    More specifically, let $c_{\ref{lem:univseq}}=c_{\ref{lem:initialuniv}} \cdot c_{\ref{lem:iterdec}.2}$ and then we start from Lemma \ref{lem:initialuniv} to conclude that the statement \ref{eq:univseq} holds for $\alpha_1(0)=1,\alpha_5(0)=2$ and $\alpha_2(0)=\alpha_3(0)=\alpha_4(0)=0$.

    Next, assume that the statement \ref{eq:univseq} holds for $k$, then we have
    \begin{align*}
        &\norm{VY}_{2q_0\cdot 2^k} \\\le& c_{\ref{lem:univseq}}^k \Bigg((K_{q_0,r}+\sqrt{pd\norm{Y}_{2q_0\cdot 2^k}})\norm{V} \\&+ \left(R_{1,2r}(V)+R_{2,2r}(V)\right)^{\alpha_1(k)}K_{q_0,r}^{\alpha_2(k)}\norm{V}^{\alpha_3(k)} \norm{Y}_{2q_0 \cdot 2^k}^{\alpha_4(k)}q^{\alpha_5(k)}\Bigg)
    \end{align*}

    Then, we use Lemma \ref{lem:iterdec} with $q=q_0 \cdot 2^k$ and $\kappa=2q_0\cdot 2^{k}$ to claim for any $r>2q=q_0 \cdot 2^{k+1}$
    \begin{align*}
    &\norm{VY}_{2q_0\cdot 2^{k+1}}\\=&\norm{VY}_{4q} \\\le& c_{\ref{lem:iterdec}.2}c_{\ref{lem:univseq}}^k \Bigg((K_{q_0 \cdot 2^k,r}+\sqrt{pd\norm{Y}_{4q}})\norm{V} \\&+ \left(R_{1,2r}(V)+R_{2,2r}(V)\right)^{\alpha_1(k)/2} \\& \cdot K_{q,r}^{(\alpha_2(k)+\min\{\frac{1}{2},\alpha_1(k)\})/2}\norm{V}^{\alpha_3(k)+\alpha_1/2} \norm{Y}_{4q}^{(\alpha_4(k)+\max\{\frac{1}{2},\alpha_3(k)\})/2}q^{\alpha_5(k)/2}\Bigg)
    \\ \le & c_{\ref{lem:univseq}}^{k+1} \Bigg((K_{q_0 \cdot 2^k,r}+\sqrt{pd\norm{Y}_{2q_0\cdot 2^{k+1}}})\norm{V} \\&+ \left(R_{1,2r}(V)+R_{2,2r}(V)\right)^{\alpha_1(k)/2} \\& \cdot K_{q,r}^{(\alpha_2(k)+\min\{\frac{1}{2},\alpha_1(k)\})/2}\norm{V}^{\alpha_3(k)+\alpha_1/2} \norm{Y}_{2q_0\cdot 2^{k+1}}^{(\alpha_4(k)+\max\{\frac{1}{2},\alpha_3(k)\})/2}q^{\alpha_5(k)/2}\Bigg)
    \\ \le & c_{\ref{lem:univseq}}^{k+1} \Bigg((K_{q_0,r}+\sqrt{pd\norm{Y}_{2q_0\cdot 2^{k+1}}})\norm{V} \\&+ \left(R_{1,2r}(V)+R_{2,2r}(V)\right)^{\alpha_1(k+1)}K_{q_0,r}^{\alpha_2(k+1)}\norm{V}^{\alpha_3(k+1)} \norm{Y}_{2q_0 \cdot 2^{k+1}}^{\alpha_4(k+1)}q^{\alpha_5(k+1)}\Bigg)
\end{align*}

    So the statement holds for $k+1$ because $\alpha_1(k+1), \cdots ,\alpha_5(k+1)$ are determined by the recurrence relation.
\end{proof}

In the next lemma, we plug in $V=I_d$ into the final moment estimate for $VY$ (\ref{eq:univseq}). Since $\norm{Y}_{2q} \le \norm{Y}_{2q \cdot 2^k}$, we can cancel out some powers of $\norm{Y}_{2q \cdot 2^k}$ on both sides and then get an absolute bound for $\norm{Y}_{2q \cdot 2^k}$.

\begin{lemma}[Averaged Moment Estimates]\label{lem:averuniv}
Let $S$, $p$, $\{Z_{(l,\gamma)}\}_{(l,\gamma) \in \Xi}$, $\{(S_{\lambda},\{Z_{\lambda,(l,\gamma)}\}_{(l,\gamma) \in \Xi})\}_{\lambda=0,1,2,...}$, $U$, $X$, $X_1,X_2,...$, $Y_1,Y_2,...$, $R_{1,2r}(V)$, $R_{1,2r}(V)$, and $K_{q,r}$ be as in Lemma \ref{lem:iterdec}.
Then there exist  constants $c_{\ref{lem:averuniv}.1},c_{\ref{lem:averuniv}.2}$
such that for any integers $k \ge 1$, $q > \log(d)$, and an $n \times d$ deterministic matrix $V$, we have
\begin{equation}\label{eq:averuniv}
    \begin{aligned}
    \norm{Y}_{2q \cdot 2^k} \le \exp(\exp(c_{\ref{lem:averuniv}.1}k))\left(R_{1,2q \cdot 2^k}(I_d)+R_{2,2q \cdot 2^k}(I_d)\right)^{\frac{2}{2k+1}}K_{q,q}^{\frac{2k-1}{2k+1}} q^{\frac{4}{2k+1}}+c_{\ref{lem:averuniv}.2}^k \cdot K_{q,q}
\end{aligned}
\end{equation}

\end{lemma}

\begin{proof}
For convenience, let $K=K_{q,q}$. By direct calculation, we have
\begin{align*}
    K_{q \cdot 2^k,q \cdot 2^k} \le 2^k K_{q,q}=2^k K
\end{align*}
By Lemma \ref{lem:univseq}, choosing $r=q \cdot 2^k $, we have
\begin{align*}
    &\norm{VY}_{2q \cdot 2^k} \\\le& c_{\ref{lem:univseq}}^k \Bigg((K_{q \cdot 2^k,q \cdot 2^k}+\sqrt{pd\norm{Y}_{2q \cdot 2^k}})\norm{V} \\&+ \left(R_{1,2q \cdot 2^k}(V)+R_{2,2q \cdot 2^k}(V)\right)^{\alpha_1(k)}K_{q \cdot 2^k,q \cdot 2^k}^{\alpha_2(k)}\norm{V}^{\alpha_3(k)} \norm{Y}_{2q \cdot 2^k}^{\alpha_4(k)}q^{\alpha_5(k)}\Bigg)
    \\ \le & c_1^k \Bigg((K+\sqrt{pd\norm{Y}_{2q \cdot 2^k}})\norm{V} \\&+ \left(R_{1,2q \cdot 2^k}(V)+R_{2,2q \cdot 2^k}(V)\right)^{\alpha_1(k)}K^{\alpha_2(k)}\norm{V}^{\alpha_3(k)} \norm{Y}_{2q \cdot 2^k}^{\alpha_4(k)}q^{\alpha_5(k)}\Bigg)
\end{align*}
for some new constant $c_1>1$ where $\alpha_1(k),...,\alpha_5(k)$ is $\alpha_1(k)=\frac{1}{2^k}$, $\alpha_2(k)=\frac{2k-1}{2^{k+1}}$, $\alpha_3(k)=1-\frac{1}{2^k}$, $\alpha_4(k)=1-\frac{2k+1}{2^{k+1}}$, $\alpha_5(2k)=\frac{2}{2^k}$ for $k \ge 1$.

Now, we set $V=I_d$ and therefore the above inequality becomes
\begin{align*}
    &\norm{Y}_{2q \cdot 2^k} \\\le& c_{1}^k \Bigg((K+\sqrt{pd\norm{Y}_{2q \cdot 2^k}}) \\&+ \left(R_{1,2q \cdot 2^k}(I_d)+R_{2,2q \cdot 2^k}(I_d)\right)^{\alpha_1(k)}K^{\alpha_2(k)}\norm{Y}_{2q \cdot 2^k}^{\alpha_4(k)}q^{\alpha_5(k)}\Bigg)
\end{align*}

Therefore, we have
\begin{align*}
    \norm{Y}_{2q\cdot 2^k} \le 3c_{1}^k K
\end{align*}
or
\begin{align*}
    \norm{Y}_{2q\cdot 2^k} \le 3c_{1}^k\sqrt{pd\norm{Y}_{2q\cdot 2^k}}
\end{align*}
or
\begin{align*}
    \norm{Y}_{2q \cdot 2^k} \le 3c_{1}^k\left(R_{1,2r}(I_d)+R_{2,2r}(I_d)\right)^{\alpha_1(k)}K^{\alpha_2(k)} \norm{Y}_{2q \cdot 2^k}^{\alpha_4(k)}q^{\alpha_5(k)}
\end{align*}

In the first and second cases, the desired result holds directly.

It suffices to consider the third case
\begin{align*}
    \norm{Y}_{2q\cdot 2^k} \le 3c_{1}^k\left(R_{1,2q \cdot 2^k}(I_d)+R_{2,2q \cdot 2^k}(I_d)\right)^{\frac{1}{2^k}}K^{\frac{2k-1}{2^{k+1}}} \norm{Y}_{2q\cdot 2^k}^{1-\frac{2k+1}{2^{k+1}}}q^{\frac{2}{2^k}}
\end{align*}

Therefore, we have
\begin{align*}
    \norm{Y}_{2q \cdot 2^k}^{\frac{2k+1}{2^{k+1}}} \le 3c_{1}^k\left(R_{1,2q \cdot 2^k}(I_d)+R_{2,2q \cdot 2^k}(I_d)\right)^{\frac{1}{2^k}}K^{\frac{2k-1}{2^{k+1}}} q^{\frac{2}{2^k}}
\end{align*}

Taking power $2^{k+1}$ on both sides, we have
\begin{align*}
    \norm{Y}_{2q \cdot 2^k}^{2k+1} \le (3c_{1}^k)^{2^{k+1}}\left(R_{1,2q \cdot 2^k}(I_d)+R_{2,2q \cdot 2^k}(I_d)\right)^{2}K_{q,r}^{2k-1} q^{4}
\end{align*}

Taking $2k+1$th root on both sides, we have
\begin{align*}
    \norm{Y}_{2q \cdot 2^k} \le &(3c_{1}^k)^{\frac{2^{k+1}}{2k+1}}\left(R_{1,2q \cdot 2^k}(I_d)+R_{2,2q \cdot 2^k}(I_d)\right)^{\frac{2}{2k+1}}K^{\frac{2k-1}{2k+1}} q^{\frac{4}{2k+1}}
    \\\le & \exp(\exp(c_2k))\left(R_{1,2q \cdot 2^k}(I_d)+R_{2,2q \cdot 2^k}(I_d)\right)^{\frac{2}{2k+1}}K^{\frac{2k-1}{2k+1}} q^{\frac{4}{2k+1}}
\end{align*}
for some constant $c_2$.

\end{proof}

\section{Final Results for OSE.} \label{sec:finalres}

We can combine the above final moment estimate with Lemma \ref{lem:decoup} to obtain the bound for the trace moments of the embedding error for OSNAP.

\begin{theorem}[Trace Moments of Embedding Error for OSNAP] \label{prop:momestdecoupmatrix}
Let $S$, $p$, $\{Z_{(l,\gamma)}\}_{(l,\gamma) \in \Xi}$,\\ $\{(S_{\lambda},\{Z_{\lambda,(l,\gamma)}\}_{(l,\gamma) \in \Xi})\}_{\lambda=0,1,2,...}$, $U$, $X$, $X_1,X_2,...$, $Y_1,Y_2,...$, $R_{1,2r}(V)$, and $R_{1,2r}(V)$ be as in Lemma \ref{lem:iterdec}.
Let
\begin{align*}
    K=&K(m,d,p,q)\\=&\left(p\max\{m,q\}p\max\{d,q\}+(pm)^{1/(q)}q(pd+q) \right)^{1/2}
\end{align*}
Then, there exist constants $c_{\ref{prop:momestdecoupmatrix}.1}> 0, c_{\ref{prop:momestdecoupmatrix}.2}>0$ such that for $k \in \N$, $q \in \N$, and $\varepsilon>0$ satisfying $\log (d) \le q \le m$ we have
\begin{align*}
    \E[\tr(\frac{1}{pm}X^TX - I_d)^{2q}]^\frac{1}{2q} \leq  \varepsilon
\end{align*}
when either of the following conditions are satisfied

(1) $k=0$ and
\begin{align*}
c_{\ref{prop:momestdecoupmatrix}.1}\left(\left(R_{1,2q }(I_d)+R_{2,2q }(I_d)\right)q^{2}+K\right) \le pm \varepsilon
\end{align*} 

(2) $k \ge 1$ and 
and \begin{align*}
\exp(\exp(c_{\ref{prop:momestdecoupmatrix}.2}k))\left(\left(R_{1,2q \cdot 2^k}(I_d)+R_{2,2q \cdot 2^k}(I_d)\right)^{\frac{2}{2k+1}}K^{\frac{2k-1}{2k+1}} q^{\frac{4}{2k+1}}+K\right) \le pm \varepsilon
\end{align*} 

\end{theorem}

\begin{proof}
The desired bound
\begin{align*}
    \E[\tr(\frac{1}{pm}X^TX - I_d)^{2q}]^\frac{1}{2q} \leq  \varepsilon
\end{align*}
is equivalent to
\begin{align*}
    \E[\tr(X^TX - pmI_d)^{2q}]^\frac{1}{2q} \leq  pm\varepsilon
\end{align*}

By Lemma \ref{lem:decoup}, we know that
\begin{align*}
    \E[\tr(X^TX - pmI_d)^{2q}]^\frac{1}{2q} \le c_1 \norm{X_1^TX_2}_{2q}
\end{align*}
for some constant $c_1$ where $X_1,X_2$ are i.i.d. copies of $X$.

When $k=0$, we can use Lemma \ref{lem:initialuniv} to get
\begin{align*}
    \E[\tr(X^TX - pmI_d)^{2q}]^\frac{1}{2q} \le &c_1 \norm{X_1^TX_2}_{2q} \\\le&  
    c_{2}(\sqrt{p\max\{m,q\}p\max\{d,q\}} + q^{2}R_{2,2r}(I_d) + q^{2}R_{1,2r}(I_d))
    \\\le&  
    c_{2}((R_{2,2q}(I_d) + R_{1,2q}(I_d))q^2+K)
\end{align*}
for some constant $c_2$.

Therefore, it suffices to require
\begin{align*}
    c_{2}((R_{2,2q}(I_d) + R_{1,2q}(I_d))q^2+K) \le pm \varepsilon
\end{align*}

When $k \ge 1$, by Lemma \ref{lem:averuniv}, we know that either $\norm{X_1^TX_2}_{2q} \le K$
or we have $\norm{X_1^TX_2}_{4q} \ge \norm{X_1^TX_2}_{2q} > K$ and therefore we have
\begin{align*}
    &\norm{X_1^TX_2}_{2q} \\\le& \norm{X_1^TX_2}_{2q \cdot 2^k} \\ \le & \exp(\exp(c_3k))\left(R_{1,2q \cdot 2^k}(I_d)+R_{2,2q \cdot 2^k}(I_d)\right)^{\frac{2}{2k+1}}K^{\frac{2k-1}{2k+1}} q^{\frac{4}{2k+1}}+c_4^k \cdot K
\end{align*}
for some constant $c_3$ and $c_4$.

In both cases, we have
\begin{align*}
    &\E[\tr(X^TX - pmI_d)^{2q}]^\frac{1}{2q} 
    \\\le& c_1 \norm{X_1^TX_2}_{2q}\\\le& \exp(\exp(c_5k))(\left(R_{1,2q \cdot 2^k}(I_d)+R_{2,2q \cdot 2^k}(I_d)\right)^{\frac{2}{2k+1}}K^{\frac{2k-1}{2k+1}} q^{\frac{4}{2k+1}}+K)
\end{align*}
for some constant $c_5$.

Therefore, it suffices to require
\begin{align*}
    \exp(\exp(c_5k))\left(\left(R_{1,2q \cdot 2^k}(I_d)+R_{2,2q \cdot 2^k}(I_d)\right)^{\frac{2}{2k+1}}K^{\frac{2k-1}{2k+1}} q^{\frac{4}{2k+1}}+K\right) \le pm \varepsilon
\end{align*}

\end{proof}

Combining the above result with Lemma \ref{lem:R1R2}, we can get explicit requirements on embedding dimension and sparsity to ensure OSE moment bound to hold.

\begin{corollary}[Sparsity Requirement for Moment Bound]\label{cor:momenttheta}
Let $S$ be an $m \times n$ matrix distributed according to the fully independent unscaled OSNAP distribution with parameter $p$. Let $0 < \delta,  \varepsilon < 1$ and $d>10$. Let $U$ be an arbitrary $n \times d$ deterministic matrix such that $U^TU=I_d$.
Let $X=SU$.
There exists a constants $c_{\ref{cor:momenttheta}.1},c_{\ref{cor:momenttheta}.2}$ such that the following holds. Let $q=\log(d/\delta)$ and $k \ge 1$ be an arbitrary integer.
Assume that we choose $m=\theta\frac{d+\log(d/\delta)}{\varepsilon^2}$ with $\theta\ge 16\exp(\exp(c_{\ref{cor:momenttheta}.2}k))^2$. Then we have
\begin{equation}\label{eq:trmomose}
    \begin{aligned}
    \E[\tr(\frac{1}{pm}X^TX - I_d)^{2q}]^\frac{1}{2q} \leq  \varepsilon
\end{aligned}
\end{equation}
as long as
\begin{align*}
    \begin{aligned}
    {pm } \ge {\exp(\exp(c_{\ref{cor:momenttheta}.1}k))}\left(\frac{1}{\varepsilon^{1+\frac{1}{2q-1}}}(\theta q+\frac{q^{5/2}}{\theta^{k/2-1/4}})+\frac{ q^{4}}{\theta^{k+1/2}}\right)
\end{aligned}
\end{align*}

In particular, there are two special cases of interest.

(1) We can choose $\theta=c_{\ref{cor:momenttheta}.3} q^{\frac{5}{k-1/2}}$ such that $m=c_{\ref{cor:momenttheta}.3}\frac{\log(d/\delta)^{\frac{5}{k-1/2}}(d+\log(d/\delta))}{\varepsilon^2}$ and the OSE moment bound (\ref{eq:trmomose}) holds when
\begin{align*}
     pm \ge c_{\ref{cor:momenttheta}.4} \frac{q^{1+{\frac{5}{k - 1/2} + \frac{1}{k}}}}{\varepsilon^{1+\frac{1}{2q-1}}}
\end{align*}
where we define the function $h(x)=\max\{\log(x),1\}$, and let
\begin{align*}
    k=k(d/\delta)=\lceil c_{\ref{cor:momenttheta}.5} h \circ h \circ h (d/\delta)\rceil
\end{align*}

(2) We can choose $\theta={O(1)}$ such that $m=\frac{c_{\ref{cor:momenttheta}.6}(d+\log(d/\delta))}{\varepsilon^2}$ and the OSE moment bound (\ref{eq:trmomose}) holds 
when 
\begin{align*}
     pm \ge c_{\ref{cor:momenttheta}.7}\left(\frac{q^{5/2}}{\varepsilon^{1+\frac{1}{q}}}+q^4\right)
\end{align*}
for some absolute constants $c_{\ref{cor:momenttheta}.3}$, $c_{\ref{cor:momenttheta}.4}$, $c_{\ref{cor:momenttheta}.5}$, $c_{\ref{cor:momenttheta}.6}$ and $c_{\ref{cor:momenttheta}.7}$.
\end{corollary}

\begin{proof}

We first prove the general result. Let
\begin{align*}
    K=&K(m,d,p,q)\\=&\left(p\max\{m,q\}p\max\{d,q\}+(pm)^{1/(q)}q(pd+q) \right)^{1/2}
\end{align*}
as in Theorem \ref{prop:momestdecoupmatrix}.
When $k \ge 1$, by Theorem \ref{prop:momestdecoupmatrix}, it is sufficient to require \begin{align*}
\exp(\exp(c_{\ref{prop:momestdecoupmatrix}.2}k))(\left(R_{1,2q \cdot 2^k}(I_d)+R_{2,2q \cdot 2^k}(I_d)\right)^{\frac{2}{2k+1}}K^{\frac{2k-1}{2k+1}} q^{\frac{4}{2k+1}}+K) \le pm \varepsilon
\end{align*}

Assume that $m=16\exp(\exp(c_{\ref{prop:momestdecoupmatrix}.2}k))^2\phi(d)\frac{d+q}{\varepsilon^2}$ where $\phi(d) \ge 1$. Then we need
\begin{align*}
\left(K+\left(R_{1,2q \cdot 2^k}(V)+R_{2,2q \cdot 2^k}(V)\right)^{\frac{2}{2k+1}}K^{\frac{2k-1}{2k+1}} q^{\frac{4}{2k+1}}\right) \le 4\sqrt{\phi(d)}\sqrt{pmp(d+q)}
\end{align*}

Then, it suffices to require
\begin{align*}
    K \le 2\sqrt{pmp(d+q)}
\end{align*}
and
\begin{align*}
    \left(R_{1,2q \cdot 2^k}(V)+R_{2,2q \cdot 2^k}(V)\right)^{\frac{2}{2k+1}}K^{\frac{2k-1}{2k+1}} q^{\frac{4}{2k+1}}\le 2K \sqrt{\phi(d)}
\end{align*}

These requirements are equivalent to
\begin{align*}
    p\max\{m,q\}p\max\{d,q\}+(pm)^{1/(q)}q(pd+q)  &\le 4pmp(d+q) \\
    \text{or, } pmp\max\{d,q\}+(pm)^{1/(q)}q(pd+q)  &\le 4pmp(d+q) \quad \text{ (since $m \ge q$)} \\
    \text{or, } (pm)^{1/(q)}pdq + (pm)^{1/(q)}q^2  &\le 3pmp(d+q) 
\end{align*}
and
\begin{align*}
    \left(R_{1,2q \cdot 2^k}(V)+R_{2,2q \cdot 2^k}(V)\right)^{\frac{2}{2k+1}} q^{\frac{4}{2k+1}}\le 2K^{\frac{2}{2k+1}}\sqrt{\phi(d)}
\end{align*}

Using the identity $m\varepsilon^2=16\exp(\exp(c_{\ref{lem:averuniv}}k))^2\phi(d){(d+q)}$, it suffices to require
\begin{align*}
    (pm)^{\frac{q-1}{q}} \ge q &\quad \text{( to get }  (pm)^{1/(q)}pdq \le pmpd \le pmp(d+q) \text{)}\\
    (pm\varepsilon)^2 \ge 16\exp(\exp(c_{\ref{lem:averuniv}}k))^2\phi(d)(pm)^{1/(q)}q^2 &\quad \text{( to get } (pm)^{1/(q)}q^2 \le pmp(d+q) \text{)}
\end{align*}
and
\begin{align*}
    2^{2k+1}K^2 \ge \frac{\left(R_{1,2q \cdot 2^k}(V)+R_{2,2q \cdot 2^k}(V)\right)^2 q^{4}}{\phi(d)^{k+1/2}}
\end{align*}

After simplification and using the fact that $K^2 \ge \frac{1}{2}pmp(d+q)=\frac{(pm \varepsilon)^2}{32\exp(\exp(c_{\ref{lem:averuniv}}k))^2\phi(d)}$, it suffices to require
\begin{align*}
    pm \ge q^{1+\frac{1}{q-1}}, \quad pm \ge (4\exp(\exp(c_{\ref{lem:averuniv}}k))\sqrt{\phi(d)}\frac{q}{\varepsilon})^{1+\frac{1}{2q-1}}
\end{align*}
and
\begin{align*}
    \frac{(pm \varepsilon)^2}{32\exp(\exp(c_{\ref{lem:averuniv}}k))^2\phi(d)} \ge \frac{\left(R_{1,2q \cdot 2^k}(V)+R_{2,2q \cdot 2^k}(V)\right)^2 q^{4}}{(\phi(d))^{k+1/2}}
\end{align*}

After further simplification, it suffices to require
\begin{align*}
    pm \ge \exp(\exp(c_{2}k))(\phi(d))^{1/2+\frac{1}{4q-2}}\frac{q^{1+\frac{1}{2q-1}}}{\varepsilon^{1+\frac{1}{2q-1}}}
\end{align*}
and
\begin{align*}
    {pm } \ge {\exp(\exp(c_3k))}\frac{\left(R_{1,2q \cdot 2^k}(V)+R_{2,2q \cdot 2^k}(V)\right) q^{2}}{(\phi(d))^{k/2-1/4}\varepsilon}
\end{align*}
for some constants $c_2$ and $c_3$.

When $m=\theta(d/\delta)\frac{d+q}{\varepsilon^2}$ with $16\exp(\exp(c_{\ref{lem:averuniv}}k))^2 \le \theta(d/\delta)$, we can choose $\phi(d)=\frac{\theta(d/\delta)}{16\exp(\exp(c_{\ref{lem:averuniv}}k))^2}$ and we have $\phi(d) \ge 1$ by direct calculation. 

Therefore, it suffices to require
\begin{align*}
    pm \ge \exp(\exp(c_{4}k))(\theta(d/\delta))^{1/2+\frac{1}{4q-2}}\frac{q^{1+\frac{1}{2q-1}}}{\varepsilon^{1+\frac{1}{2q-1}}}
\end{align*}
and
\begin{align*}
    {pm } \ge {\exp(\exp(c_4k))}\frac{\left(R_{1,2q \cdot 2^k}(V)+R_{2,2q \cdot 2^k}(V)\right) q^{2}}{(\theta(d/\delta))^{k/2-1/4}\varepsilon}
\end{align*}
for some constant $c_4>0$. 

Recall that
\begin{align*}
    R_{1,2q \cdot 2^k}(I_d) \le (pm)^{1/(2q \cdot 2^k)}(c_{\ref{lem:R1R2}.1}\sqrt{pd+2q })
    \le 2^k(pm)^{1/(2q)}(c_{\ref{lem:R1R2}.1}\sqrt{pd+2q })
\end{align*}
and
\begin{align*}
    R_{2,2q \cdot 2^k}(I_d) \le (pm)^{1/(2q \cdot 2^k)}(c_{\ref{lem:R1R2}.2}\sqrt{pd+2q \cdot 2^k})\le 2^k(pm)^{1/(2q)}(c_{\ref{lem:R1R2}.1}\sqrt{pd+2q})
\end{align*}
by Lemma \ref{lem:R1R2}.

Therefore, it suffices to require
\begin{align*}
    pm \ge \exp(\exp(c_{5}k))(\theta(d/\delta))\frac{q}{\varepsilon^{1+\frac{1}{2q-1}}},
\end{align*}
\begin{align*}
    {pm } \ge {\exp(\exp(c_5k))}\frac{(pm)^{1/(2q)}(\sqrt{2q}) q^{2}}{(\theta(d/\delta))^{k/2-1/4}\varepsilon},
\end{align*}
and
\begin{align*}
    {pm } \ge {\exp(\exp(c_5k))}\frac{(pm)^{1/(2q)}(\sqrt{pd}) q^{2}}{(\theta(d/\delta))^{k/2-1/4}\varepsilon}.
\end{align*}
for some constant $c_5>0$. 

After simplification and using the identity $m=\theta(d/\delta)\frac{d}{\varepsilon^2}$ to substitute $\sqrt{pd} = \frac{\sqrt{pm}\varepsilon}{\sqrt{\theta(d/\delta)}}$, it suffices to require
\begin{align*}
    pm \ge \exp(\exp(c_{6}k))(\theta(d/\delta))\frac{q}{\varepsilon^{1+\frac{1}{2q-1}}},
\end{align*}
\begin{align*}
    {pm } \ge \frac{\exp(\exp(c_6k))}{\varepsilon^{1+\frac{1}{2q-1}}} \paren*{\frac{q^{5/2}}{(\theta(d/\delta))^{k/2-1/4}}}^\frac{2q}{2q-1},
\end{align*}
and
\begin{align*}
    {pm } \ge ({\exp(\exp(c_6k))}\frac{ q^{2}}{(\theta(d/\delta))^{k/2+1/4}})^{\frac{1}{1/2-1/(2q)}}.
\end{align*}
for some constant $c_6$, where we use the fact that $\theta(d/\delta)>1$. 

After further simplification, it suffices to require
\begin{equation}\label{eq: sparsityreq}
\begin{aligned}
    {pm } \ge {\exp(\exp(c_7k))}\left(\frac{1}{\varepsilon^{1+\frac{1}{2q-1}}}\left(\theta(d/\delta)q+\paren*{\frac{q^{5/2}}{(\theta(d/\delta))^{k/2-1/4}}}^\frac{2q}{2q-1} \right)+\frac{ q^{4}}{(\theta(d/\delta))^{k+1/2}}\right)
\end{aligned}
\end{equation}

Now we move on to prove the two special cases of interest. To prove special case (1), we define the function $h(x)=\max\{\log(x),1\}$, and let
\begin{align*}
    k=k(d/\delta)=\lceil\frac{1}{2\max \{c_7, c_{\ref{lem:averuniv}}, 1\}}h \circ h \circ h (d/\delta)\rceil
\end{align*}
and therefore we have
\begin{align*}
    c_7 k &\le \frac{1}{2} \log\log\log (d/\delta)+c_7
\end{align*}
for $d/\delta >4000000 >\exp(e^e)$.

When $d/\delta $ is large enough, we also have
\begin{align*}
    \frac{1}{2} \log\log\log (d/\delta)+c_7  
    &\le \log\log\log(d/\delta) - \log\log\log\log(d/\delta)
\end{align*}
because
\begin{align*}
    \lim_{x\to +\infty}\frac{\log\log\log(x)}{\log\log\log\log(x)}=\infty
\end{align*}

Therefore, in summary, when $d/\delta $ is large enough, we have,
\begin{align*}
    c_7 k &\le \frac{1}{2} \log\log\log (d/\delta)+c_7  \\
    &\le \log\log\log(d/\delta) - \log\log\log\log(d/\delta)  \\
    &= \log \paren*{ \frac{ \log\log(d/\delta)}{\log \log \log (d/\delta)}} \\
    &= \log \paren*{ \log \paren*{ \paren*{\log(d/\delta)}^\frac{1}{\log\log\log(d/\delta)} } } \end{align*}
Taking exponential twice on both sides, for large enough $d/\delta$, we have
\begin{align*}  {\exp(\exp(c_7k))} \le q^{\frac{1}{\log\log\log(d/\delta)}} \le q^{\frac{1}{k}} = q^{o(1)}
\end{align*}
where $O( \cdot)$ and $o(\cdot)$ correspond to the asymptotic behavior when $d/\delta \to +\infty$, and we use the fact that
\begin{align*}
    k \le \frac{1}{2}\log\log\log\!\left(\frac{d}{\delta}\right)+1
	< \log\log\log\!\left(\frac{d}{\delta}\right)
\end{align*}
for large enough $d/\delta$.

In summary, there exists an absolute constant $c_8$ such that, for all $d/\delta>c_8$, we have
\begin{align*}  {\exp(\exp(c_7k))} \le q^{\frac{1}{\log\log\log(d/\delta)}} \le q^{\frac{1}{k}} = q^{o(1)}
\end{align*}

To prove the result also for any $d/\delta \le c_8$, we observe that, on the compact domain $10\le d/\delta \le c_8$, the continuous function ${\exp(\exp(c_7k))}={\exp(\exp(c_7k(d/\delta)))}$ is bounded, so have
\begin{align*}
    {\exp(\exp(c_7k))} \le c_9
\end{align*}
for some constant $c_9>1$.

Therefore, for any $10<d/\delta $, we have
\begin{align*}  {\exp(\exp(c_7k))} \le c_9q^{\frac{1}{k}} = c_9q^{o(1)}
\end{align*}

Then we choose $\theta(d/\delta)=\rho q^{\frac{5}{k-\frac{1}{2}}}$ for some constant $\rho$ to be chosen later. We claim that
\begin{align*}
    \theta(d/\delta) \ge 16\exp(\exp(c_{\ref{lem:averuniv}}k))^2 \\
    \text{or, } \sqrt{\rho} q^{\frac{2.5}{k-\frac{1}{2}}} \ge 4 \exp(\exp(c_{\ref{lem:averuniv}}k)) 
\end{align*}
so that the choice of $\theta$ and $k$ are valid when the constant $\rho$ is large enough.

By similar calculations as before, for $d/\delta \ge c_{10}$ where $c_{10}$ is some absolute constant, we have
\begin{align*}  {\exp(\exp(c_{\ref{lem:averuniv}}k))} &\le q^{\frac{1}{\log\log\log(d/\delta)}} \\
&\le  q^{\frac{5}{\log\log\log(d/\delta) + 1}} \\
&\le  q^{\frac{2.5}{k-\frac{1}{2}}}
\end{align*}
where we use the fact that, 
\begin{align*}
    k &< \frac{\log\log\log (d/\delta)}{2}+1 \\
    \text{so, } k - \frac{1}{2} &< \frac{\log\log\log (d/\delta) + 1}{2}
\end{align*}
For $d/\delta < c_{10}$, we have $k \le c_{11}$ so $4 \exp(\exp(c_{\ref{lem:averuniv}}k)) \le 4 \exp(\exp(c_{12})) $. Thus, by letting $\rho =  16\exp(\exp(c_{12}))^2$, we have, 
\begin{align*}
    \theta(d/\delta)= \rho q^{\frac{5}{k-\frac{1}{2}}}\ge 16\exp(\exp(c_{\ref{lem:averuniv}}k))^2
\end{align*}
is satisfied.

By this choice of $\theta(d/\delta)$, we have
\begin{align*}
    \frac{q^{5/2}}{(\theta(d/\delta))^{k/2-1/4}}=O(1)
\end{align*}
and
\begin{align*}
    (\frac{ q^{4}}{(\theta(d/\delta))^{k+1/2}})^{1+\frac{1}{1-q}}=O(1)
\end{align*}

Therefore, the requirement (\ref{eq: sparsityreq}) becomes
\begin{align*}
    {pm } \ge c_{13}{\exp(\exp(c_{7}k))}\left(\frac{q \cdot \theta(\frac{d}{\delta})}{\varepsilon^{1+\frac{1}{2q-1}}}\right) = c_{13} \frac{q^{1+o(1)}}{\varepsilon^{1+\frac{1}{2q-1}}}
\end{align*}
for some new constants $c_{12}$ and $c_{13}$ and it will be satisfied by requiring
\begin{align*}
    pm \ge c_{13} c_9q^{\frac{1}{k}}\left(\frac{q \cdot \rho q^{\frac{5}{k-\frac{1}{2}}}}{\varepsilon^{1+\frac{1}{2q-1}}}\right) =c_{13} c_9q^{\frac{1}{k}}\left(\frac{q \cdot 16\exp(\exp(c_{12}))^2 q^{\frac{5}{k-\frac{1}{2}}}}{\varepsilon^{1+\frac{1}{2q-1}}}\right) 
\end{align*}
where we recall
\begin{align*}  {\exp(\exp(c_7k))} \le c_9q^{\frac{1}{k}} 
\end{align*}

To prove special case (2), we simply choose $k=1$ and \begin{align*}
    \theta(d/\delta)=16\exp(\exp(c_{\ref{lem:averuniv}}k))^2=16\exp(\exp(c_{\ref{lem:averuniv}}))^2
\end{align*}

\end{proof}

\begin{remark}
In the proof of special case (2), we choose $k=1$ which means we do one step of iterative decoupling. However, this choice is just for convenience and if we don't do iterative decoupling and use Theorem \ref{prop:momestdecoupmatrix} with $k=0$ (which comes from initial moment estimate), we will get exactly the same result. In general, when $\theta=O(1)$, we always get the same result for any constant choice of $k$.
\end{remark}

Theorem \ref{thm:osedecoup} now easily follows from Corollary \ref{cor:momenttheta}.

\begin{theorem}[Theorem~\ref{thm:osedecoup} Restated]
    Let $\Pi$ be an $m \times n$ matrix distributed according to the OSNAP distribution with parameter $p$ as in Definition \ref{def:osnap} and let $U$ be an arbitrary $n \times d$ deterministic matrix such that $U^TU=I_d$. Let $0 < \delta, \varepsilon < 1$ and $d>10$. Then, there exist constants $c_{\ref{thm:osedecoup}.1}, c_{\ref{thm:osedecoup}.2}, c_{\ref{thm:osedecoup}.3}, c_{\ref{thm:osedecoup}.4}, c_{\ref{thm:osedecoup}.5} >0$ such that when,
\begin{itemize}
    \item $k$ is an arbitrary positive integer.
    \item $\theta$ is a positive parameter such that $\theta \ge c_{\ref{thm:osedecoup}.1} \exp(\exp(c_{\ref{thm:osedecoup}.2}k))^2$. 
\end{itemize}
we have,
\begin{equation}
    \begin{aligned}
\Pb \left( 1 - \varepsilon  \leq s_{\min}(\Pi U)   \leq s_{\max}(\Pi U) \leq 1 + \varepsilon \right) \geq 1-\delta
\end{aligned}
\end{equation}
when $m=\theta\cdot(d+\log(d/\delta))/\varepsilon^2$ and the number of non-zeros per column of $\Pi$ satisfies:
\begin{align}
    \begin{aligned} 
    {pm } \ge {\exp(\exp(c_{\ref{thm:osedecoup}.3}k))}\left(\frac{1}{\varepsilon^{1+\frac{1}{\log(d/\delta)}}}\Big(\theta \log(d/\delta)+\frac{\log(d/\delta)^{5/2}}{\theta ^{k/2-1/4}}\Big)+\frac{ \log(d/\delta)^{4}}{\theta ^{k+1/2}}\right)
\end{aligned}
\end{align}

\end{theorem}

\begin{proof}[Proof of Theorem \ref{thm:osedecoup}]
The claim \eqref{osepro} follows from the moment bound \eqref{eq:trmomose} in exactly the same way as the proof of Theorem 3.2 in \cite{chenakkod2025optimal} (Full version). 
\end{proof}

\section{Bound for diagonal terms: Proof of Lemma \ref{lem:decoupfirstlayer} and Lemma \ref{lem:decoupsecondlayer}.} \label{sec:diagterms}

In this section, we provide the proofs of Lemma \ref{lem:decoupfirstlayer} and Lemma \ref{lem:decoupsecondlayer}. 

\subsection{Proof of Lemma \ref{lem:decoupfirstlayer}.}

\begin{proof}[Proof of Lemma \ref{lem:decoupfirstlayer}]

We write
\begin{align*}
    \norm{Y_1^TV^TVY_1}_{2q}=\norm{X_2^TX_1V^TVX_1^TX_2}_{2q}
\end{align*}

The goal of this lemma is to replace $X_2$ by two i.i.d. copies of it, namely $X_3$ and $X_4$. To this end, we will condition on $X_1$. More precisely, we do decoupling for $\norm{X_2^TW^TWX_2}_{2q}$ for fixed matrix $W$, and then plug in $W=VX_1^T$ and uncondition on $X_1$.

\textbf{Step 1: conditioning on $W=VX_1^T$ and obtaining a general bound in terms of $W$.}

By Lemma \ref{decgeneral}, we have
\begin{align*}
    &\norm{X_2^TW^TWX_2}_{2q}\\ \le &\norm{\sum \limits_{(l,\gamma) \in \Xi}(WZ_{2,(l,\gamma)}U)^T(WZ_{2,(l,\gamma)}U)}_{2q}+4\norm{(WX_3)^T(WX_{4})}_{2q}
\end{align*}

Direct calculation shows that after plugging in $W=VX_1^T$ and uncondition on $X_1$, the off-diagonal term $\norm{(WX_3)^T(WX_{4})}_{2q}$ becomes $\norm{X_3^TX_1V^TVX_1^TX_4}_{2q}$ which is exactly what we want. Therefore, it suffices to bound the diagonal term
\begin{align*}
    \norm{\sum \limits_{(l,\gamma) \in \Xi}(WZ_{2,(l,\gamma)}U)^T(WZ_{2,(l,\gamma)}U)}_{2q}
\end{align*}

By Corollary 7.4 (Matrix Rosenthal inequality). in \cite{mackey2014matrix} and Cauchy Schwartz, we have
\begin{align*}
    &\norm{\sum \limits_{(l,\gamma) \in \Xi}(WZ_{2,(l,\gamma)}U)^T(WZ_{2,(l,\gamma)}U)}_{2q} \\\le& c_1 (\norm{\sum \limits_{(l,\gamma) \in \Xi}\E(WZ_{2,(l,\gamma)}U)^T(WZ_{2,(l,\gamma)}U)}_{2q}\\&+(2q)\cdot (\sum \limits_{(l,\gamma) \in \Xi} \norm{(WZ_{2,(l,\gamma)}U)^T(WZ_{2,(l,\gamma)}U)}_{2q}^{2q})^{1/(2q)})
\end{align*}

To estimate $\norm{\sum \limits_{(l,\gamma) \in \Xi}\E(WZ_{2,(l,\gamma)}U)^T(WZ_{2,(l,\gamma)}U)}_{2q}$, we observe
\begin{align*}
    &\sum \limits_{(l,\gamma) \in \Xi}(WZ_{2,(l,\gamma)}U)^T(WZ_{2,(l,\gamma)}U)\\=&\sum \limits_{(l,\gamma) \in \Xi}U^TZ_{2,(l,\gamma)}^TW^TWZ_{2,(l,\gamma)}U\\=& \sum \limits_{(l,\gamma) \in \Xi}\xi_{2,(l,\gamma)}^2u_le_{\mu_{2,(l,\gamma)}}^TW^TWe_{\mu_{2,(l,\gamma)}}u_l^T
    \\=&\sum \limits_{(l,\gamma) \in \Xi}u_le_{\mu_{2,(l,\gamma)}}^TW^TWe_{\mu_{2,(l,\gamma)}}u_l^T
    \\=&\sum \limits_{l}(\sum \limits_{\gamma}e_{\mu_{2,(l,\gamma)}}^TW^TWe_{\mu_{2,(l,\gamma)}})u_lu_l^T
\end{align*}

Therefore, we have
\begin{align*}
    \sum \limits_{(l,\gamma) \in \Xi}\E U^TZ_{2,(l,\gamma)}^TW^TWZ_{2,(l,\gamma)}U=&\sum \limits_{l}(\sum \limits_{\gamma}\E e_{\mu_{2,(l,\gamma)}}^TW^TWe_{\mu_{2,(l,\gamma)}})u_lu_l^T
    \\=&\sum \limits_{l}(\frac{pm}{m}\sum \limits_{i \in [m]}e_{{i}}^TW^TWe_{i})u_lu_l^T
    \\=&(\frac{pm}{m}\sum \limits_{i \in [m]}e_{{i}}^TW^TWe_{{i}}) \sum \limits_{l}u_lu_l^T
    \\=&(p\sum \limits_{i \in [m]}\norm{We_{{i}}}_{l_2([m])}^2) \sum \limits_{l}u_lu_l^T
    \\=& (p\sum \limits_{i \in [m]}\norm{We_{{i}}}_{l_2([d])}^2) I_d
\end{align*}

Therefore, we have
\begin{align*}
    \norm{\sum \limits_{(l,\gamma) \in \Xi}\E(WZ_{2,(l,\gamma)}U)^T(WZ_{2,(l,\gamma)}U)}_{2q}=&(p\sum \limits_{i \in [m]}\norm{We_{i}}^2_{l_2([d])})
    \\=&(p\norm{W}_{F}^2)
\end{align*}

Now we have finished bounding the first terms in the RHS of the matrix Rosenthal inequality. We will move on to bound the second term
\begin{align*}
    (2q)\cdot (\E \sum \limits_{(l,\gamma) \in \Xi} \norm{(WZ_{2,(l,\gamma)}U)^T(WZ_{2,(l,\gamma)}U)}_{2q}^{2q})^{1/(2q)}
\end{align*}

To this end, we observe
\begin{align*}
    &(\sum \limits_{(l,\gamma) \in \Xi} \norm{(WZ_{2,(l,\gamma)}U)^T(WZ_{2,(l,\gamma)}U)}_{2q}^{2q})^{1/(2q)}
    \\ = &(\sum\limits_{(l,\gamma) \in \Xi}(\norm{(e_{\mu_{2,(l,\gamma)}}^TW^TWe_{\mu_{2,(l,\gamma)}})u_lu_l^T}_{2q}^{2q}))^{1/(2q)}
    \\ = & (\sum\limits_{(l,\gamma) \in \Xi}\frac{1}{d}\E_{\mu_{2,(l,\gamma)}}((\norm{We_{\mu_{2,(l,\gamma)}}}_{l_2([d])}^2)\norm{u_l}_{l_2([d])}^{2})^{2q})^{1/(2q)}
\end{align*}

In summary, we have
\begin{align*}
    &\norm{\sum \limits_{(l,\gamma) \in \Xi}(WZ_{2,(l,\gamma)}U)^T(WZ_{2,(l,\gamma)}U)}_{2q} \\\le& c_1 (p\norm{W}_{F}^2+(2q)\cdot (\sum\limits_{(l,\gamma) \in \Xi}\frac{1}{d}\E_{\mu_{2,(l,\gamma)}}((\norm{We_{\mu_{2,(l,\gamma)}}}_{l_2([d])}^2)\norm{u_l}_{l_2([d])}^{2})^{2q})^{1/(2q)}
\end{align*}
for any fixed matrix $W$.

\textbf{Step 2: plugging in $W=VX_1^T$ and unconditioning on $X_1$.}

Plugging in $W=VX_1^T$ and unconditioning on $X_1$ and using triangle inequality for $\norm{\cdot}_{L_{2q}(X_1)}$, we have
\begin{align*}
    &\norm{\sum \limits_{(l,\gamma) \in \Xi}(VX_1^TZ_{2,(l,\gamma)}U)^T(VX_1^TZ_{2,(l,\gamma)}U)}_{2q} \\\le& c_1 (p\norm{\norm{VX_1^T}_{F}^2}_{2q}+(2q)\cdot (\sum\limits_{(l,\gamma) \in \Xi}\frac{1}{d}\E_{X_1,\mu_{2,(l,\gamma)}}(\norm{VX_1^Te_{\mu_{2,(l,\gamma)}}}_{l_2([d])}^2\norm{u_l}_{l_2([d])}^{2})^{2q})^{1/(2q)})
\end{align*}

To bound the first term $p\norm{\norm{VX_1^T}_{F}^2}_{2q}$, we observe that
\begin{align*}
    \norm{\norm{VX_1^T}_{F}^2}_{2q} \le \norm{\norm{V}^2\norm{X_1^T}_{F}^2}_{2q} \le \norm{V}^2\norm{\norm{X_1^T}_{F}^2}_{2q} 
\end{align*}

By Lemma \ref{lem:xfrobnorm}, we have
\begin{align*}
    \norm{(\norm{X_1^T}_{F}^2)}_{L_{2q}(\Pb)}=&\norm{\norm{X_1^T}_{F}}_{L_{4q}(\Pb)}^{2} \\\le& (c_2\sqrt{pmd+4q})^2
\end{align*}
for some constant $c_2$.

Therefore, for the first term, we have the bound
\begin{align*}
    p\norm{\norm{VX_1^T}_{F}^2}_{2q}\le c_2^2\norm{V}^2p(pmd+4q) \le 4c_2^2\norm{V}^2(pmpd+q^2)
\end{align*}
since $p<1$.

To bound the second term
\begin{align*}
    (2q)\cdot (\sum\limits_{(l,\gamma) \in \Xi}\frac{1}{d}\E_{X_1,\mu_{2,(l,\gamma)}}(\norm{VX_1^Te_{\mu_{2,(l,\gamma)}}}_{l_2([d])}^2\norm{u_l}_{l_2([d])}^{2})^{2q})^{1/(2q)}
\end{align*}
we observe
\begin{align*}
    &(\sum\limits_{(l,\gamma) \in \Xi}\frac{1}{d}\E_{X_1,\mu_{2,(l,\gamma)}}(\norm{VX_1^Te_{\mu_{2,(l,\gamma)}}}_{l_2([d])}^2\norm{u_l}_{l_2([d])}^{2})^{2q})^{1/(2q)}
    \\ \le & (\sum\limits_{(l,\gamma) \in \Xi}\frac{1}{d}\E_{X_1,\mu_{2,(l,\gamma)}}(\norm{V}^2\norm{X_1^Te_{\mu_{2,(l,\gamma)}}}_{l_2([d])}^2\norm{u_l}_{l_2([d])}^{2})^{2q})^{1/(2q)}
    \\ \le & \norm{V}^2(\sum\limits_{(l,\gamma) \in \Xi}\frac{1}{d}\E_{X_1,\mu_{2,(l,\gamma)}}(\norm{X_1^Te_{\mu_{2,(l,\gamma)}}}_{l_2([d])}^2\norm{u_l}_{l_2([d])}^{2})^{2q})^{1/(2q)}
\end{align*}

Conditioning on $\mu_{2,(l,\gamma)}$, by Lemma \ref{lem:rowbounds}, we have
\begin{align*}
    \E_{X_1}(\norm{X_1^Te_{\mu_{2,(l,\gamma)}}}_{l_2([d])}^2)^{2q} \le (c_3(pd+q))^{2q}
\end{align*}
for some constant $c_3$. And then by iterated integration, we have
\begin{align*}
    \E_{X_1,\mu_{2,(l,\gamma)}}(\norm{X_1^Te_{\mu_{2,(l,\gamma)}}}_{l_2([d])}^2)^{2q} \le (c_3(pd+q))^{2q}
\end{align*}
for any $(l,\gamma) \in \Xi$.

Therefore, we have
\begin{align*}
    &(\sum\limits_{(l,\gamma) \in \Xi}\frac{1}{d}\E_{X_1,\mu_{2,(l,\gamma)}}(\norm{VX_1^Te_{\mu_{2,(l,\gamma)}}}_{l_2([d])}^2\norm{u_l}_{l_2([d])}^{2})^{2q})^{1/(2q)}
    \\ \le & \norm{V}^2(\sum\limits_{(l,\gamma) \in \Xi}\frac{1}{d}\E_{X_1,\mu_{2,(l,\gamma)}}(\norm{X_1^Te_{\mu_{2,(l,\gamma)}}}_{l_2([d])}^2\norm{u_l}_{l_2([d])}^{2})^{2q})^{1/(2q)}
    \\ \le & c_3(pd+q)\norm{V}^2(\sum\limits_{(l,\gamma) \in \Xi}\frac{1}{d}(\norm{u_l}_{l_2([d])}^{2})^{2q})^{1/(2q)}
    \\ \le & c_3(pd+q)\norm{V}^2(pm)^{1/(2q)}
\end{align*}

\end{proof}

\subsection{Proof of Lemma \ref{lem:decoupsecondlayer}}

\begin{proof}[Proof of Lemma \ref{lem:decoupsecondlayer}.]

The goal of this lemma is to replace $X_1$ in the expression
\begin{align*}
    \norm{X_3^TX_1V^TVX_1^TX_4}_{2q}
\end{align*}
by two i.i.d. copies of it, namely $X_5$ and $X_6$. To this end, we will condition on $X_3$ and $X_4$. More precisely, we do decoupling for $\norm{W_3^TX_1V^TVX_1^TW_4}_{2q}$ for fixed matrix $W_3$ and $W_4$, and then plug in $W_3=VX_3^T$ and $W_4=VX_4^T$, and uncondition on $X_3$ and $X_4$.

\textbf{Step 1: conditioning on $W_3=X_3^T$ and $W_4=X_4^T$ and obtain a general bound in terms of $W_3$ and $W_4$.}

By Lemma \ref{decgeneral}, we have
\begin{align*}
    &\norm{W_3^TX_1^TV^TVX_1W_4}_{2q}\\ \le &\norm{\sum \limits_{(l,\gamma) \in \Xi}W_3^T(Z_{1,(l,\gamma)}U)^TV^TV(Z_{1,(l,\gamma)}U)W_4}_{2q}+4\norm{W_3^TX_5^TV^TVX_6W_4}_{2q}
\end{align*}

Direct calculation shows that after plugging in $W_3=X_3^T$ and $W_4=X_4^T$ and unconditioning on $X_3$ and $X_4$, the off-diagonal term $\norm{(W_3X_5)^T(W_4X_{6})}_{2q}$ becomes $\norm{X_3^TX_5V^TVX_6^TX_4}_{2q}$ which is exactly what we want. Therefore, it suffices to bound the diagonal term
\begin{align*}
    \norm{\sum \limits_{(l,\gamma) \in \Xi}W_3^T(Z_{1,(l,\gamma)}U)^TV^TV(Z_{1,(l,\gamma)}U)W_4}_{2q}
\end{align*}

We observe that $\sum \limits_{(l,\gamma) \in \Xi}(Z_{1,(l,\gamma)}U)^TV^TV(Z_{1,(l,\gamma)}U)$ is a sum of positive semi-definite matrices, so it is also positive semi-definite. Therefore, there exists a random matrix $\Lambda$ such that
\begin{align*}
    \sum \limits_{(l,\gamma) \in \Xi}(Z_{1,(l,\gamma)}U)^TV^TV(Z_{1,(l,\gamma)}U)=\Lambda^T\Lambda
\end{align*}

Therefore, we have
\begin{align*}
    &\norm{\sum \limits_{(l,\gamma) \in \Xi}W_3^T(Z_{1,(l,\gamma)}U)^TV^TV( Z_{1,(l,\gamma)}U)W_4}_{2q}^{2q}
    \\=&\norm{W_3^T\Lambda^T\Lambda W_4}_{2q}^{2q}
    \\=&\E\tr(|W_3^T\Lambda^T\Lambda W_4|^{2q})
    \\=&\E\tr(|(\Lambda W_3)^T\Lambda W_4|^{2q})
    \\=&\E\tr(((\Lambda W_4)^T\Lambda W_3(\Lambda W_3)^T\Lambda W_4)^{2q})
    \\=&\E\tr(\Lambda W_4(\Lambda W_4)^T\Lambda W_3(\Lambda W_3)^T \cdots \Lambda W_4(\Lambda W_4)^T\Lambda W_3(\Lambda W_3)^T)
    \\ \le & \norm{\Lambda W_4(\Lambda W_4)^T}_{2q}^q\norm{\Lambda W_3(\Lambda W_3)^T}_{2q}^q
\end{align*}
where we use the Holder inequality from Lemma \ref{lem:holdervanhandel}.

To estimate $\norm{\Lambda W_4(\Lambda W_4)^T}_{2q}$, we observe that
\begin{align*}
    \norm{\Lambda W_4(\Lambda W_4)^T}_{2q}=&\norm{(\Lambda W_4)^T \Lambda W_4}_{2q}
    \\=&\norm{\sum \limits_{(l,\gamma) \in \Xi}W_4^T(Z_{1,(l,\gamma)}U)^TV^TV(Z_{1,(l,\gamma)}U)W_4}_{2q}
\end{align*}
because $\Lambda W_4$ is a square matrix.

By Corollary 7.4 (Matrix Rosenthal inequality). in \cite{mackey2014matrix} and Cauchy Schwartz, we have
\begin{align*}
    &\norm{\sum \limits_{(l,\gamma) \in \Xi}W_4^T(Z_{1,(l,\gamma)}U)^TV^TV(Z_{1,(l,\gamma)}U)W_4}_{2q} \\\le& c_1 (\norm{\sum \limits_{(l,\gamma) \in \Xi}\E W_4^T(Z_{1,(l,\gamma)}U)^TV^TV(Z_{1,(l,\gamma)}U)W_4}_{2q}\\&+(2q)\cdot (\sum \limits_{(l,\gamma) \in \Xi} \norm{W_4^T(Z_{1,(l,\gamma)}U)^TV^TV(Z_{1,(l,\gamma)}U)W_4}_{2q}^{2q})^{1/(2q)})
\end{align*}

To estimate $\norm{\sum \limits_{(l,\gamma) \in \Xi}\E W_4^T(Z_{1,(l,\gamma)}U)^TV^TV(Z_{1,(l,\gamma)}U)W_4}_{2q}$, we observe
\begin{align*}
    &\sum \limits_{(l,\gamma) \in \Xi}\E W_4^T(Z_{1,(l,\gamma)}U)^TV^TV(Z_{1,(l,\gamma)}U)W_4\\=& \sum \limits_{(l,\gamma) \in \Xi}\xi_{2,(l,\gamma)}^2W_4^Tu_le_{\mu_{1,(l,\gamma)}}^TV^TVe_{\mu_{1,(l,\gamma)}}u_l^TW_4
    \\=&\sum \limits_{(l,\gamma) \in \Xi}W_4^Tu_le_{\mu_{1,(l,\gamma)}}^TV^TVe_{\mu_{1,(l,\gamma)}}u_l^TW_4
    \\=&\sum \limits_{l}(\sum \limits_{\gamma}e_{\mu_{1,(l,\gamma)}}^TV^TVe_{\mu_{1,(l,\gamma)}})W_4^Tu_lu_l^TW_4
\end{align*}

Therefore, we have
\begin{align*}
    &\sum \limits_{(l,\gamma) \in \Xi}\E W_4^T(Z_{1,(l,\gamma)}U)^TV^TV(Z_{1,(l,\gamma)}U)W_4\\=&\sum \limits_{l}(\sum \limits_{\gamma}\E e_{\mu_{1,(l,\gamma)}}^TV^TVe_{\mu_{1,(l,\gamma)}})W_4^Tu_lu_l^TW_4
    \\=&\sum \limits_{l}(\frac{pm}{m}\sum \limits_{i \in [m]}e_{{i}}^TV^TVe_{\mu_{i}})W_4^Tu_lu_l^TW_4
    \\=&(\frac{pm}{m}\sum \limits_{i \in [m]}e_{{i}}^TV^TVe_{{i}}) \sum \limits_{l}W_4^Tu_lu_l^TW_4
    \\=&(p\sum \limits_{i \in [m]}\norm{Ve_{{i}}}_{l_2([m])}^2) \sum \limits_{l}W_4^Tu_lu_l^TW_4
    \\=& (p\sum \limits_{i \in [m]}\norm{Ve_{{i}}}_{l_2([d])}^2) W_4^T I_d W_4
\end{align*}

Therefore, we have
\begin{align*}
    &\norm{\sum \limits_{(l,\gamma) \in \Xi}\E W_4^T(Z_{1,(l,\gamma)}U)^TV^TV(Z_{1,(l,\gamma)}U)W_4}_{2q}\\=&(p\sum \limits_{i \in [m]}\norm{Ve_{i}}^2_{l_2([d])})\norm{W_4^T I_d W_4}_{2q}
    \\=&(p\norm{V}_{F}^2)\norm{W_4^TW_4}_{2q}
    \\ \le &(p\norm{V}^2d)\norm{W_4^TW_4}_{2q}
\end{align*}

Now we have finished bounding the first terms in the RHS of the matrix Rosenthal inequality. We will move on to bound the second term
\begin{align*}
    (2q)\cdot (\E \sum \limits_{(l,\gamma) \in \Xi} \norm{W_4^T(Z_{1,(l,\gamma)}U)^TV^TV(Z_{1,(l,\gamma)}U)W_4}_{2q}^{2q})^{1/(2q)}
\end{align*}

To this end, we observe
\begin{align*}
    &(\sum \limits_{(l,\gamma) \in \Xi} \norm{W_4^T(Z_{1,(l,\gamma)}U)^TV^TV(Z_{1,(l,\gamma)}U)W_4}_{2q}^{2q})^{1/(2q)}
    \\ = &(\sum \limits_{(l,\gamma) \in \Xi} \norm{e_{\mu_{1,(l,\gamma)}}^TV^TVe_{\mu_{1,(l,\gamma)}}W_4^Tu_lu_l^TW_4}_{2q}^{2q})^{1/(2q)}
    \\ = & (\sum\limits_{(l,\gamma) \in \Xi}\frac{1}{d}\E_{\mu_{1,(l,\gamma)}}((\norm{Ve_{\mu_{1,(l,\gamma)}}}_{l_2([d])}^2)\norm{W_4^Tu_l}_{l_2([d])}^{2})^{2q})^{1/(2q)}
    \\ \le & (\sum\limits_{(l,\gamma) \in \Xi}\frac{1}{d}(\norm{V}^2\norm{W_4^Tu_l}_{l_2([d])}^{2})^{2q})^{1/(2q)}
\end{align*}

In summary, we have
\begin{align*}
    &\norm{\sum \limits_{(l,\gamma) \in \Xi}W_4^T(Z_{1,(l,\gamma)}U)^TV^TV(Z_{1,(l,\gamma)}U)W_4}_{2q}  \\\le& c_1 ((pd\norm{V}^2)\norm{W_4^TW_4}_{2q}+(2q)\cdot (\sum\limits_{(l,\gamma) \in \Xi}\frac{1}{d}(\norm{V}^2\norm{W_4^Tu_l}_{l_2([d])}^{2})^{2q})^{1/(2q)}
\end{align*}
for any fixed matrix $W_4$.

Similarly, we have
\begin{align*}
    &\norm{\sum \limits_{(l,\gamma) \in \Xi}W_3^T(Z_{1,(l,\gamma)}U)^TV^TV(Z_{1,(l,\gamma)}U)W_3}_{2q}  \\\le& c_1 ((pd\norm{V}^2)\norm{W_3^TW_3}_{2q}+(2q)\cdot (\sum\limits_{(l,\gamma) \in \Xi}\frac{1}{d}(\norm{V}^2\norm{W_3^Tu_l}_{l_2([d])}^{2})^{2q})^{1/(2q)}
\end{align*}
for any fixed matrix $W_3$.

Therefore, we have
\begin{align*}
    &\norm{\sum \limits_{(l,\gamma) \in \Xi}W_3^T(Z_{1,(l,\gamma)}U)^TV^TV( Z_{1,(l,\gamma)}U)W_4}_{2q}^{2q}
    \\ \le & \norm{\Lambda W_4(\Lambda W_4)^T}_{2q}^q\norm{\Lambda W_3(\Lambda W_3)^T}_{2q}^q
    \\ \le & \left(c_1 ((pd\norm{V}^2)\norm{W_4^TW_4}_{2q}+(2q)\cdot (\sum\limits_{(l,\gamma) \in \Xi}\frac{1}{d}(\norm{V}^2\norm{W_4^Tu_l}_{l_2([d])}^{2})^{2q})^{1/(2q)}\right)^q \\ & \cdot \left(c_1 ((pd\norm{V}^2)\norm{W_3^TW_3}_{2q}+(2q)\cdot (\sum\limits_{(l,\gamma) \in \Xi}\frac{1}{d}(\norm{V}^2\norm{W_3^Tu_l}_{l_2([d])}^{2})^{2q})^{1/(2q)}\right)^q
\end{align*}

\textbf{Step 2: plugging in $W_3=X_3$ and $W_4=X_4$ and uncondition on $X_3$ and $X_4$.}

plugging in $W_3=X_3$ and $W_4=X_4$, by Cauchy Schwarz inequality on the probability space, we have
\begin{align*}
    &\E_{X_3,X_4}\norm{\sum \limits_{(l,\gamma) \in \Xi}X_3^T(Z_{1,(l,\gamma)}U)^TV^TV( Z_{1,(l,\gamma)}U)X_4}_{2q}^{2q}
    \\ \le & \left(\E_{X_4}\left(c_1 ((pd\norm{V}^2)\norm{X_4^TX_4}_{S_{2q}}+(2q)\cdot (\sum\limits_{(l,\gamma) \in \Xi}\frac{1}{d}(\norm{V}^2\norm{X_4^Tu_l}_{l_2([d])}^{2})^{2q})^{1/(2q)}\right)^{2q}\right)^{1/2} \\ & \cdot \left(\E_{X_3}\left(c_1 ((pd\norm{V}^2)\norm{X_3^TX_3}_{S_{2q}}+(2q)\cdot (\sum\limits_{(l,\gamma) \in \Xi}\frac{1}{d}(\norm{V}^2\norm{X_3^Tu_l}_{l_2([d])}^{2})^{2q})^{1/(2q)}\right)^{2q}\right)^{1/2}
    \\=&\E_{X_4}\left(c_1 ((pd\norm{V}^2)\norm{X_4^TX_4}_{S_{2q}}+(2q)\cdot (\sum\limits_{(l,\gamma) \in \Xi}\frac{1}{d}(\norm{V}^2\norm{X_4^Tu_l}_{l_2([d])}^{2})^{2q})^{1/(2q)}\right)^{2q}
\end{align*}
where we use the fact that $X_3$ and $X_4$ have the same distribution.

Using triangle inequality for $\norm{\cdot}_{L_{2q}(X_4)}$, we have
\begin{align*}
    &\E_{X_3,X_4}\norm{\sum \limits_{(l,\gamma) \in \Xi}X_3^T(Z_{1,(l,\gamma)}U)^TV^TV( Z_{1,(l,\gamma)}U)X_4}_{2q} \\\le& c_1 \left((pd\norm{V}^2)\norm{X_4^TX_4}_{2q}+(2q)\cdot (\sum\limits_{(l,\gamma) \in \Xi}\frac{1}{d}\E_{X_4}(\norm{V}^2\norm{X_4^Tu_l}_{l_2([d])}^{2})^{2q})^{1/(2q)}\right)
    \\=& c_1 \norm{V}^2\left(pd\norm{X_4^TX_4}_{2q}+(2q)\cdot (\sum\limits_{(l,\gamma) \in \Xi}\frac{1}{d}\E_{X_4}(\norm{X_4^Tu_l}_{l_2([d])}^{2})^{2q})^{1/(2q)}\right)
\end{align*}

To bound the first term $pd\norm{X_4^TX_4}_{2q}$, we observe that
\begin{align*}
\norm{X_4^TX_4}_{2q} \le \norm{X_4^TX_4-pm I_d }_{2q}+\norm{pm I_d }_{2q}
\end{align*}

Recall that by Lemma \ref{lem:decoup}, we have
\begin{align*}
    \norm{X_4^TX_4-pm I_d }_{2q} \le 4\norm{X_3^TX_4}_{2q}=4\norm{Y_1}_{2q}
\end{align*}

Therefore, for the first term, we have
\begin{align*}
    pd\norm{X_4^TX_4}_{2q} \le pmpd + 4pd\norm{Y_1}_{2q}
\end{align*}

To bound the second term
\begin{align*}
    (2q)\cdot (\sum\limits_{(l,\gamma) \in \Xi}\frac{1}{d}\E_{X_4}(\norm{X_4^Tu_l}_{l_2([d])}^{2})^{2q})^{1/(2q)}
\end{align*}
we use Lemma \ref{lem:rowbounds} and obtain
\begin{align*}
    &(\sum\limits_{(l,\gamma) \in \Xi}\frac{1}{d}\E_{X_4}(\norm{X_4^Tu_l}_{l_2([d])}^{2})^{2q})^{1/(2q)}\\ \le &(\sum\limits_{(l,\gamma) \in \Xi}\frac{1}{d}(c_3 (pd+q)\norm{u_l}_{l_2([d])}^{2})^{2q})^{1/(2q)}
    \\ = &(\sum\limits_{(l,\gamma) \in \Xi}\frac{1}{d}(c_3 (pd+q))^{2q}\norm{u_l}_{l_2([d])}^{4q})^{1/(2q)}
    \\ \le &(\sum\limits_{(l,\gamma) \in \Xi}\frac{1}{d}(c_3 (pd+q))^{2q}\norm{u_l}_{l_2([d])}^{2})^{1/(2q)}
    \\ = &(\frac{1}{d}(c_3 pd)^{2q}(pmd))^{1/(2q)}
    \\=&c_3(pd+q)(pm)^{1/(2q)}
\end{align*}

\end{proof}

\section{Moment Bound for Vector Norm.}\label{sec:vecmombd}

In this section, we bound $\norm{\norm{(S_1 U)^T(S_2 U)u}_2}_{L_{2q}(\Pb)}$ for an arbitrary fixed vector $u$ which is required in the proof of Lemma \ref{lem:Rsimplify}. 

For convenience, let $X_i=S_iU$ for $i=1,2,...$. We will use the iterative method used in 
Lemma 8.9 (Row Norm for LESS-IC) in \cite{chenakkod2025optimal}(Full version). More precisely, we observe that
\begin{align*}
    \norm{\norm{X_1^TX_2u}_2}_{L_{2q}(\Pb)} &= \paren*{ \E \sqbr*{ \norm{X_1^TX_2u}_2^{2q} } }^{1/2q} \\ 
    &= \paren*{ \E \sqbr*{ \paren*{\ip{X_1^TX_2v,X_1^TX_2v}}^{q} } }^{1/2q}\\
    &= \norm{\ip{X_1^TX_2v,X_1^TX_2v}}_{L_{q}(\Pb)}^{1/2}
\end{align*}
and then by using decoupling we decompose
\begin{align} \label{eq:vecdecoup}
    \norm{\ip{X_1^TX_2v,X_1^TX_2v}}_{L_{q}(\Pb)} \le C(\norm{\text{diagonal terms})}_{L_{q}(\Pb)}+\norm{\ip{X_1^TX_2v,X_3^TX_4v}}_{L_{q}(\Pb)})
\end{align}
We shall soon provide the details of this step. For now, we focus on the latter term. We have
\begin{align*}
    \ip{X_1^TX_2v,X_3^TX_4v}=\ip{X_3X_1^TX_2v,X_4v}=\ip{S_3(UX_1^TX_2v),S_4(Uv)}
\end{align*}

By independence, we can condition on $X_1$ and $X_2$. More precisely, we can fix the vectors $u=(UX_1^TX_2v)$ and $w=Uv$ and consider the conditional moments
\begin{align*}
    \E_{S_3,S_4}\ip{S_3(UX_1^TX_2v),S_4(Uv)}=\E_{S_3,S_4}\ip{S_3u,S_4w}
\end{align*}
This motivates the next lemma.

\begin{lemma}\label{lem:scalarnorm}
Let $S$ be an $m \times n$ matrix distributed according to the fully independent unscaled OSNAP distribution with parameter $p$. Let $S_1$ and $S_2$ be i.i.d. copies of $S$. Let $u \in \R^n$ and $v \in \R^n$ be a vector with $\norm{u}_{l_2([n])}=\norm{v}_{l_2([n])}=1$. Let $q \ge 1$ be an integer. Then there exists a constant $c_{\ref{lem:scalarnorm}}$ such that
\begin{align*}
    \norm{\ip{S_1u,S_2v}}_{L_{2q}(\Pb)} \le c_{\ref{lem:scalarnorm}}q
\end{align*}
\end{lemma}

\begin{proof}
By definition, we have
\begin{align*}
    S = \sum_{l=1}^n \sum_{\gamma=1}^s \xi_{(l,\gamma)} e_{\mu_{(l, \gamma)}} e_l ^\top
\end{align*}
We can rewrite $S$ as
\begin{align*}
    S = \sum_{j=1}^n \sum_{i=1}^m \zeta_{(j,i)} (\sum_{\gamma=1}^s\mathbbm{1}_{\{i\}}(\mu_{(j, \gamma)}))e_{i} e_j ^\top
=\sum_{j=1}^n \sum_{\gamma=1}^s \sum_{i \in [(m/s)(\gamma-1)+1:(m/s)\gamma]} \zeta_{(j,i)} (\mathbbm{1}_{\{i\}}(\mu_{(j, \gamma)}))e_{i} e_j ^\top
\end{align*}
where $\zeta_{(j,i)}$ are i.i.d. random signs.

In fact, we can set $\xi_{(l,\gamma)} = \zeta_{(l, \mu_{(l, \gamma)})}$ and recover the original definition.

Define $\eta_{i,j}=\sum_{\gamma=1}^s\mathbbm{1}_{\{i\}}(\mu_{(j, \gamma)})$. 

Therefore, we have
\begin{align*}
    \ip{S_1u,S_2v}=&\sum_r (\sum_{j_1} (S_1)_{r,j_1}u_{j_1})(\sum_{j_2} (S_2)_{r,j_2}v_{j_2})
    \\=&\sum_r (\sum_{j_1,j_2} (S_1)_{r,j_1}\cdot u_{j_1}\cdot(S_2)_{r,j_2}\cdot v_{j_2})
    \\=&\sum_r (\sum_{j_1,j_2} \zeta_{1,r,j_1}\eta_{1,r,j_1}\cdot u_{j_1}\cdot \zeta_{2,r,j_2}\eta_{2,r,j_2}\cdot v_{j_2})
\end{align*}

Let $\iota$ be the default embedding 
\[
\iota : [m] \times [n] \to [mn]
\]
defined by
\[
\iota(i,j) = n\,(i-1) + j.
\]

We define random vectors $W_1 \in \R^{mn}$ and $W_2 \in \R^{mn}$ such that
\begin{align*}
    W_1(\iota(i,j))=\zeta_{1,i,j}
\end{align*}
and
\begin{align*}
    W_2(\iota(i,j))=\zeta_{2,i,j}
\end{align*}

Let $A$ be an $mn \times mn$ block diagonal matrix whose $r$th block along the diagonal is
\begin{align*}
A_r=&\diag(\eta_{1,r,1},...,\eta_{1,r,n})u \cdot (\diag(\eta_{2,r,1},...,\eta_{2,r,n})v)^T
    \\=&\begin{pmatrix} \eta_{1,r,1} u_1 \\ \eta_{1,r,2} u_2 \\ \vdots \\ \eta_{1,r,n} u_n \end{pmatrix} \cdot  \begin{pmatrix} \eta_{2,r,1} v_1 \\ \eta_{2,r,2} v_2 \\ \vdots \\ \eta_{2,r,n} v_n \end{pmatrix}^T
    \\=&\begin{pmatrix}
\eta_{1,r,1}u_1 \, \eta_{2,r,1}v_1 & \eta_{1,r,1}u_1 \, \eta_{2,r,2}v_2 & \cdots & \eta_{1,r,1}u_1 \, \eta_{2,r,n}v_n \\
\eta_{1,r,2}u_2 \, \eta_{2,r,1}v_1 & \eta_{1,r,2}u_2 \, \eta_{2,r,2}v_2 & \cdots & \eta_{1,r,2}u_2 \, \eta_{2,r,n}v_n \\
\vdots & \vdots & \ddots & \vdots \\
\eta_{1,r,n}u_n \, \eta_{2,r,1}v_1 & \eta_{1,r,n}u_n \, \eta_{2,r,2}v_2 & \cdots & \eta_{1,r,n}u_n \, \eta_{2,r,n}v_n
\end{pmatrix}
\end{align*}

Therefore, direct calculation shows that
\begin{align*}
    \ip{S_1u,S_2v}=W_1^TAW_2
\end{align*}

In fact, since \(A\) is block diagonal, we may write
	\begin{align*}
		W_1^T A W_2 = \sum_{r \in [m]} \; \bigl(W_1^r\bigr)^T A_r W_2^r,
	\end{align*}
	where \(W_1^r \in \mathbb{R}^n\) is the subvector of \(W_1\) corresponding to the indices \(\iota(r,1), \dots, \iota(r,n)\) and similarly for \(W_2^r\).
	
	For a fixed \(r\), we have
	\begin{align*}
		\bigl(W_1^r\bigr)^T A_r W_2^r 
		&= \sum_{j_1=1}^n \sum_{j_2=1}^n \Bigl[W_1^r\Bigr]_{j_1} \, \Bigl[A_r\Bigr]_{j_1,j_2} \, \Bigl[W_2^r\Bigr]_{j_2} \\
		&= \sum_{j_1=1}^n \sum_{j_2=1}^n \zeta_{1,r,j_1} \Bigl( \eta_{1,r,j_1} u_{j_1} \, \eta_{2,r,j_2} v_{j_2} \Bigr) \zeta_{2,r,j_2} \\
		&= \sum_{j_1,j_2} \zeta_{1,r,j_1}\,\eta_{1,r,j_1}\, u_{j_1} \, \zeta_{2,r,j_2}\,\eta_{2,r,j_2}\, v_{j_2}.
	\end{align*}
	
	Summing over all \(r\) gives
	\begin{align*}
		W_1^T A W_2 
		&= \sum_r \; \bigl(W_1^r\bigr)^T A_r W_2^r \\
		&= \sum_r \sum_{j_1,j_2} \zeta_{1,r,j_1}\,\eta_{1,r,j_1}\, u_{j_1} \, \zeta_{2,r,j_2}\,\eta_{2,r,j_2}\, v_{j_2} \\
		&= \langle S_1 u, S_2 v \rangle.
	\end{align*}

Conditioning on $\mu_{1,(l,\gamma)}$ and $\mu_{2,(l,\gamma)}$ (and therefore conditioning on $A$) and using iterated expectation, Lemma \ref{lem:subgaussian-chaos} implies
\begin{align*}
    \norm{W_1^T A W_2}_{L_{2q}(\Pb)} \le c_1(\sqrt{2q} \norm{\norm{A}_F}_{L_{2q}(\Pb)}+2q \norm{\norm{A}}_{L_{2q}(\Pb)})
\end{align*}
for some constant $c_1$.

To estimate $\norm{\norm{A}_F}_{L_{2q}(\Pb)}$, we observe that
\begin{align*}
    \norm{A}_F^2=&\sum_{r,j_1,j_2} (\eta_{1,r,j_1}u_{j_1} \, \eta_{2,r,j_2}v_{j_2})^2\\=&\sum_{j_1,j_2} \sum_{r} (\eta_{1,r,j_1}u_{j_1} \, \eta_{2,r,j_2}v_{j_2})^2=\sum_{j_1,j_2} u_{j_1}^2 v_{j_2}^2(\sum_{r} \eta_{1,r,j_1}^2  \eta_{2,r,j_2}^2)
\end{align*}

We further observe that
\begin{align*}
    (\sum_{r} \eta_{1,r,j_1}^2  \eta_{2,r,j_2}^2)=&\sum_{\gamma=1}^s \sum_{r \in [(m/s)(\gamma-1)+1:(m/s)\gamma]} \eta_{1,r,j_1}^2  \eta_{2,r,j_2}^2
    \\=&\sum_{\gamma=1}^s \mathbbm{1}_{\{\mu_{1,r,j_1}=\mu_{2,r,j_1}\}}
\end{align*}

Since $\{\mathbbm{1}_{\{\mu_{1,r,j_1}=\mu_{2,r,j_1}\}}\}_{r=1}^2$ are independent Bernoulli-$p$ random variables, by Lemma 2 in \cite{cohen2018simple}, we have
\begin{align*}
    \norm{\sum_{r} \eta_{1,r,j_1}^2  \eta_{2,r,j_2}^2}_{L_{q}}=\norm{\sum_{\gamma=1}^s \mathbbm{1}_{\{\mu_{1,r,j_1}=\mu_{2,r,j_1}\}}}_{L_{q}} \le c_2(q)
\end{align*}
for some constant $c_2$.

Therefore, by triangle inequality, we have
\begin{align*}
\norm{\norm{A}_F}_{L_{2q}}^2=&\norm{\norm{A}_F^2}_{L_{q}} 
\\ \le & \sum_{j_1,j_2} u_{j_1}^2 v_{j_2}^2\norm{(\sum_{r} \eta_{1,r,j_1}^2  \eta_{2,r,j_2}^2)}_{L_{q}} 
\\ \le& c_2\sum_{j_1,j_2} u_{j_1}^2 v_{j_2}^2 q
\\=&c_2q
\end{align*}

Next, we have
\begin{align*}
    \norm{A}=&\max \limits_{r} \norm{A_r}
    \\=&\max \limits_{r} \norm*{\begin{pmatrix} \eta_{1,r,1} u_1 \\ \eta_{1,r,2} u_2 \\ \vdots \\ \eta_{1,r,n} u_n \end{pmatrix}  }_2 \cdot \norm*{\begin{pmatrix} \eta_{2,r,1} v_1 \\ \eta_{2,r,2} v_2 \\ \vdots \\ \eta_{2,r,n} v_n \end{pmatrix}}_2
    \\ \le & \norm{u} \cdot \norm{v}
    \\=& 1
\end{align*}
almost surely.

Therefore, we have
\begin{align*}
    \norm{W_1^T A W_2}_{L_{2q}(\Pb)} \le c_1(\sqrt{2q} \sqrt{c_2q}+2q \cdot 1) \le c_3 q
\end{align*}

\end{proof}

We now proceed to bound the moment $\norm{\norm{(S_1 U)^T(S_2 U)u}_2}_{L_{2q}(\Pb)}$ or equivalently, \begin{align*}
    \norm{\ip{X_1^TX_2v,X_1^TX_2v}}_{L_{q}(\Pb)}
\end{align*} for $X_i = S_iU$ by formally describing the decoupling steps used to bound the latter.

\begin{lemma}\label{lem:x1x2unorm}
    Let $S$ be an $m \times n$ matrix distributed according to the fully independent unscaled OSNAP distribution with parameter $p$. Let $U$ be an arbitrary $n \times d$ deterministic matrix such that $U^TU=\operatorname{Id}$. Let $X=SU$. Let $X_1$ and $X_2$ be i.i.d. copies of $X$. Let $u \in \R^d$ be a vector with $\norm{u^T}_{l_2([d])}=1$. Let $q \ge 1$ be an integer. Then there exists a constant $c_{\ref{lem:x1x2unorm}}$ such that
\begin{align*}
    \norm{\norm{X_1^TX_2u}_2}_{L_{2q}(\Pb)} \le c_{\ref{lem:x1x2unorm}}\paren*{\sqrt{pmpd}+ (pmd)^\frac{1}{q} (\sqrt{pdq} + q)}
\end{align*}
\end{lemma}
\begin{proof}

Recall that for a matrix $S$ with the fully independent unscaled OSNAP distribution we have $S= \sum_{(l,\gamma) \in \Xi} Z_{(l,\gamma)}$, for independent summands $Z_{(l,\gamma)}$ as described in Definition \ref{def:osnap}. Then,
\begin{align*}
    \ip{X_1^TX_2v,X_1^TX_2v} &=\sum_{(l,\gamma) \in \Xi} \ip{X_1^TZ_{2,(l,\gamma)}Uv,X_1^TZ_{2,(l,\gamma)}Uv} \\
    & \qquad + \sum_{(l,\gamma)\in \Xi} \ip{U^T Z^T_{1,(l,\gamma)}X_2v,U^TZ^T_{1,(l,\gamma) }X_2v} \\
    &\qquad \qquad +\sum_{\substack{(l_1, \gamma_1) \neq (l_2, \gamma_2) \\ (l_3, \gamma_3) \neq (l_4, \gamma_4)}} \ip{U^TZ^T_{1,(l_3,\gamma_3) }Z_{2,(l_1,\gamma_1)}Uv,U^TZ^T_{1,(l_4,\gamma_4) }Z_{2,(l_2, \gamma_2)}Uv} \\
    &\qquad \qquad \qquad - \sum_{(l,\gamma) \in \Xi} \ip{U^TZ^T_{1,(l,\gamma) }Z_{2,(l,\gamma)}Uv,U^TZ^T_{1,(l,\gamma) }Z_{2,(l, \gamma)}Uv} \\
    &=: Y_1 + Y_2 + Y_3 - Y_4 \\
    &\le Y_1 + Y_2 + Y_3
\end{align*}
since $Y_4 = \sum_{(l,\gamma) \in \Xi} \norm{U^TZ^T_{1,(l,\gamma) }Z_{2,(l, \gamma)}Uv}_2^2 \ge 0$. We use this decomposition by noting that,
\begin{align*}
    \norm{\norm{X_1^TX_2u}_2}_{L_{2q}(\Pb)} &= \paren*{ \E \sqbr*{ \norm{X_1^TX_2u}_2^{2q} } }^{1/2q} \\ 
    &= \paren*{ \E \sqbr*{ \paren*{\ip{X_1^TX_2v,X_1^TX_2v}}^{q} } }^{1/2q}\\
    &\le \paren*{ \E \sqbr*{ \paren*{Y_1 + Y_2 + Y_3 }^{q} } }^{1/2q}\\
    &\le \paren*{\E \sqbr*{Y_1^q}}^{1/2q} + \paren*{\E \sqbr*{Y_2^q}}^{1/2q} + \paren*{\E \sqbr*{Y_3^q}}^{1/2q}
\end{align*}

First we consider the $Y_1$ term, 
\begin{align*}
    Y_1 &= \sum_{(l,\gamma) \in \Xi} v^TU^T Z^T_{2,(l,\gamma)}X_1 X_1^TZ_{2,(l,\gamma)}Uv \\
    &= \sum_{(l,\gamma) \in \Xi} v^TU^T e_l e_{\mu_{2,(l, \gamma)}}^T X_1 X_1^T e_{\mu_{2,(l, \gamma)}}e_l^TUv \\
    &= \sum_{(l,\gamma) \in \Xi} \norm{e^T_{\mu_{2,(l, \gamma)}}X_1}_2^2 \ip{u_l, v}^2 \\
\end{align*}
where $\{u_l\}_{l \in [n]}$ are the rows of $U$.

We condition on $X_1$ and evaluate $\E \sqbr*{ Y_1^q \vert X_1}$. After conditioning on $X_1$, the summands of $Y_1$ are independent, so we can use the following version of Rosenthal's inequality,

\begin{lemma}[Theorem 14.10 from \cite{boucheron2013concentration}]\label{lem:rosenthal}
Let $R=\sum_{j=1}^nZ_i$ where $R_1,...,R_n$ are independent and nonnegative random variables. Then the exists a constant $c_{\ref{lem:rosenthal}}$ such that, for all integers $q \ge 1$, we have
\begin{align*}
    (\E (R^q))^{1/q} \le 2 \E( R)+ c_{\ref{lem:rosenthal}}q  \paren{\sum_{j \in [n]} \E \sqbr{R_j^q} }^{1/q}
\end{align*}
Hence,
\begin{align*}
    (\E [R^q]) &\le 2^q \paren*{2^q (\E[R])^q+ c_{\ref{lem:rosenthal}}^qq^q  \paren{\sum_{j \in [n]} \E \sqbr{R_j^q} }}
\end{align*}
\end{lemma}

In our case,
\begin{align*}
    \E[Y_1 \vert X_1] &= \sum_{(l,\gamma) \in \Xi} \ip{u_l, v}^2 \E \sqbr*{\norm{e^T_{\mu_{2,(l, \gamma)}}X_1}_2^2}  \\
    &= \sum_{(l,\gamma) \in \Xi} \ip{u_l, v}^2 (s/m)\sum_{j \in [(m/s)(\gamma-1):(m/s)\gamma]} \norm{e^T_{j}X_1}_2^2  \\
    &= \sum_{l \in [n]} p \cdot \ip{u_l, v}^2 \norm{X_1}_F^2 \\
    &= p\norm{X_1}_F^2
\end{align*}
since $\sum_{l \in [n]} \ip{u_l, v}^2 = \norm{Uv}_2^2 = 1$.

Similarly,
\begin{align*}
    \sum_{(l,\gamma) \in \Xi} \ip{u_l, v}^{2q} \E \sqbr*{\norm{e^T_{\mu_{2,(l, \gamma)}}X_1}_2^{2q} \big\vert X_1 } &= \sum_{l \in [n], j \in [m]} p \ip{u_l, v}^{2q} \norm{e^T_{j}X_1}_2^{2q} \\
    &\le pm \norm{e^T_I X_1}^{2q}
\end{align*}
where $I$ is a integer valued random variable uniformly distributed in $[m]$ independent of all existing random variables (the reason for expressing the sum in this way is to later use an existing result to control the moments of this quantity after unconditioning on $X_1$)

By Lemma \ref{lem:rosenthal}, 
\begin{align*}
    \E \sqbr*{Y_1^q} &= \E \sqbr*{ \E \sqbr*{ Y_1^q \vert X_1 }} \\
    &\le \E \sqbr*{ 4^q (p)^q \norm{X_1}_F^{2q} + (2c_{\ref{lem:rosenthal}}q)^q pm \norm{e^T_I X_1}^{2q} } \\
    \text{and, } \E \sqbr*{Y_1^q}^{1/2q} &\le 2 \sqrt{p} \E \sqbr*{\norm{X_1}_F^{2q}}^\frac{1}{2q} + \sqrt{2c_{\ref{lem:rosenthal}}q} (pm)^{1/2q} \E \sqbr*{\norm{e_I^TX_1}^{2q}}^\frac{1}{2q} 
\end{align*}
Since the rows of $X_1$ are identically distributed, $\E \sqbr*{\norm{e_I^TX_1}^{2q}}^\frac{1}{2q} = \E \sqbr*{\norm{e_1^TX_1}^{2q}}^\frac{1}{2q}$.
By Lemma \ref{lem:xfrobnorm}, $\E \sqbr*{\norm{X_1}_F^{2q}}^\frac{1}{2q} \le c_{\ref{lem:xfrobnorm}.1}(\sqrt{pmd} + \sqrt{q})$,  and by Lemma \ref{lem:rowbounds}, $\E \sqbr*{\norm{e_1^TX_1}_2^{2q}}^\frac{1}{2q} \le c_{\ref{lem:rowbounds}.1} \paren{\sqrt{pd} + \sqrt{q}}$. Finally we get,
\begin{align*}
    \E \sqbr*{Y_1^q}^\frac{1}{2q} &\le C_1 \paren*{\sqrt{pmpd} + (pm)^{1/2q} (\sqrt{pdq} + q) }
\end{align*}
Here we have absorbed the $\sqrt{pq}$ term coming from $2 \sqrt{p} \E \sqbr*{\norm{X_1}_F^{2q}}^\frac{1}{2q}$ into the $\sqrt{pdq} + q$ term by adjusting $C_1$.

Moving on to the $Y_2$ term, 
\begin{align*}
    Y_2 &= \sum_{(l,\gamma)\in \Xi} \ip{U^T Z^T_{1,(l,\gamma)}X_2v,U^TZ^T_{1,(l,\gamma) }X_2v} \\
    &= \sum_{(l,\gamma)\in \Xi} v^T X_2^T  e_{\mu_{2,(l, \gamma)}} u_l^T u_l e_{\mu_{2,(l, \gamma)}}^T X_2 v \\
    &= \sum_{(l,\gamma) \in \Xi} \norm{u_l}^2 (X_2 v)^2_{\mu_{1, (l,\gamma)}} \\
\end{align*}

As in the previous case we condition on $X_2$ and evaluate $\E \sqbr{ Y_2^q \vert X_2 }$ using Rosenthal's inequality.
\begin{align*}
    \E[Y_2 \vert X_2] &= \sum_{(l,\gamma) \in \Xi} \norm{u_l}^2 \E \sqbr*{\paren{X_2 v}_{\mu_{1, (l,\gamma)}}^2}  \\
    &= \sum_{(l,\gamma) \in \Xi} \norm{u_l}^2 (s/m)\sum_{j \in [(m/s)(\gamma-1):(m/s)\gamma]} \paren{X_2 v}_{j}^2  \\
    &= \sum_{l \in [n]} p \cdot \norm{u_l}^2 \norm{X_2v}_2^2 \\
    &= pd \norm{X_2v}_2^2
\end{align*}
since $\sum_{l \in [n]} \norm{u_l}^2 = d$, and, 
\begin{align*}
    \sum_{(l,\gamma) \in \Xi} \norm{u_l}^{2q} \E \sqbr*{\paren{X_2 v}_{\mu_{1, (l,\gamma)}}^{2q} \vert X_2 } &= \sum_{l \in [n], j \in [m]} p \norm{u_l}^{2q} \paren{X_2 v}_{j}^{2q} \\
    &\le dpm \paren{X_2 v}_{I}^{2q}
\end{align*}
where $I$ is a integer valued random variable uniformly distributed in $[m]$ independent of all existing random variables.
By Lemma \ref{lem:rosenthal}, 
\begin{align*}
    \E \sqbr*{Y_2^q} &= \E \sqbr*{ \E \sqbr*{ Y_2^q \vert X_2 }} \\
    &\le \E \sqbr*{ 4^q (pd)^q \norm{X_2v}_2^{2q} + (2c_{\ref{lem:rosenthal}}q)^q pmd \paren{X_2v}_I^{2q} } \\
    \text{and, } \E \sqbr*{Y_2^q}^{1/2q} &\le 2 \sqrt{pd} \E \sqbr*{\norm{X_2v}_2^{2q}}^\frac{1}{2q} + \sqrt{2c_{\ref{lem:rosenthal}}q} (pmd)^{1/2q} \E \sqbr*{\paren{X_2v}_I^{2q}}^\frac{1}{2q} 
\end{align*}
Since the rows of $X_2$ are identically distributed, $\E \sqbr*{\paren{X_2v}_I^{2q}}^\frac{1}{2q} = \E \sqbr*{\paren{X_2v}_1^{2q}}^\frac{1}{2q}$. By Lemma \ref{lem:xfrobnorm}, $\E \sqbr*{\norm{X_2v}^{2q}}^\frac{1}{2q} \le c_{\ref{lem:xfrobnorm}.2}(\sqrt{pm} + \sqrt{q})$, and by Lemma \ref{lem:rowbounds}, $\E \sqbr*{\paren{X_2v}_I^{2q}}^\frac{1}{2q} \le 2 c_{\ref{lem:rowbounds}.2} \sqrt{q}$. Finally we get,
\begin{align*}
    \E \sqbr*{Y_2^q}^\frac{1}{2q} &\le C_2(\sqrt{pmpd} + \sqrt{pdq} + (pmd)^{1/2q} \cdot q )
\end{align*}

Now recall the $Y_3$ term,
\[ Y_3 = \sum_{\substack{(l_1, \gamma_1) \neq (l_2, \gamma_2) \\ (l_3, \gamma_3) \neq (l_4, \gamma_4)}} \ip{U^TZ^T_{1,(l_3,\gamma_3) }Z_{2,(l_1,\gamma_1)}Uv,U^TZ^T_{1,(l_4,\gamma_4) }Z_{2,(l_2, \gamma_2)}Uv}\]

We wish to look at $\E \sqbr*{Y_3^q}^\frac{1}{2q}$, and we evaluate it by first conditioning on $\mathcal{Z}_1 := \{ Z_{1, (l,\gamma} \}_{l \in [n], \gamma \in [pm]}$, i.e. as $\E \sqbr*{\E\sqbr*{Y_3 \vert \mathcal{Z}_1 }^q}^\frac{1}{2q}$ . 

For fixed $(l_1, \gamma_1), (l_2, \gamma_2)$, let \[ h_2(Z_{2,(l_1,\gamma_1)}, Z_{2,(l_2,\gamma_2)}) =  \sum_{ (l_3, \gamma_3) \neq (l_4, \gamma_4)} \ip{U^TZ^T_{1,(l_3,\gamma_3) }Z_{2,(l_1,\gamma_1)}Uv,U^TZ^T_{1,(l_4,\gamma_4) }Z_{2,(l_2, \gamma_2)}Uv} \] 
Clearly $h_2(\cdot, \cdot)$ is a deterministic multilinear function after conditioning on $\mathcal{Z}_1$. Then, 
\[ Y_3 = \sum_{ (l_1, \gamma_1) \neq (l_2, \gamma_2)} h_2(Z_{2,(l_1,\gamma_1)}, Z_{2,(l_2,\gamma_2)}) \]
By Lemma \ref{lem:masterdecoup}, 
\begin{align*}
    \E \sqbr{\abs{Y_3}^{q}} &\le \E \sqbr*{\E \sqbr*{\abs*{2 \sum_{ (l_1, \gamma_1), (l_2, \gamma_2) \in \Xi} h_2(Z'_{2,(l_1,\gamma_1)}, Z_{2,(l_2,\gamma_2)}) + h_2(Z_{2,(l_1,\gamma_1)}, Z'_{2,(l_2,\gamma_2)}) }^{q} \Bigg\vert \mathcal{Z}_1} } \\
    &\le \E \sqbr*{\abs*{2 \sum_{ (l_1, \gamma_1), (l_2, \gamma_2) \in \Xi} h_2(Z'_{2,(l_1,\gamma_1)}, Z_{2,(l_2,\gamma_2)}) + h_2(Z_{2,(l_1,\gamma_1)}, Z'_{2,(l_2,\gamma_2)}) }^{q} }  \\
\end{align*}

Now, 
\begin{align*}
    &\sum_{ (l_1, \gamma_1), (l_2, \gamma_2) \in \Xi} h_2(Z'_{2,(l_1,\gamma_1)}, Z_{2,(l_2,\gamma_2)})\\  =& \sum_{ (l_3, \gamma_3) \neq (l_4, \gamma_4)} \ip{U^TZ^T_{1,(l_3,\gamma_3) } \paren*{ \sum_{ (l_1, \gamma_1)\in \Xi} Z'_{2,(l_1,\gamma_1)}} Uv,U^TZ^T_{1,(l_4,\gamma_4) }\paren*{ \sum_{  (l_2, \gamma_2) \in \Xi}Z_{2,(l_2, \gamma_2)}}Uv} \\
    =& \sum_{ (l_3, \gamma_3) \neq (l_4, \gamma_4)} \ip{U^TZ^T_{1,(l_3,\gamma_3) } X_3  v,U^TZ^T_{1,(l_4,\gamma_4) } X_2 v}
\end{align*}
where $X_3$ is an independent copy of $X_2$. So we get, 
\begin{align*}
    \E \sqbr{\abs{Y_3}^{q}} &\le \E \sqbr*{ \abs*{2 \sum_{ (l_1, \gamma_1), (l_2, \gamma_2) \in \Xi} h_2(Z'_{2,(l_1,\gamma_1)}, Z_{2,(l_2,\gamma_2)}) + h_2(Z_{2,(l_1,\gamma_1)}, Z'_{2,(l_2,\gamma_2)}) }^{q}  } \\
    &= \E \sqbr*{\abs*{2 \sum_{ (l_3, \gamma_3) \neq (l_4, \gamma_4)} \ip{U^TZ^T_{1,(l_3,\gamma_3) } X_3  v,U^TZ^T_{1,(l_4,\gamma_4) } X_2 v} + \ip{U^TZ^T_{1,(l_3,\gamma_3) } X_2  v,U^TZ^T_{1,(l_4,\gamma_4) } X_3 v} }^{q}  } \\
    &= \E \sqbr*{\abs*{ \sum_{ (l_3, \gamma_3) \neq (l_4, \gamma_4)} h_1(Z_{1,(l_3,\gamma_3)}, Z_{1,(l_4,\gamma_4)}) }^{q} } \\
\end{align*}
where, 
\[ h_1(Z_{1,(l_3,\gamma_3)}, Z_{1,(l_4,\gamma_4)}) = 2\ip{U^TZ^T_{1,(l_3,\gamma_3) } X_3  v,U^TZ^T_{1,(l_4,\gamma_4) } X_2 v} + 2\ip{U^TZ^T_{1,(l_3,\gamma_3) } X_2  v,U^TZ^T_{1,(l_4,\gamma_4) } X_3 v} \] 
is a deterministic multilinear function in its inputs after conditioning on $X_2$ and $X_3$. Applying Lemma \ref{lem:masterdecoup} again and proceeding similarly, we get,
\begin{align*}
    \E \sqbr{\abs{Y_3}^{q}} &\le \E \sqbr*{\abs*{ 4 \paren*{ \ip{X_4^T X_3  v, X_1^T X_2 v} + \ip{X_1^T X_3  v, X_4^T X_2 v} + \ip{X_4^T X_2  v, X_1^T X_3 v} + \ip{X_1^T X_2  v, X_4^T X_3 v} } }^{q}} \\
\end{align*}
Thus,
\begin{align*}
    \E \sqbr{Y_3^{q}}^\frac{1}{2q} &\le 4 \E \sqbr{\ip{X_1^TX_2v,X_3^TX_4v}^{q}}^\frac{1}{2q} \\
    &=  4 \E \sqbr{ \E \sqbr{\ip{X_1^TX_2v,X_3^TX_4v}^{q} \vert X_1, X_2} }^\frac{1}{2q} \\
    &= 4 \E \sqbr{ \norm{X_1^TX_2v}^q}^\frac{1}{2q} \E \sqbr{\ip{u,X_3^TX_4v}^{q}}^\frac{1}{2q} \\
    &=  4 \E \sqbr{ \norm{X_1^TX_2v}^q}^\frac{1}{2q} \E \sqbr{\ip{X_3u,X_4v}^{q}}^\frac{1}{2q} \\
    &\le C_3 \sqrt{q} \E \sqbr{ \norm{X_1^TX_2v}^q}^\frac{1}{2q}
\end{align*}
where $u \in \R^d$ is some fixed unit vector and the last inequality follows from Lemma \ref{lem:scalarnorm}.

Putting the estimates for the $Y_1, Y_2$ and $Y_3$ terms together, we get,
\begin{align*}
    \norm{\norm{X_1^TX_2u}_2}_{L_{2q}(\Pb)} &\le C_1(\sqrt{pmpd} + (pm)^{1/2q} (\sqrt{pdq} + q) ) \\
    &\qquad +  C_2(\sqrt{pmpd} + \sqrt{pdq} + (pmd)^{1/2q} \cdot q ) \\
    &\qquad \qquad + C_3 \sqrt{q} \E \sqbr{ \norm{X_1^TX_2v}_2^q}^\frac{1}{2q} \\
    &\le C_4(\sqrt{pmpd} + (pmd)^\frac{1}{2q} (\sqrt{pdq} + q) ) + C_3 \sqrt{q} \E \sqbr{ \norm{X_1^TX_2v}_2^q}^\frac{1}{2q}
\end{align*}

Let $C_5>0$ be a constant which satisfies $C_4 + C_3\sqrt{C_5} < C_5$ and after fixing such a $C_5$ let $C_6>0$ be chosen so that it satisfies $C_4 + C_3\sqrt{C_5} + C_3\sqrt{C_6} \le C_6$. 

Assume that $\E \sqbr{ \norm{X_1^TX_2v}_2^q}^\frac{1}{q} \le C_5 \sqrt{pmpd} + C_6(pmd)^{1/q}(\sqrt{pdq} + q)$. This is true when $q=2$ for $C_5 \ge 1$ since $\E \sqbr{X_2^T X_1 X_1^T X_2} = pmpd I_d$. Then,
\begin{align*}
    \norm{\norm{X_1^TX_2u}_2}_{L_{2q}(\Pb)} &\le C_4(\sqrt{pmpd} + (pmd)^\frac{1}{2q} (\sqrt{pdq} + q) ) \\
    & \qquad +  C_3 \sqrt{q}\sqrt{C_5} \paren{pmpd}^{1/4} + C_3 \sqrt{q} \sqrt{C_6}(pmd)^{1/2q}(\sqrt{pdq} + q)^\frac{1}{2} \\
\end{align*}
Note that $q^\frac{1}{2}(\sqrt{pdq} + q)^\frac{1}{2} \le (\sqrt{pdq} + q)$. So, 
\begin{align*}
    \norm{\norm{X_1^TX_2u}_2}_{L_{2q}(\Pb)} &\le C_4\sqrt{pmpd} +  C_3 \sqrt{q}\sqrt{C_5} \paren{pmpd}^{1/4}  \\
    & \qquad  + (C_3 \sqrt{C_6} + C_4)(pmd)^{1/2q}(\sqrt{pdq} + q) \\
\end{align*}
If $q\le \sqrt{pmpd} $, then $ C_4\sqrt{pmpd} +  C_3 \sqrt{q}\sqrt{C_5} \paren{pmpd}^{1/4} \le (C_4 + C_3\sqrt{C_5})\sqrt{pmpd} \le C_5\sqrt{pmpd}$ by our choice of $C_5$ and $(C_3 \sqrt{C_6} + C_4) \le C_6$ by our choice of $C_6$, so,
\begin{align*}
    \norm{\norm{X_1^TX_2u}_2}_{L_{2q}(\Pb)} &\le C_5\sqrt{pmpd}  + C_6 (pmd)^{1/2q}(\sqrt{2pdq} + 2q) \\
\end{align*}
If $q>\sqrt{pmpd}$, then $C_3 \sqrt{q}\sqrt{C_5} \paren{pmpd}^{1/4} < C_3\sqrt{C_5}(pmd)^\frac{1}{2q}(\sqrt{pdq} + q)$. In this case, 
\begin{align*}
    \norm{\norm{X_1^TX_2u}_2}_{L_{2q}(\Pb)} &\le C_4\sqrt{pmpd} \\ 
    & \qquad + (C_3 \sqrt{C_6} + C_4 + C_3\sqrt{C_5})(pmd)^{1/2q}(\sqrt{pdq} + q) \\
    &\le C_5\sqrt{pmpd}  + C_6 (pmd)^{1/2q}(\sqrt{2pdq} + 2q) \\
\end{align*}
by our choice of $C_5$ and $C_6$. 

\end{proof}

\printbibliography
\appendix
\section*{Appendix}
\section{Decoupling.} \label{sec:decoup}

In this section, we provide formal proofs of the decoupling steps mentioned in Section \ref{sec:overview}. We start with a general decoupling lemma that will later be specialized to the specific cases required.

\subsection{General Decoupling Lemma.}

\begin{lemma}[General Decoupling]\label{lem:masterdecoup}
    For a natural number $n$, let $ \mathcal{Z} := \{ Z_i \}_{i \in [n]}$ be a collection of zero mean independent random variables taking values in some Euclidean space $\R^N$. Let $ \mathcal{Z}' := \{ Z'_i \}_{i \in [n]}$ be an independent copy of $\mathcal{Z}$. Let $h:\R^N \times \R^N \to \R^M$ be a multilinear function, and let $\norm{ \cdot}$ be any norm on $\R^M$ such that $\E \sqbr*{\norm*{h(Z_i, Z_j)}} < \infty$ for $i \neq j$. Let $\Phi: [0, \infty) \to [0, \infty)$ be a convex nondecreasing function such that $\E \sqbr*{\Phi\paren*{\norm*{ h(Z_i, Z_j)}}} < \infty$ for $i \neq j$. Then, 
    \[ \E \sqbr*{ \Phi\paren*{\norm*{ \sum_{i,j \in [n] , i \neq j} h(Z_i, Z_j)}} } \le  \E \sqbr*{\Phi\paren*{\norm*{ 2 \sum_{i,j \in [n]} h(Z_i, Z'_j) + h(Z'_i, Z_j)}}} \]
\end{lemma}

\begin{proof}
    Let $\mathbf{w}^T = (w_1, \ldots, w_n)$ be a vector of independent random variables such that $\Pb(w_i=1)=\Pb(w_i=0)=1/2$ and $\mathbf{w}$ is independent from $\{ Z_i \}_{i \in [n]}$  and $\{ Z'_i \}_{i \in [n]}$. Then, whenever $i\neq j$, $\E[\mathbf{1}_{w_i \neq w_{j}}] = 1/2$. So we have, 
    \begin{align*}
        \sum_{i,j =1, i \neq j}^n h(Z_i, Z_j) = 2\E_{\mathbf{w}} \left[  \sum_{i,j =1}^n \mathbf{1}_{w_i \neq w_{j}} h(Z_i, Z_j) \right]
    \end{align*}
    By Jensen's inequality \cite[Lemma 4.5, p.86]{kallenberg2021foundations}
    \begin{align*}
        \E \left[ \Phi\paren*{\norm*{ \sum_{i,j \in [n] , i \neq j} h(Z_i, Z_j)}} \right] &= \E_{\cZ} \left[ \Phi\paren*{ \norm*{ 2\E_{\mathbf{w}} \left[  \sum_{i,j =1}^n \mathbf{1}_{w_i \neq w_{j}} h(Z_i, Z_j) \right] }} \right] \\
        &\le \E_{\cZ} \E_{\mathbf{w}} \left[ \Phi\paren*{ \norm*{  \sum_{i,j =1}^n 2 \cdot\mathbf{1}_{w_i \neq w_{j}} h(Z_i, Z_j)  }} \right] \\
        &\le  \E_{\mathbf{w}} \E_{\cZ} \left[ \Phi\paren*{ \norm*{  \sum_{i,j =1}^n 2 \cdot\mathbf{1}_{w_i \neq w_{j}} h(Z_i, Z_j)  }} \right] \\
    \end{align*}
Fix $\mathbf{w}$, and let $J = \{ j \in [n] | w_j=1 \}$. Then,
\begin{align*}
    &\E_{\cZ} \left[ \Phi\paren*{ \norm*{  \sum_{i,j =1}^n 2 \cdot\mathbf{1}_{w_i \neq w_{j}} h(Z_i, Z_j)  }} \right]  \\ 
    = &\E_{\cZ} \left[ \Phi\paren*{ \norm*{  \sum_{i \in J, j \in [n]\backslash J } 2 h(Z_i, Z_j) +  \sum_{ i \in [n]\backslash J, j \in J } 2 h(Z_i, Z_j) }} \right] \\
    = &\E_{\cZ} \Biggl[ \Phi\paren*{ \norm*{  \sum_{i \in J, j \in J^C } 2 \paren*{h(Z_i, Z_j)+h(Z_j,Z_i)} }}\Biggr] \\
\end{align*}
The key observation for decoupling is that the collection of random variables $ \mathcal{Z}_J := \{ Z_j \}_{ j\in J}$ is independent of the collection $\mathcal{Z}_{J^C} := \{ Z_j \}_{j\in J^C}$
, so the above expectation does not change if the collection $ \mathcal{Z}_J$ is replaced by $\mathcal{Z}'_J :=\{Z' \}_{ j\in J}$ thought of as a subset of the collection $\cZ'$, i.e.,
\begin{align*}
    &\E_{\mathcal{Z}_J, \mathcal{Z}_{J^C} } \Biggl[ \Phi\paren*{ \norm*{  \sum_{(i,j) \in J \times J^C } 2 (h(Z_i, Z_j)+h(Z_j,Z_i)) }}\Biggr] \\
    =&\E_{\mathcal{Z'}_J, \mathcal{Z}_{J^C} } \Biggl[ \Phi\paren*{ \norm*{  \sum_{(i,j) \in J \times J^C } 2 (h(Z_i', Z_j)+h(Z_j,Z_i')) }}\Biggr] \\
    =&\E_{\mathcal{Z'}_J, \mathcal{Z}_{J^C} } \Biggl[ \Phi\Biggl( \Biggl\|  \sum_{(i,j) \in J \times J^C } 2 (h(Z_i', Z_j)+h(Z_j,Z_i')) \\
    &\quad+ \E_{\mathcal{Z}_J, \mathcal{Z'}_{J^C} }\biggl[\sum_{(i,j) \notin J \times J^C} 2 (h(Z_i', Z_j)+h(Z_j,Z_i'))\biggr]\Biggr\| \Biggr) \Biggr] \\
    \le & \E_{\mathcal{Z'}_J, \mathcal{Z}_{J^C} } \E_{\mathcal{Z}_J, \mathcal{Z'}_{J^C} }\Biggl[ \Phi\Biggl( \Biggl\|  \sum_{(i,j) \in J \times J^C } 2 (h(Z_i', Z_j)+h(Z_j,Z_i')) \\
    &\quad+ \sum_{(i,j) \notin J \times J^C} 2 (h(Z_i', Z_j)+h(Z_j,Z_i'))\Biggr\| \Biggr) \Biggr] \\
    =&\E_{\mathcal{Z'}, \mathcal{Z} } \Biggl[ \Phi\Biggl( \Biggl\|  \sum_{i \in [n], j \in [n] } 2 (h(Z_i', Z_j)+h(Z_j,Z_i')) \Biggr\| \Biggr) \Biggr]
\end{align*}
where in the third equality the expectations we add are all $0$ by multilinearity of $h$ and the zero mean property of $Z_i$, and the inequality step follows by taking the expectations outside using Jensen's inequality. 
\end{proof}

\subsection{Decoupling for General Random Matrices.}

Next, we use Lemma \ref{lem:masterdecoup} to bound $\norm{X^TX}_{2q}$ when $X$ is a sum of independent random matrices.

\begin{lemma}[Decoupling]\label{lem:decgen}
Let $X=\sum \limits_{l=1}^L Z_l$ where $Z_1,...,Z_L$ are independent $m \times d$ random matrices.
 Let $X_1$ and $X_2$ be independent copies of $X$. Let $q$ be a positive integer. Then we have
 \begin{align*}
     \norm{X^TX}_{2q} \le \norm{\sum \limits_{l \in [L]}(Z_l)^TZ_l}_{2q}+2\norm{X_1^TX_2+X_2^TX_1}_{2q}
 \end{align*}
\end{lemma}

\begin{proof}
By definition, for any symmetric matrix $M$, we have
\begin{align*}
     \norm{M}_{2q}=&\big(\E(\tr |M|^{2q})\big)^{\frac{1}{2q}}
     \\=&\big(\E(\tr (\sqrt{(M)^2})^{2q})\big)^{\frac{1}{2q}}
     \\=&\big(\E(\tr ((M)^2)^{q})\big)^{\frac{1}{2q}}
     \\=&\big(\E(\tr ((M)^{2q})\big)^{\frac{1}{2q}}
 \end{align*}

Therefore, since $X^TX$, $\sum \limits_{l \in [L]}(Z_l)^TZ_l$, and $X_1^TX_2+X_2^TX_1$ are symmetric, we have $\norm{X^TX}_{2q}=\big(\E(\tr ((X^TX)^{2q})\big)^{\frac{1}{2q}}$, $\norm{\sum \limits_{l \in [L]}(Z_l)^TZ_l}_{2q}=\big(\E(\tr ((\sum \limits_{l \in [L]}(Z_l)^TZ_l)^{2q})\big)^{\frac{1}{2q}}$, and $\norm{X_1^TX_2+X_2^TX_1}_{2q}=\big(\E(\tr ((X_1^TX_2+X_2^TX_1)^{2q})\big)^{\frac{1}{2q}}$.

Since $X=\sum \limits_{l \in [L]}Z_l$. Then we have
\begin{align*}
    X^TX=&(\sum \limits_{l \in [L]}(Z_l)^T)(\sum \limits_{l \in [L]}Z_l)
    \\=&\sum \limits_{l \in [L]}(Z_l)^TZ_l + \sum \limits_{l_1 \ne l_2}(Z_{l_1})^TZ_{l_2}
\end{align*}

By triangle inequality for norms, we have
\begin{align*}
    \norm{X^TX}_{2q} \le \norm{\sum \limits_{l \in [L]}(Z_l)^TZ_l}_{2q}+ \norm{\sum \limits_{l_1 \ne l_2}(Z_{l_1})^TZ_{l_2}}_{2q}
\end{align*}

Now, $\norm*{\sum \limits_{l_1 \ne l_2}(Z_{l_1})^TZ_{l_2}}_{2q} = \E \sqbr*{ \norm*{\sum \limits_{l_1 \ne l_2}(Z_{l_1})^TZ_{l_2}}_{S_{2q}}^{2q} }^\frac{1}{2q}$, where $\norm{\cdot}_{S_{2q}}$ denotes the normalised Schatten norm (see Section \ref{subsec:notation}).

By Lemma \ref{lem:masterdecoup}, 
\begin{align*}
    \E \sqbr*{ \norm*{\sum \limits_{l_1 \ne l_2}(Z_{l_1})^TZ_{l_2}}_{S_{2q}}^{2q} } &\le \E \sqbr*{ \norm*{\sum \limits_{l_1, l_2 \in [L]} 2(Z_{l_1})^TZ'_{l_2} + 2(Z'_{l_2})^TZ_{l_1} }_{S_{2q} }^{2q} } \\
    &\le  \E \sqbr*{ \norm*{ 2\paren*{\sum \limits_{l \in [L]}Z_l}^T \paren*{\sum \limits_{l \in [L]}Z'_l} + 2\paren*{\sum \limits_{l \in [L]}Z'_l}^T \paren*{\sum \limits_{l \in [L]}Z_l} }_{S_{2q}}^{2q} } \\
    &\le  \E \sqbr*{ \norm*{ 2X_1^TX_2+2X_2^TX_1}^{2q}_{S_{2q}} } \\
\end{align*}
Taking the $2q$-th root gives, $\norm{\sum \limits_{l_1 \ne l_2}(Z_{l_1})^TZ_{l_2}}_{2q} \le 2\norm{X_1^TX_2+X_2^TX_1}_{2q}$.

\end{proof}

\begin{lemma}\label{decgeneral}
    Let $X=\sum \limits_{l=1}^L Z_l$ where $Z_1,...,Z_L$ are independent $m \times d$ random matrices.
 Let $X_1$ and $X_2$ be independent copies of $X$. Let $V_1,V_2$ be a $d \times m$ (deterministic) matrices, and $V_3$ be a $d \times d$ (deterministic) matrix. Let $q \ge 1$ be an integer. Then we have
    \begin{align*}
        \norm{V_1^TX^TV_3XV_2}_{2q} \le \norm{\sum \limits_{l \in [L]}V_1^TZ_{l}^TV_3Z_{l}V_2}_{2q}+4\norm{V_1^TX_1^TV_3X_2V_2}_{2q}
    \end{align*}
    where $X_1$ and $X_2$ are independent copies of $X$.
\end{lemma}

\begin{proof}
Since $X=\sum \limits_{l \in [L]}Z_l$. Then we have
\begin{align*}
    V_1^TX^TV_3XV_2=&V_1^T(\sum \limits_{l \in [L]}(Z_l)^T)V_3(\sum \limits_{l \in [L]}Z_l)V_2
    \\=&(\sum \limits_{l \in [L]}V_1^T(Z_l)^TV_3)(\sum \limits_{l \in [L]}Z_lV_2)
    \\=&\sum \limits_{l \in [L]}V_1^T(Z_l)^TV_3Z_lV_2 + \sum \limits_{l_1 \ne l_2}V_1^T(Z_{l_1})^TV_3Z_{l_2}V_2
\end{align*}

By triangle inequality for norms, we have
\begin{align*}
    \norm{V_1^TX^TV_3XV_2}_{2q} \le \norm{\sum \limits_{l \in [L]}V_1^T(Z_l)^TV_3Z_lV_2}_{2q}+ \norm{\sum \limits_{l_1 \ne l_2}V_1^T(Z_{l_1})^TV_3Z_{l_2}V_2}_{2q}
\end{align*}

Applying Lemma \ref{lem:masterdecoup} as in the proof of Lemma \ref{lem:decgen}, we see that 
\begin{align*}
\norm{\sum \limits_{l_1 \ne l_2}V_1^T(Z_{l_1})^TV_3Z_{l_2}V_2}_{2q} &\le 2\norm{V_1^TX_1^TV_3X_2V_2 + V_1^TX_2^TV_3X_1V_2}_{2q} \\
&\le 2\norm{V_1^TX_1^TV_3X_2V_2}_{2q} + 2\norm{V_1^TX_2^TV_3X_1V_2}_{2q} \\
&\le 4\norm{V_1^TX_1^TV_3X_2V_2}_{2q}
\end{align*}

\[ \]

\end{proof}

\subsection{Decoupling for OSNAP.}

In the case where $S$ has the fully independent unscaled OSNAP distribution as described in Definition \ref{def:osnap}, we have, 
\begin{align*}
    S &= \sum_{l=1}^n \sum_{\gamma=1}^{pm} \xi_{l,\gamma} e_{\mu_{(l, \gamma)}} e_l ^T \\
    &=: \sum_{l=1}^n \sum_{\gamma=1}^{pm} Z_{l,\gamma} 
\end{align*}
where $\{ \xi_{l,\gamma} \}_{l \in [n], \gamma \in [s]}$ is a collection of  independent random variables with $\Pb(\xi_{l,\gamma}=1)=\Pb(\xi_{l,\gamma}=-1)=1/2$, $\{ \mu_{l,\gamma} \}_{l \in [n], \gamma \in [s]}$ is a collection of  independent random variables such that each $\mu_{l,\gamma}$ is uniformly distributed in $[(m/s)(\gamma-1)+1:(m/s)\gamma]$ and $e_{\mu_{(l, \gamma)}}$ and $e_l$ represent basis vectors in $\R^m$ and $\R^n$ respectively. 

Recalling that our goal is to look at the moments of $(SU)^T(SU) - pm I_d$, we observe that,
\begin{align*}
    U^TS^TSU - pm\cdot I_d &= U^T \paren*{ \sum_{l=1}^n \sum_{\gamma=1}^{pm} Z_{l,\gamma} }^T \paren*{ \sum_{l=1}^n \sum_{\gamma=1}^{pm} Z_{l,\gamma} } U - pm\cdot I_d
\end{align*}
Observe that for $\gamma, \gamma' \in [pm]$ with $\gamma \neq \gamma'$, $e_{\mu_{(l, \gamma)}}$ and $e_{\mu_{(l, \gamma')}}$ have disjoint supports, which means $e_{\mu_{(l, \gamma)}}^T e_{\mu_{(l, \gamma')}} = 0$. Thus, we also have $Z_{l,\gamma}^T Z_{l',\gamma'} = 0$ for any $l, l' \in [n]$. So we get,
\begin{align*}
    U^TS^TSU - pm\cdot I_d &= \sum_{l, l'=1}^n \sum_{\gamma=1}^{pm} U^T Z_{l',\gamma} ^T Z_{l,\gamma}  U - pm\cdot I_d
\end{align*}
Separating the cases where $l=l'$ and $l \neq l'$,
\begin{equation}\label{eq:diagoffdiag}
    \begin{aligned}
    U^TS^TSU - pm\cdot I_d &= \sum_{l=1}^n \sum_{\gamma=1}^{pm} U^T Z_{l,\gamma} ^T Z_{l,\gamma}  U - pm\cdot I_d + \sum_{\substack{l, l'=1 \\ l\neq l'}}^n \sum_{\gamma=1}^{pm} U^T Z_{l',\gamma} ^T Z_{l,\gamma}  U 
\end{aligned}
\end{equation}

Notice that $Z_{l,\gamma} ^T Z_{l,\gamma} = e_l e_l^T$, so,
\begin{align*}
    \sum_{l=1}^n \sum_{\gamma=1}^{pm} U^T Z_{l,\gamma} ^T Z_{l,\gamma}  U - pm\cdot I_d &= \sum_{l=1}^n \sum_{\gamma=1}^{pm} U^T e_l e_l^T  U - pm\cdot I_d \\
    &= \sum_{\gamma=1}^{pm} U^T \paren*{ \sum_{l=1}^n e_l e_l^T}  U - pm\cdot I_d \\
    &= \sum_{\gamma=1}^{pm} I_d - pm\cdot I_d = 0\\
\end{align*}
where we used the fact that $\sum_{l=1}^n e_l e_l^T = I_n$ and $U^T I_n U = U^TU = I_d$.

To analyze the off-diagonal term, we use the decoupling results developed previously, 

\begin{lemma}[Decoupling] \label{lem:decoup}
When $S$ has the fully independent unscaled OSNAP distribution, we have
\begin{align*}
    \E [ \tr (U^TS^TSU - pm\cdot I_d)^{2q} ] &= \E \left[ \tr \left( \sum_{\substack{l, l'=1 \\ l\neq l'}}^n \sum_{\gamma=1}^{pm} U^T Z_{l',\gamma} ^T Z_{l,\gamma}  U  \right)^{2q} \right] \\
\end{align*}
Consequently, we have
\begin{align*}
    \E [ \tr (U^TS^TSU - pm\cdot I_d)^{2q} ] &\le \E_{S,S'} \left[ \tr \left(  2\paren*{(S'U)^TSU + (SU)^TS'U} \right)^{2q} \right]
\end{align*}
where $S'$ is an independent copy of $S$.
\end{lemma}
\begin{proof}

Consider the independent tuples 
\begin{align*}
   \{ (Z_{1,1}, Z_{1,2} \etc Z_{1, pm} ) \etc (Z_{n,1}, Z_{n,2} \etc Z_{n, pm} )\} = \{ \widetilde{Z}_1 \etc \widetilde{Z}_n \}
\end{align*}
Then, 
\[ \sum_{\substack{l, l'=1 \\ l\neq l'}}^n \sum_{\gamma=1}^{pm} U^T Z_{l',\gamma} ^T Z_{l,\gamma}  U = \sum_{\substack{l, l'=1 \\ l\neq l'}}^n h(\widetilde{Z}_{l'}, \widetilde{Z}_{l}) \]
for a multilinear function $h$. Using Lemma \ref{lem:masterdecoup},
\[ \E \left[ \tr \left( \sum_{\substack{l, l'=1 \\ l\neq l'}}^n \sum_{\gamma=1}^{pm} U^T Z_{l',\gamma} ^T Z_{l,\gamma}  U  \right)^{2q} \right] \le \E \left[ \tr \left( \sum_{l, l' = 1}^n 2h(\widetilde{Z}_{l'}, \widetilde{Z}_{l}) + 2h(\widetilde{Z}_{l}, \widetilde{Z}_{l'})  \right)^{2q} \right] \]

Observing that, 
\begin{align*}
    \sum_{l, l' = 1}^n h(\widetilde{Z}_{l'}, \widetilde{Z}_{l}) &=  \sum_{l, l'=1 }^n \sum_{\gamma=1}^{pm} U^T Z_{l',\gamma} ^T Z_{l,\gamma}  U \\
    &= \paren*{ \sum_{l'=1 }^n \sum_{\gamma'=1}^{pm} U^T Z_{l',\gamma'} ^T } \paren*{ \sum_{l=1}^n \sum_{\gamma=1}^{pm} Z_{l,\gamma} U } \\
    &= (S'U)^T(SU)
\end{align*}
since $Z_{l,\gamma}^T Z_{l',\gamma'} = 0$ whenever $\gamma \neq \gamma'$ concludes the proof.
\end{proof}

\end{document}